\documentclass[12pt]{amsart}

\usepackage{graphicx}
\usepackage{pstricks}
\usepackage{epsfig}

\newtheorem{thm}{Theorem}[section]

\newtheorem{defn}[thm]{Definition}
\newtheorem{lemma}[thm]{Lemma}

\newtheorem{cor}[thm]{Corollary}
\newtheorem{remark}[thm]{Remark}
\newtheorem{example}[thm]{Example}

\usepackage{amsmath}
\usepackage{amsxtra}
\usepackage{amscd}
\usepackage{amsthm}
\usepackage{amsfonts}
\usepackage{amssymb}
\usepackage{eucal}

\newcommand{\Z}{{\mathbb Z}}

\newcommand{\bt}{{\mathbf t}}

\newcommand{\bu}{{\mathbf u}}
\newcommand{\bk}{{\mathbf k}}

\numberwithin{equation}{section}

\begin{document}
\title{$T$-systems, networks and dimers}
\author{Philippe Di Francesco} 
\address{Department of Mathematics, University of Illinois MC-382, Urbana, IL 61821, U.S.A. e-mail: philippe@illinois.edu}
%\address{
%Institut de Physique Th\'eorique du Commissariat \`a l'Energie Atomique, 
%Unit\'e de Recherche associ\'ee du CNRS,
%CEA Saclay/IPhT/Bat 774, F-91191 Gif sur Yvette Cedex, 
%FRANCE. e-mail: philippe.di-francesco@cea.fr}
%\author{Rinat Kedem}
%\address{Department of Mathematics, University of Illinois MC-382, Urbana, IL 61821, U.S.A. e-mail: rinat@illinois.edu}
\date{\today}
\begin{abstract}
We study the solutions of the T-system for type A, also known as the octahedron equation, 
viewed as a 2+1-dimensional discrete evolution equation. These may be expressed entirely in terms
of the  stepped surface over which the initial data are specified, via a suitably defined flat $GL_n$ connection
which embodies the integrability of this infinite rank system. 
By interpreting the connection as the transfer operator for a directed graph or network with weighted edges, 
we show that the solution at a given point is expressed as the partition function
for dimers on a bipartite graph dual to the "shadow" of the point onto the initial data stepped surface. We 
extend the result to the case of other geometries such as that of the evaporation of a cube corner crystal, and to
a reformulation of the Kenyon-Pemantle discrete hexahedron equation.
\end{abstract}

\maketitle
\tableofcontents
%\section{Introduction}
%\section{Definitions}

%AMS Subject Classification (2000): Primary 05A19; Secondary 82B20
%\draft
%\Date{09/2008}
%
%%%%%%%%%%%%%%%%%%%%%%%%%%%%%%%%%%%%%%%%%%%%%%%%%%%%%%%%%%%%%%%%%%%%%
%

\section{Introduction}

The so-called $T$-systems are $2+1$-dimensional
discrete integrable systems of evolution equations in a discrete time variable $k\in \Z$.
They were introduced in the context of integrable quantum spin chains, as a system of 
equations satisfied by the eigenvalues of transfer matrices of generalized Heisenberg magnets, 
with the symmetry of a given Lie algebra \cite{KNS}. In the case of type A, the $T$-system equation is also 
known as the octahedron recurrence, and appears to be central
in a number of combinatorial objects, such as the lambda-determinant and the Alternating Sign Matrices
\cite{RR}\cite{DFLD}, the puzzles for computing Littlewood-Richardson coefficients  
\cite{KT}, generalizations of Coxeter-Conway frieze patterns
\cite{Cox}\cite{FRISES}\cite{BER},
and the domino tilings of the Aztec diamond \cite{EKLP}\cite{SPY}. The latter is a particular case of a dimer model, whose
configurations consist of matchings of the edges of a given bipartite graph.
Dimer models  were the subject of a lot of attention, culminating in the global understanding of the arctic curve
phenomenon in the continuum limit \cite{KO} \cite{KOS}, where the phase diagram of the model was shown to exhibit separations
between frozen, disordered, and liquid phases. 

A new interpretation for the $T$-system arose from realizing that the corresponding discrete
evolution could be viewed as a particular mutation in a suitably defined cluster algebra \cite{DFK08}. 
As such, it must satisfy the Laurent property\cite{FZI}, namely that any solution is a Laurent polynomial
of any set of admissible initial data. Moreover, the general positivity conjecture for cluster
algebras would also imply that these Laurent polynomials have non-negative integer coefficients.
The $T$-systems of A type were  explicitly solved for arbitrary admissible initial data and various boundary conditions
in terms of weighted path models on specific networks, coded by the geometry of initial data \cite{DF} \cite{DFK13}.

The aim of this note is to extend the T-system/dimer correspondence initiated in \cite{SPY} to arbitrary initial data, in a 
spirit similar to that of Ref. \cite{GK}.
To this end, we use a transfer matrix formulation of the network solutions, giving rise to a natural flat $GL_2$ connection
on the space of admissible initial data, and show that the connection may be interpreted as a local transfer matrix 
for the dimer model. Our main result is Theorem \ref{gendimth}, which expresses the solution of the T-system for arbitrary 
initial data as the partition function for dimers on a suitably defined bipartite graph.

We then turn to a different geometrical setting, in which the T-system describes the evolution of the corner of a 3D cubic crystal
under the evaporation/deposition of unit cubes, and show that the same tools give access to the solution in terms of 
arbitrary evaporated configurations (Theorem \ref{thetaexact} and Corollary \ref{arbitcub}). 
This solution is then reinterpreted as a partition function for dimers
on a graph determined by the evaporated configuration (Theorem \ref{dimcub}).

Finally, we show that the $GL_2$ connection underlying the solution obeys a generalized form of the Yang-Baxter equation,
and may be used to reformulate the hexaedron relations (Lemma \ref{hexayb}), a system of recursion relations recently 
introduced by Kenyon and Pemantle \cite{KEN}.

The paper is organized as follows. For pedagogical reasons, we devote Section 2 entirely to the case of the $A_1$ 
$T$-system. This is a $1+1$-dimensional reduction of the general $T$-system, with which all the concepts
and correspondences of this paper can be illustrated: connection to cluster algebra,
exact solution via a $GL_2$ flat connection, paths on a network, and finally
domino tiling/dimer partition function.  In Section 3 this is generalized to the full $T$-system, which describes a true $2+1$
dimensional evolution on the vertices of the Centered Cubic lattice, made of elementary octahedra.
After recalling the general solution of the $T$-system, we establish its equivalence to a dimer model, 
for arbitrary initial data. The main ingredient is the 
construction of a flat connection over the space of solutions (Section 3.2),
using $GL_2$ building blocks. These are interpreted first as network chips (Section 3.3) to be concatenated
to form an oriented weighted graph, such that the $T$-system solution is the partition function of certain families of 
non-intersecting paths on this graph. Paths are then bijectively mapped onto dimer configurations on a dual bipartite
graph in Sections 3.4 and 3.5. Section 4 explores the connection between our $GL_2$ connection and the Yang-Baxter
equation of integrable statistical mechanics (Section 4.1). We use the same $GL_2$ connection to solve the $T$-system 
on the cubic lattice in Section 4.2, and develop its network formulation (Section 4.3) and dimer formulation (Section 4.4).
We gather a few concluding remarks in Section 5, where we show that our connection allows for building a staggered
solution of the Yang-Baxter relation analogous to that of spin ladders or dimerized spin chains, which eventually 
describes the hexahedron recurrence of Kenyon and Pemantle.

\medskip
\noindent{\bf Acknowledgments.} We would like to thank  M. Gekhtman, R. Kedem, R. Kenyon, G. Musiker, 
D. Speyer, N. Reshetikhin for discussions, and R. Soto Garrido for a careful reading of the manuscript.
We acknowledge support  by the CNRS PICS program INTCOMB.
%RK is supported by NSF grant DMS-0802511. 
We would like to thank the Mathematical Science Research Institute in Berkeley, CA 
and the organizers of the semester ``Cluster Algebras" (Fall 2012) for hospitality during the early stages of 
this work, and the Simons Center for Geometry and Physics and the organizers of the semester
``Conformal Geometry" (Spring 2013) for hospitality.

\section{$T$-system and Dimers: the $A_1$ case}

\subsection{Definitions, initial data, and cluster algebra connection}

The $A_1$ $T$-system is the following system of non-linear recursion relations
\begin{equation}\label{aonetsys}
T_{j,k+1}T_{j,k-1}=T_{j+1,k}T_{j-1,k}+1
\end{equation}
for some indeterminate $T_{j,k}$ say with $j,k\in \Z$ and $j+k=0$ mod 2, invertible elements of
an algebra $\mathcal A$ with unit $1$ (assumed to be commutative throughout this paper).
This system can be considered as a three-term recursion relation in $k$, 
interpreted as a discrete time.
As such it has the following sets of admissible initial data. 

We denote by $\bk$ the infinite path with vertices $(j,k_j)$, $j\in \Z$, where 
$k_j\in \Z$, and $|k_{j+1}-k_j|=1$ for all $j\in \Z$. We may think of such a path 
as connecting neighboring vertices via up (resp. down) steps of the form
$(1,1)$ (resp. $(1,-1)$).
Let also $\bt$ denote an
infinite sequence $(t_j)_{j\in \Z}$ of invertible elements in $\mathcal A$.
For any path $\bk$ and any sequence $\bt$, the following
initial data assignment:
\begin{equation}\label{aoneinit} I(\bk,\bt):\quad T_{j,k_j}=t_j\quad (j\in \Z) \end{equation}
determines uniquely the solution to the $A_1$ $T$-system. We denote by $\bk_0$ the
``flat" initial data path with $k_j^{(0)}=j$ Mod 2.

One way to understand how the $T$-system is part of a cluster algebra structure is to 
study the connection between various such admissible data. In particular, we may define
the notion of a local ``mutation" of initial data $I(\bk,\bt)$ as follows. If $k_{j-1}=k_{j+1}=k_j+\epsilon$, 
for some $j\in \Z$ and $\epsilon\in \{-1,1\}$,
the mutation $\mu_j$ at position $j$ sends $I(\bk,\bt)$ to $I(\bk',\bt')$ where:
$$k'_\ell=k_\ell+2\epsilon \delta_{\ell,j}\qquad {\rm and} 
\qquad t_\ell'= (1-\delta_{\ell,j}) t_\ell +\delta_{\ell,j}\frac{t_{j-1}t_{j+1}+1}{t_j}$$
These are actually mutations in a cluster algebra of infinite rank and geometric type, in which the initial data assignments 
$\bt=(t_j)_{j\in\Z}$ form particular clusters \cite{DFK12}., while the corresponding quiver is simply an orientation of the initial
data path $\bk$ (say with all arrows pointing up, i.e. in the direction of positive time).

Moreover, the above system is an infinite rank discrete integrable system in the following sense. 
It admits two infinite families of conserved quantities defined as follows.

\begin{lemma}\cite{DFK12}
Any solution of the $A_1$ $T$-system \eqref{aonetsys} with initial data of the form
\eqref{aoneinit} satisfies the following linear recursion
relations:
\begin{equation}\label{aonecons} T_{j-1,k-1}-d_{j-k} T_{j,k}+T_{j+1,k+1}=0 
\qquad T_{j-1,k+1}-c_{j+k} T_{j,k}+T_{j+1,k-1}=0
\end{equation}
where the conserved quantities $(c_m),(d_m)$ form two independent infinite sequences,
expressible entirely in terms of the initial data.
\end{lemma}

This dual property of both being part of a cluster algebra and being discrete integrable makes the $T$-system
particularly interesting. It is also the perfect testing ground  for the Laurent positivity conjecture\footnote{A general proof of positivity for the finite rank cluster algebras of geometric type has appeared recently \cite{KYL}.} supposed
to hold for any cluster algebra, namely that any mutated cluster (initial data here) is expressible
as a Laurent polynomial of any other, with {\it non-negative} integer coefficients.

\subsection{Matrix solution}

We define the following $2\times 2$ matrices for invertible elements $a,b\in {\mathcal A}$:
\begin{equation}\label{twobytwo}
U(a,b)=\begin{pmatrix} 1 & 0 \\ \frac{1}{b} & \frac{a}{b} \end{pmatrix}\qquad 
V(a,b)=\begin{pmatrix} \frac{a}{b} & \frac{1}{b}  \\0 & 1 \end{pmatrix}
\end{equation}

These form a $GL_2({\mathcal A})$ connection on solutions of the $T$-system in the following sense:

\begin{lemma}\label{UVone}
For any invertible elements $a,b,c\in {\mathcal A}$, we have:
\begin{equation}\label{connect} 
V(a,b)\, U(b,c)=U(a,b')\, V(b',c) \qquad {\rm iff} \qquad b b'=ac+1 \end{equation}
\end{lemma}

We may attach a product of $U,V$ matrices to any finite portion of an initial data path $\bk$
with assignments $\bt$
as follows. Going along the path for increasing values of $j$, there are two kinds of steps
say ``up" $u=(1,1)$ and ``down" $v=(1,-1)$. Consider the finite portion of path $(j,k_j)_{j\in [j_0,j_1]}$,
for integers $j_0\leq j_1$: it is made of a succession of $j_1-j_0$ steps $(1,k_{j+1}-k_j)$, 
for $j=j_0,j_0+1,...,j_1-1$
connecting the points $A_0=(j_0,k_{j_0})$ and $A_1=(j_1,k_{j_1})$. 
To each step labeled $j\in[j_0,j_1-1]$ of the path we associate a matrix $M_{\bk,\bt}(j)$ defined as follows:
$$M_{\bk,\bt}(j)=\left\{ \begin{matrix} U(t_j,t_{j+1}) & {\rm if}\, k_{j+1}=k_j+1 \\
V(t_j,t_{j+1}) & {\rm if}\, k_{j+1}=k_j-1  \end{matrix} \right. \, $$
and to the finite portion of path we associate the matrix
$$M_{\bk,\bt}(j_0,j_1)=\prod_{j=j_0}^{j_1-1} M_{\bk,\bt}(j)$$

For fixed endpoints $A_0,A_1$, we may consider the class of initial data ${\mathcal C}_{A_0,A_1}$
of the form $\{ I(\bk,\bt)\}$ 
where the paths $\bk$ pass through $A_0$ and $A_1$, are identical outside of the interval $[j_0,j_1]$
but arbitrary within this interval, while initial data are all related via iterated mutations of the form $\mu_j$
with $j\in [j_0,j_1]$.

Then we have the following 
\begin{lemma}\label{connection}
The matrix $M_{\bk,\bt}(j_0,j_1)$ is independent of the initial data $I(\bk,\bt)\in {\mathcal C}_{A_0,A_1}$,
\end{lemma}
\begin{proof}
The proof is immediate by noticing that (i) $M$ only depends on $\bk$ between the two endpoints
and $(ii)$ any other path $\bk'$ through $A_0,A_1$ can be obtained from $\bk$ via a succession
of local ``mutations" defined as above. The lemma then follows from
Lemma \ref{UVone}.
\end{proof}

This leads to the complete solution of the $T$-system:

\begin{thm}\label{aonesol}
The solution $T_{j,k}$ to the $T$-system \eqref{aonetsys} subject to the initial condition $I(\bk,\bt)$
\eqref{aoneinit} and for $k\geq k_j$  is expressed as follows. 
Let $j_0$ be the largest integer $m$ such that $(m,k+m-j)\in \bk$
and $j_1$ the smallest integer $m$ such that $(m,k+j-m)\in \bk$. Then we have:
\begin{equation}\label{formuone} T_{j,k}=t_{j_1} \left( M_{\bk,\bt}(j_0,j_1) \right)_{1,1} \end{equation}
\end{thm}
\begin{proof}
The proof is by induction under mutation. Let $A_i=(j_i,k_{j_i})$ for $i=0,1$ as before. 
We start from the ``maximal path" $\bk_{\rm max}$ that connects $A_0$ and $(j,k)$
via up steps only, and $(j,k)$ to $A_1$ via down steps only, and the corresponding assignments
$\bu$. For this path, 
we have
$$ M_{\bk_{\rm max},\bu}(j_0,j_1)=\prod_{\ell=j_0}^{j-1} U(u_\ell,u_{\ell+1}) \prod_{\ell=j}^{j_1-1}
V(u_\ell,u_{\ell+1}) $$
As all $U$'s are lower triangular, and all $V$'s upper triangular, we get:
$$u_{j_1} \left(M_{\bk_{\rm max},\bu}(j_0,j_1)\right)_{1,1}=u_{j_1}\prod_{\ell=j_0}^{j-1} \left(U(u_\ell,u_{\ell+1})\right)_{1,1}
 \prod_{\ell=j}^{j_1-1}\left(V(u_\ell,u_{\ell+1})\right)_{1,1} =u_j=T_{j,k}$$
hence the formula \eqref{formuone} holds for the path $k_{\rm max}$. We may now iteratively apply mutations 
until we reach the path $\bk$.
The general formula then follows from repeated application of Lemma \ref{connection}.
\end{proof}

\begin{example}\label{exampone}
We express the solution $T_{0,4}$ in terms of the following initial data:
$${\epsfxsize=6.cm \epsfbox{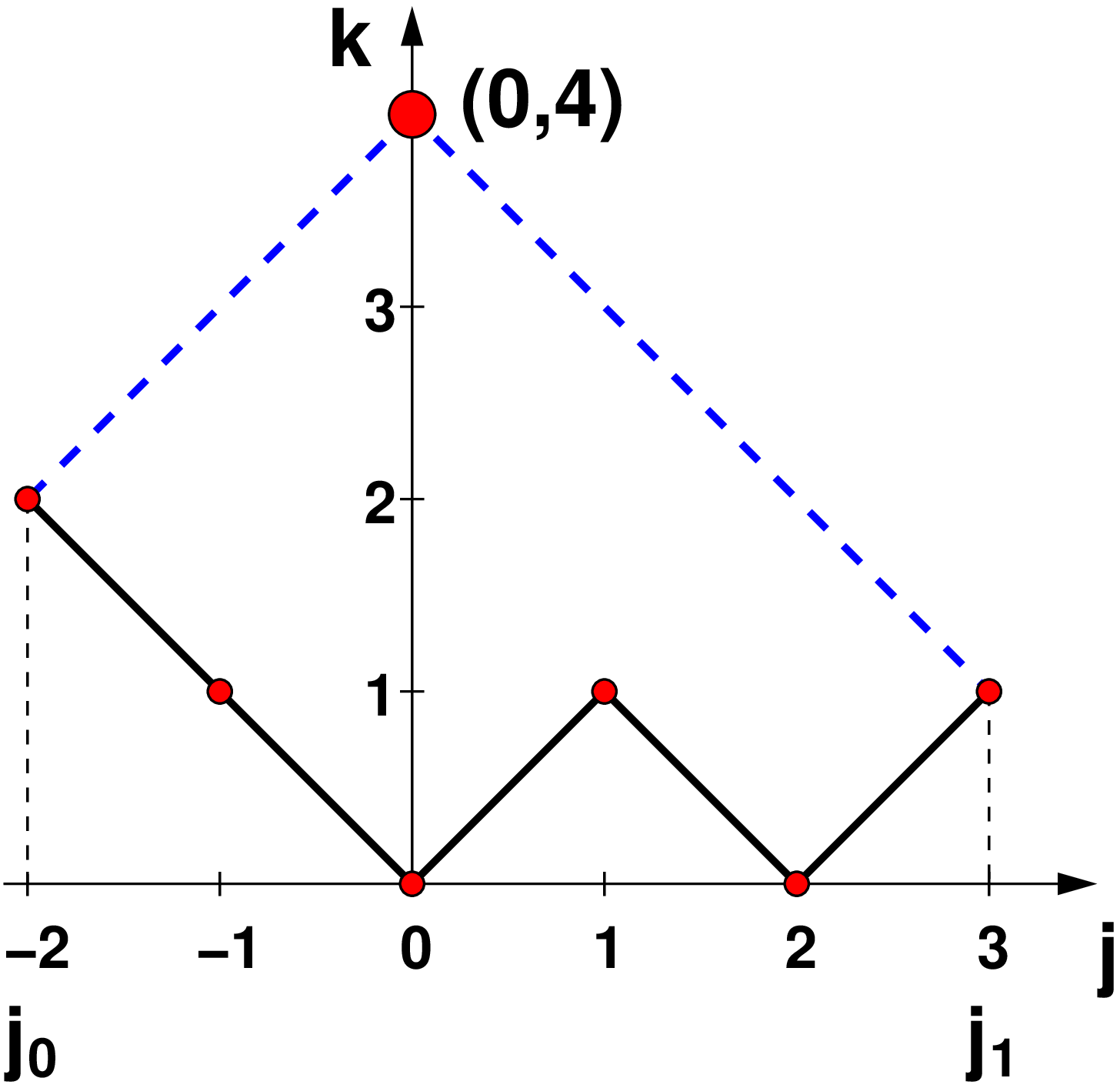}} $$
with $j_0=-2$ and $j_1=3$, and assignments $t_{-2},t_{-1},t_0,t_1,t_2,t_3$ along the finite
portion of path. We have:
\begin{eqnarray}
T_{0,4}&=&t_3 \left(V(t_{-2},t_{-1})V(t_{-1},t_0)U(t_0,t_1)V(t_1,t_2)U(t_2,t_3)\right)_{1,1}\nonumber \\
&=&\frac{t_{-2}}{t_{-1}t_1}+\frac{t_0}{t_{-1}t_1}+\frac{t_{-2}}{t_0t_2}+\frac{1}{t_{-1}t_1t_2}
+\frac{t_{-2}}{t_{-1}t_0t_1t_2}+\frac{t_3}{t_{-1}t_2}+\frac{t_{-2}t_3}{t_{-1}t_0t_2}+
\frac{t_{-2}t_1t_3}{t_0t_2}\label{exofour}
\end{eqnarray}
\end{example}

\subsection{Network interpretation}

We may interpret the matrices $U$ and $V$ as describing the weights of steps in a network as follows. Introduce
two ``chips" $U,V$ i.e. elementary pieces of oriented graph with weighted edges, say oriented from left to right,
connecting two left entry vertices $1,2$ to two right exit vertices $1,2$, such that the weight of the edge 
$i\to j$ is the $(i,j)$ entry of the corresponding chip matrix.
This gives:
\begin{equation}\label{rep2UV}
U(a,b)=
\raisebox{-.5cm}{\hbox{\epsfxsize=2.2cm \epsfbox{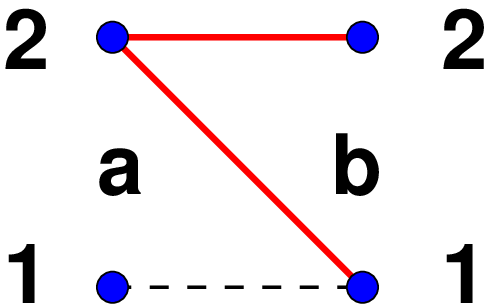}}}\qquad V(a,b)=
\raisebox{-.5cm}{\hbox{\epsfxsize=2.2cm \epsfbox{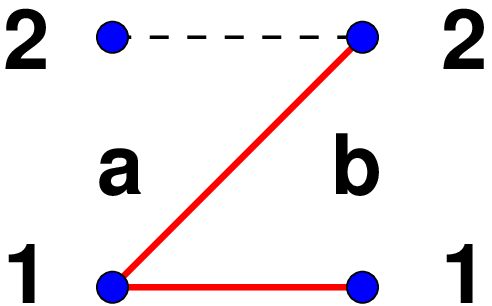}}}
\end{equation}
where we have represented in dashed line the edges with trivial weight $1$. Note that the arguments $a,b$ 
of the matrices appear as face labels on the graph. 
It is easy to see that the concatenation of chips corresponds to the multiplication of the corresponding 
weight matrices. A network is a concatenation of an arbitrary number of $U,V$ chips with compatible face labels.

We may now associate to any initial data set $I(\bk,\bt)$ the infinite network $N_{\bk,\bt}$ corresponding
to the infinite product of $U,V$ matrices $M_{\bk,\bt}(-\infty,\infty)$ along the initial data path.

For instance, in the case of flat initial data $I(\bk_0,\bt)$, the corresponding network is an alternance of $U,V$ chips:
$${\epsfxsize=12.cm \epsfbox{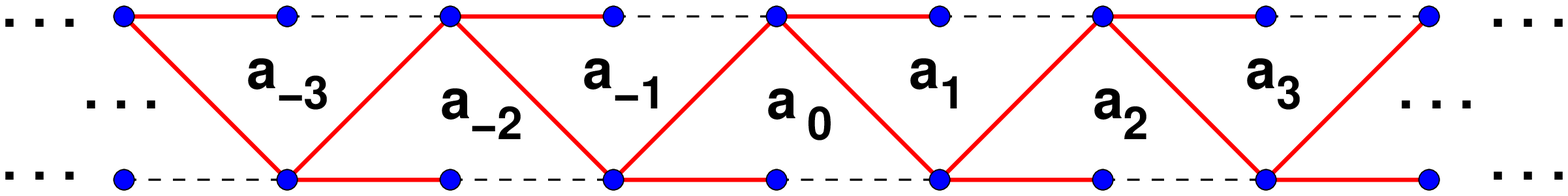}} $$

The solution of the $A_1$ $T$-system of Theorem \ref{aonesol} has the following interpretation.
\begin{cor}
The solution $T_{j,k}$ to the $T$-system \eqref{aonetsys} subject to the initial condition $I(\bk,\bt)$
\eqref{aoneinit} and for $k\geq k_j$  is 
equal to $t_{j_1}$ times the partition function of paths from entry connector $1$ to exit connector $1$
on the finite network obtained by concatenating the $j_1-j_0$ elementary chips corresponding to
the steps of initial data path $(j,k_j)_{j\in [j_0,j_1]}$.
\end{cor}

\begin{example}\label{examptwo}
Returning to the case of Example \ref{exampone}, we have the following network for $T_{0,4}$:
$${\epsfxsize=6.cm \epsfbox{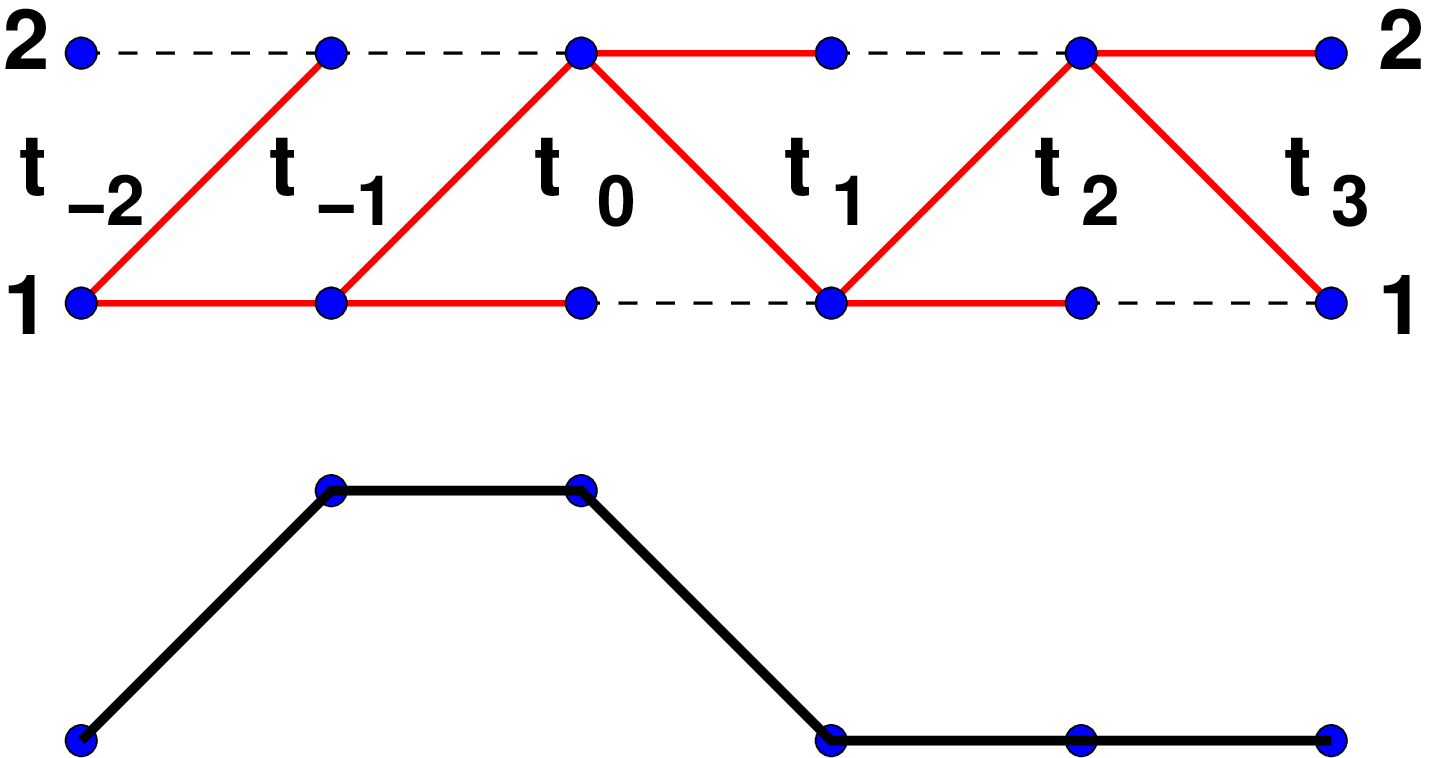}} $$
where we have represented the path contributing $t_3/(t_{-1}t_2)$ to $T_{0,4}$, 
(sixth term in the final expression \eqref{exofour}).
\end{example}

\subsection{Dominos and dimers}

In the case of the flat initial data $I(\bk_0,\bt)$, there is a simple bijection between 
paths from entry connector $1$ to exit connector $1$ on the network corresponding to say
$M=V(a_0,a_1)U(a_1,a_2)\cdots V(a_{2n-2},a_{2n-1})U(a_{2n-1},a_{2n})$
and the tilings with dominos of shape $1\times 2$ and $2\times 1$ of a rectangle of size $2\times 2n$.

Let us first transform the network graph slightly by noting that all horizontal steps of the paths
are made of pieces of length 2 (one solid followed by one dashed edge). Without altering the results,
we may therefore transform these pairs into single steps of length $2$, while keeping the original 
weight of the solid edge. Finally, let us draw a $2\times 2n$ rectangle of square lattice, with a checkerboard
bi-coloring of the faces, say grey and white, such that the top left face is white:
$$ {\epsfxsize=8.cm \epsfbox{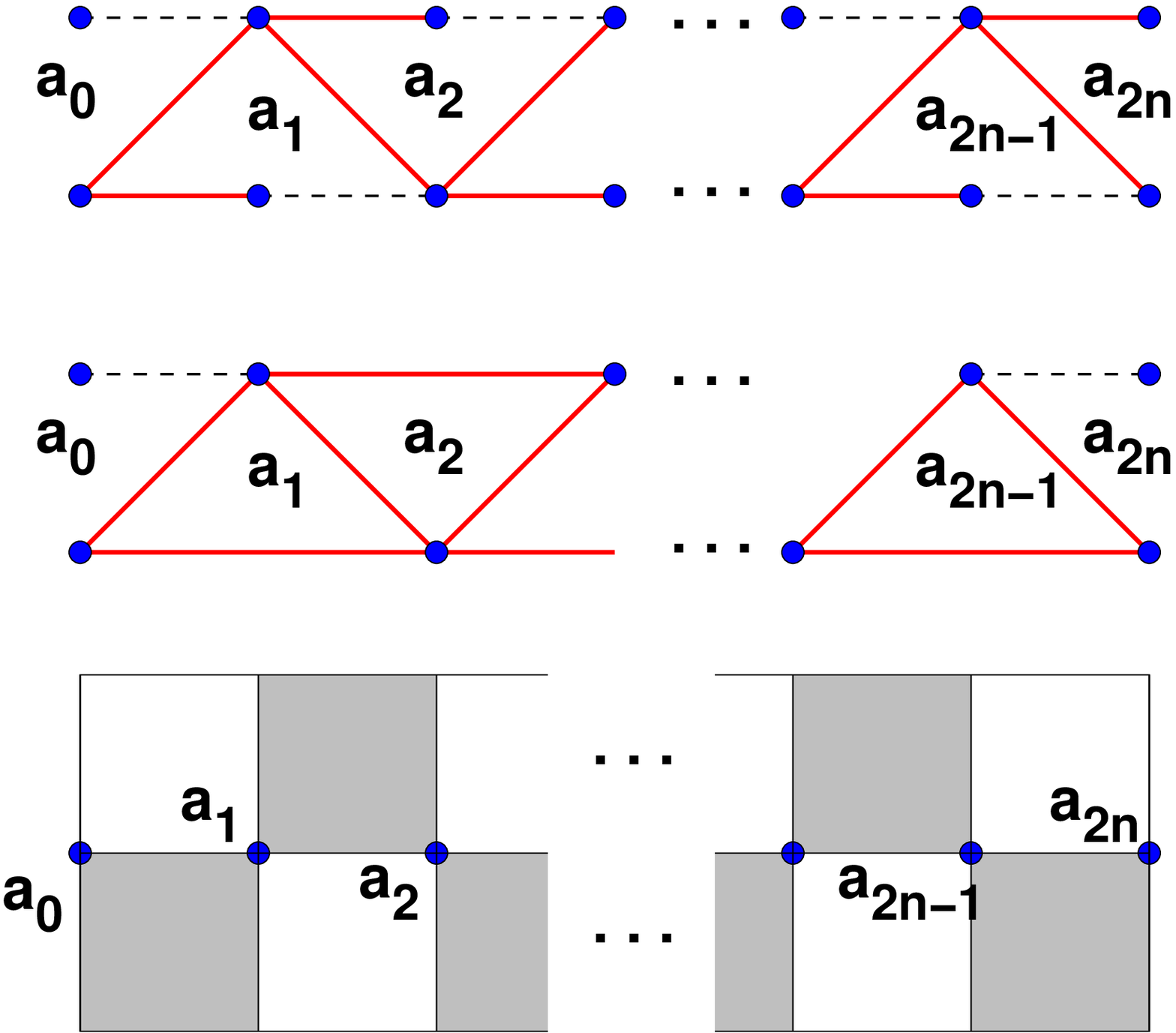}} $$
where we have associated the initial data assignments $a_0,a_1,...,a_{2n}$ to the central vertices of the rectangle.

Each path on the network is now decomposed according to its three types of steps: up, horizontal (of length 2), and down.
To each such step we associate the following tiles:
\begin{equation}\label{futiles} {\raise -1.cm \hbox{\epsfxsize=8.cm \epsfbox{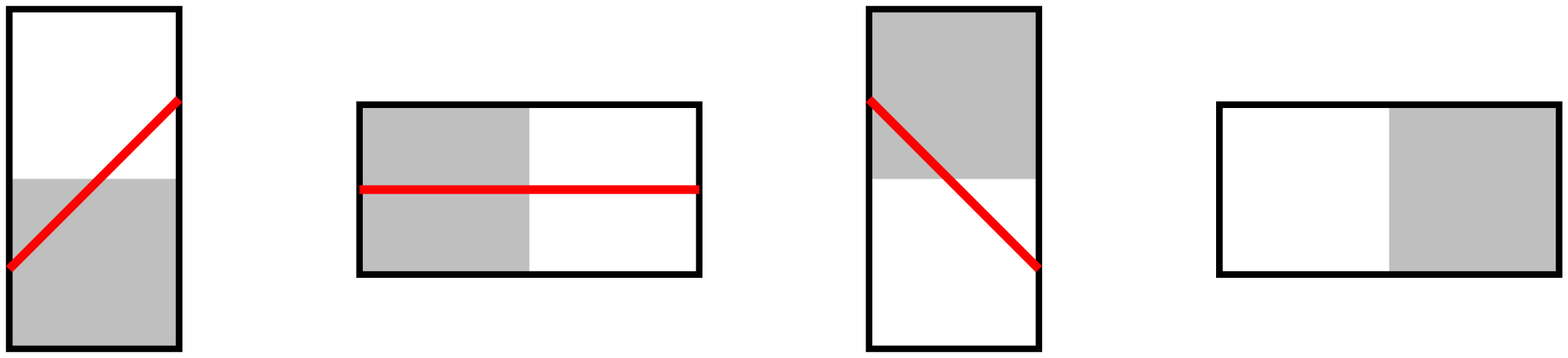}}} 
\end{equation}
where we have also represented the ``empty" tile for completeness. This clearly defines a bijection between tiling configurations
of the rectangle and path configurations on the network.

Alternatively, we may represent the tilings as dimer coverings of a vertex-bicolored square ladder-graph $L$ of size $1\times 2n+1$.
We simply represent the dual graph to the interior of the original bi-colored square lattice rectangle, 
in which the bi-coloring is transferred
to the vertices (say empty or ``white" circles for white and filled or ``black" circles for grey). 
Any tiling configuration is a pairing of neighboring
white and gray squares, which transfers dually to the pairing of neighboring vertices of the ladder via dimers
(represented as a solid edge connecting the two vertices). As all squares are
paired, all dual vertices are exhausted and we get a dimer covering of the ladder graph, such as in the example below:
$$ {\epsfxsize=10.cm \epsfbox{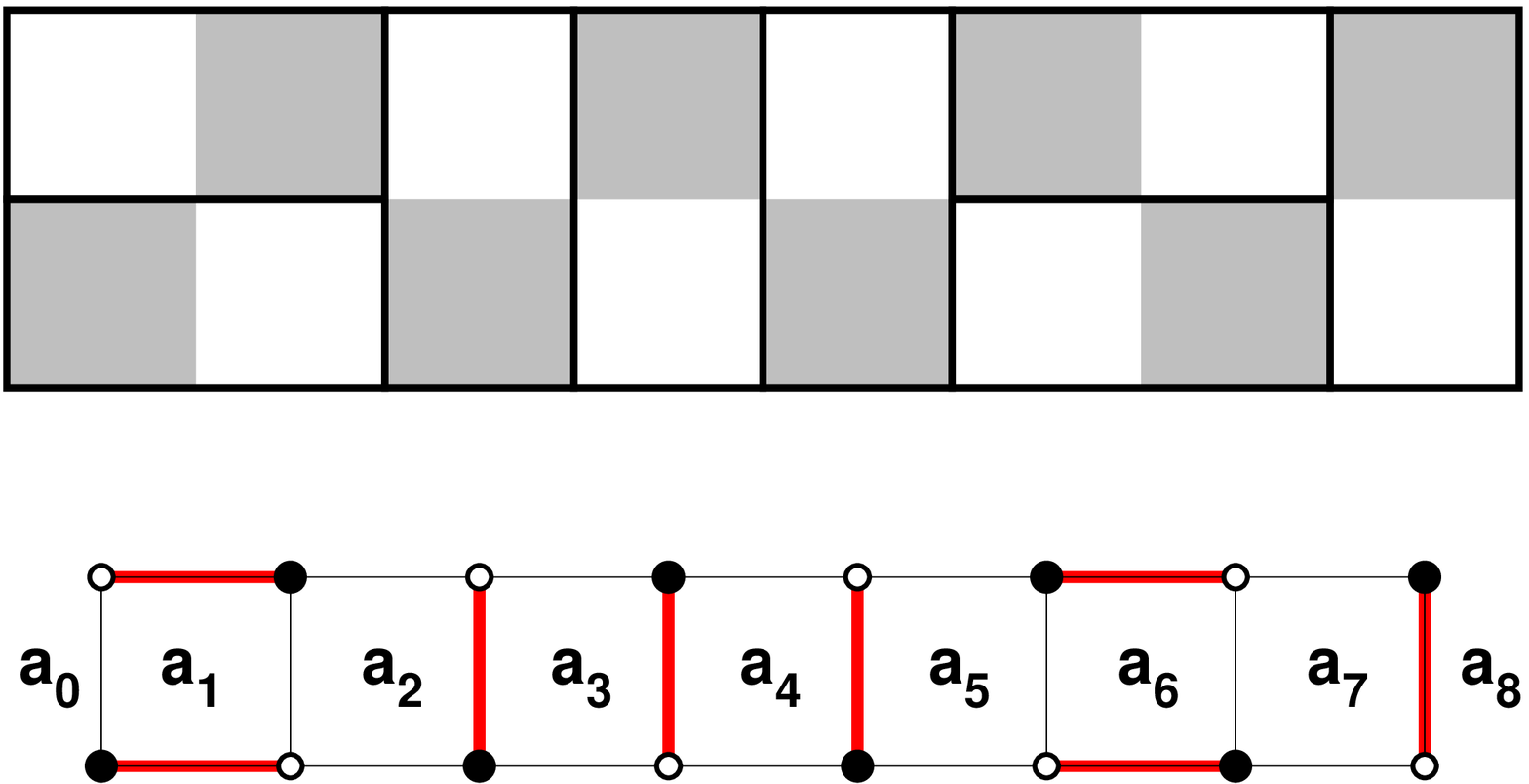}} $$
Note that the original initial data assignments $a_0,a_1,...,a_{2n}$ become face labels of the ladder in the dimer model,
including two boundaries with labels $a_0$ and $a_{2n}$.

More precisely, it is easy to translate the weights of the original path steps into weights for the dimer model as follows.
Let us concentrate on the dimers bordering a given face of the ladder. Depending on the color
of the bottom left vertex of the square, we have the following network correspondence:
\begin{eqnarray*}
\raisebox{-.5cm}{\hbox{\epsfxsize=2.2cm \epsfbox{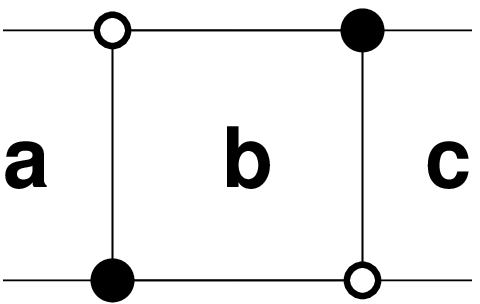}}} \quad &\to&\quad  V(a,b)U(b,c) \quad \to\quad
\raisebox{-.7cm}{\hbox{\epsfxsize=2.cm \epsfbox{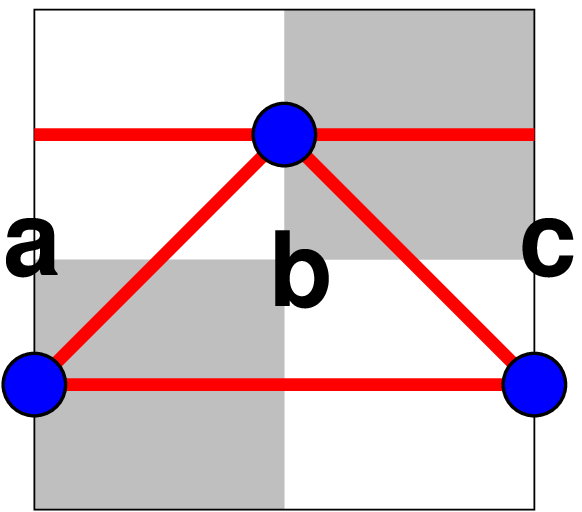}}} \\
\raisebox{-.5cm}{\hbox{\epsfxsize=2.2cm \epsfbox{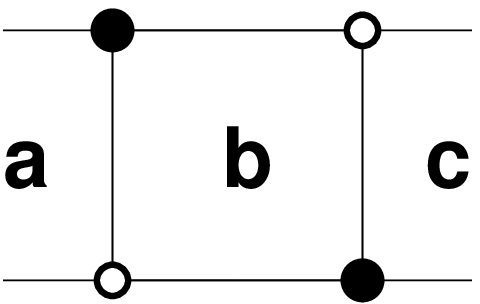}}} \quad &\to&\quad  U(a,b)V(b,c) \quad \to\quad
\raisebox{-.7cm}{\hbox{\epsfxsize=2.cm \epsfbox{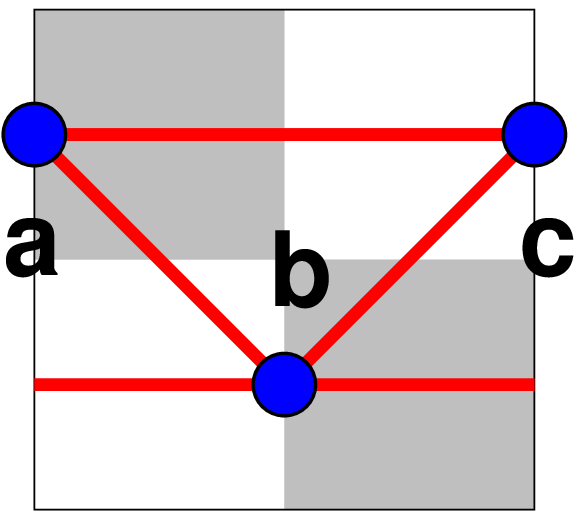}}}
\end{eqnarray*}
In both cases, the path going from left to right on the network graph may take 5 different configurations, to which there correspond 5
different dimer configurations as indicated below, together with the dependence on the variable $b$ of the corresponding weight:
$$\hbox{\epsfxsize=15.cm \epsfbox{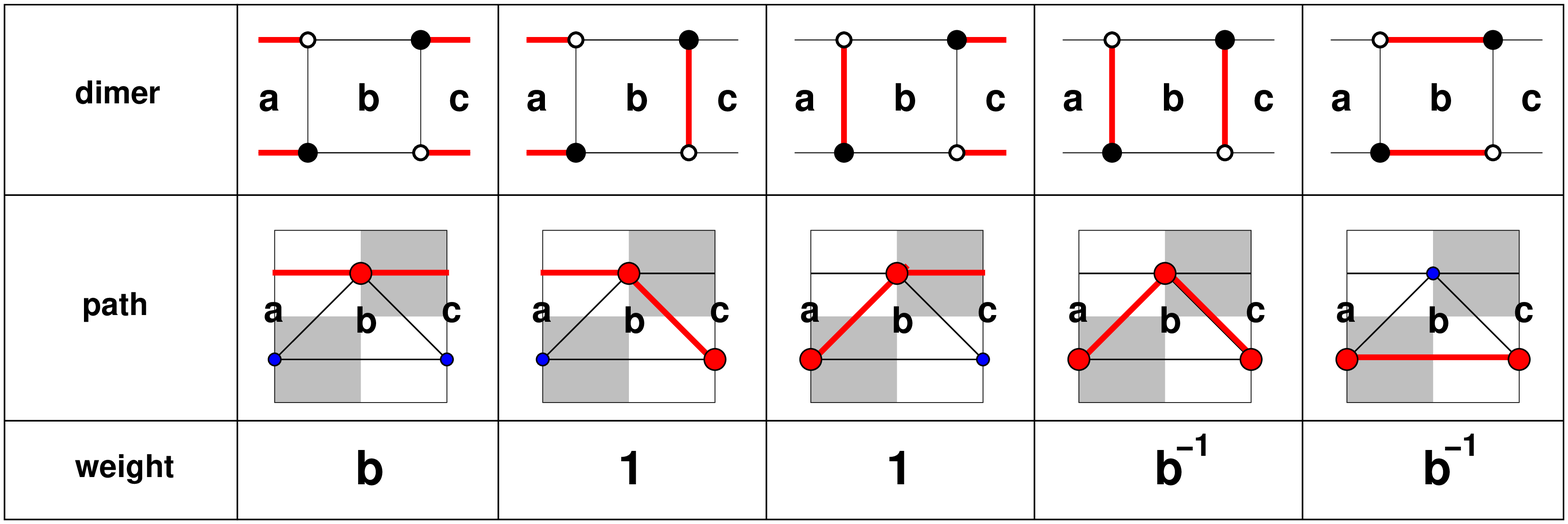}}$$
and analogously for the case of black vertex on the top left. 
Moreover, the boundary configurations on the left and right of the ladder correspond to the entry and exit connectors $1$ of the network. The correspondence between dimer and path configurations is displayed below, together with the
dependence on the boundary variable $a$. Note that we have multiplied
the path weights on their last step before the exit connector by the right boundary face value $a$, according to formula 
\eqref{formuone}.
$$\hbox{\epsfxsize=12.cm \epsfbox{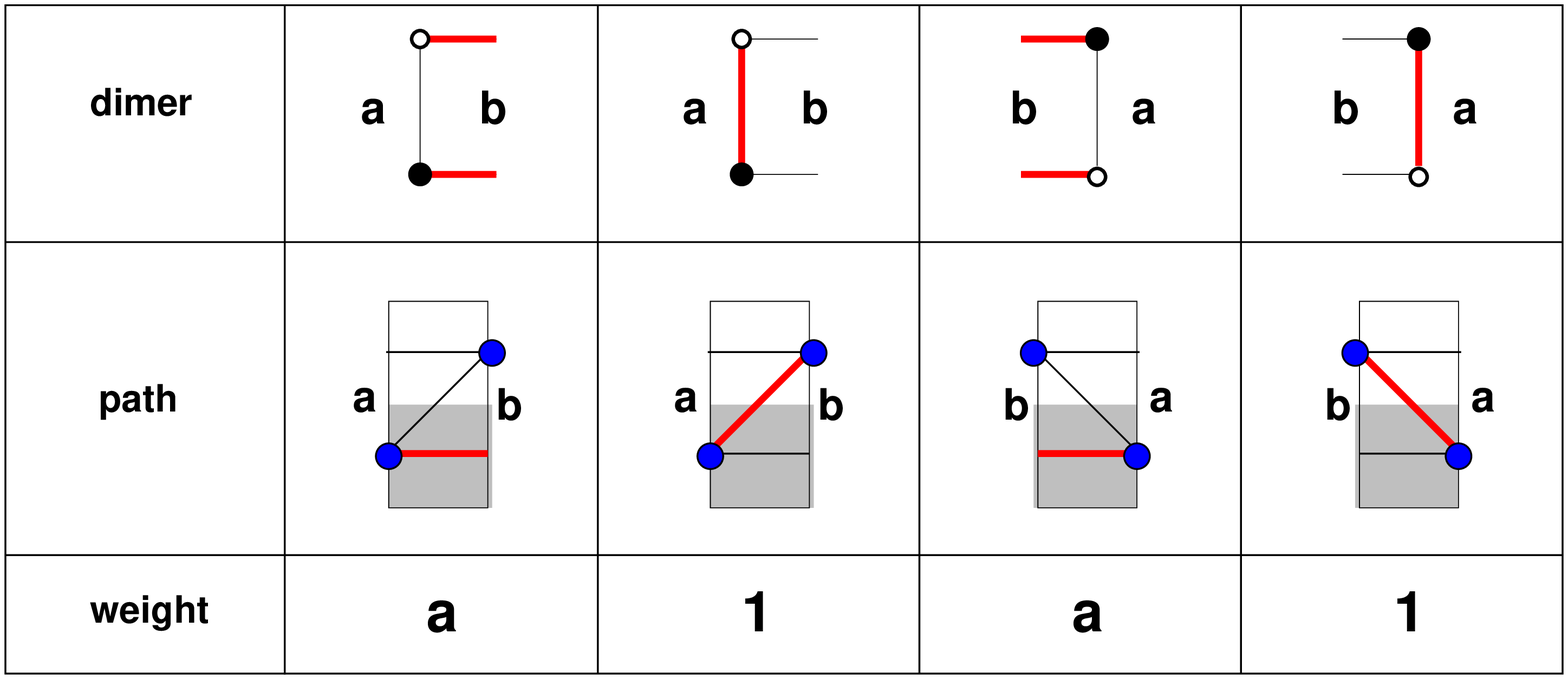}}$$

The weights of the dimer model are summarized as follows: (i) each square of the ladder with face label $b$ receives a weight
$b^{1-D}$ where D is the total number of dimers occupying edges around the face; (ii) each left/right boundary with label $a$
receives a weight $a^{1-\delta}$ where $\delta$ is the number of dimers occupying the boundary edge. With these weights,
we deduce the following:

\begin{thm}
The solution $T_{j,k}$ to the $A_1$ $T$-system \eqref{aonetsys} subject to the flat initial condition $I(\bk_0,\bt)$
\eqref{aoneinit} is the partition function of the dimer model on the corresponding ladder graph of size $2k$.
\end{thm}

Let us now consider the case of arbitrary non-flat initial data $I(\bk,\bt)$. From the network point of view,
each mutation successively applied to $I(\bk_0,\bt)$ amounts to one application of the exchange relation of 
Lemma \ref{UVone}. Translating it into the dimer language, we arrive at the following successive transformations
of the ladder graph:
$$ \hbox{\epsfxsize=14.cm \epsfbox{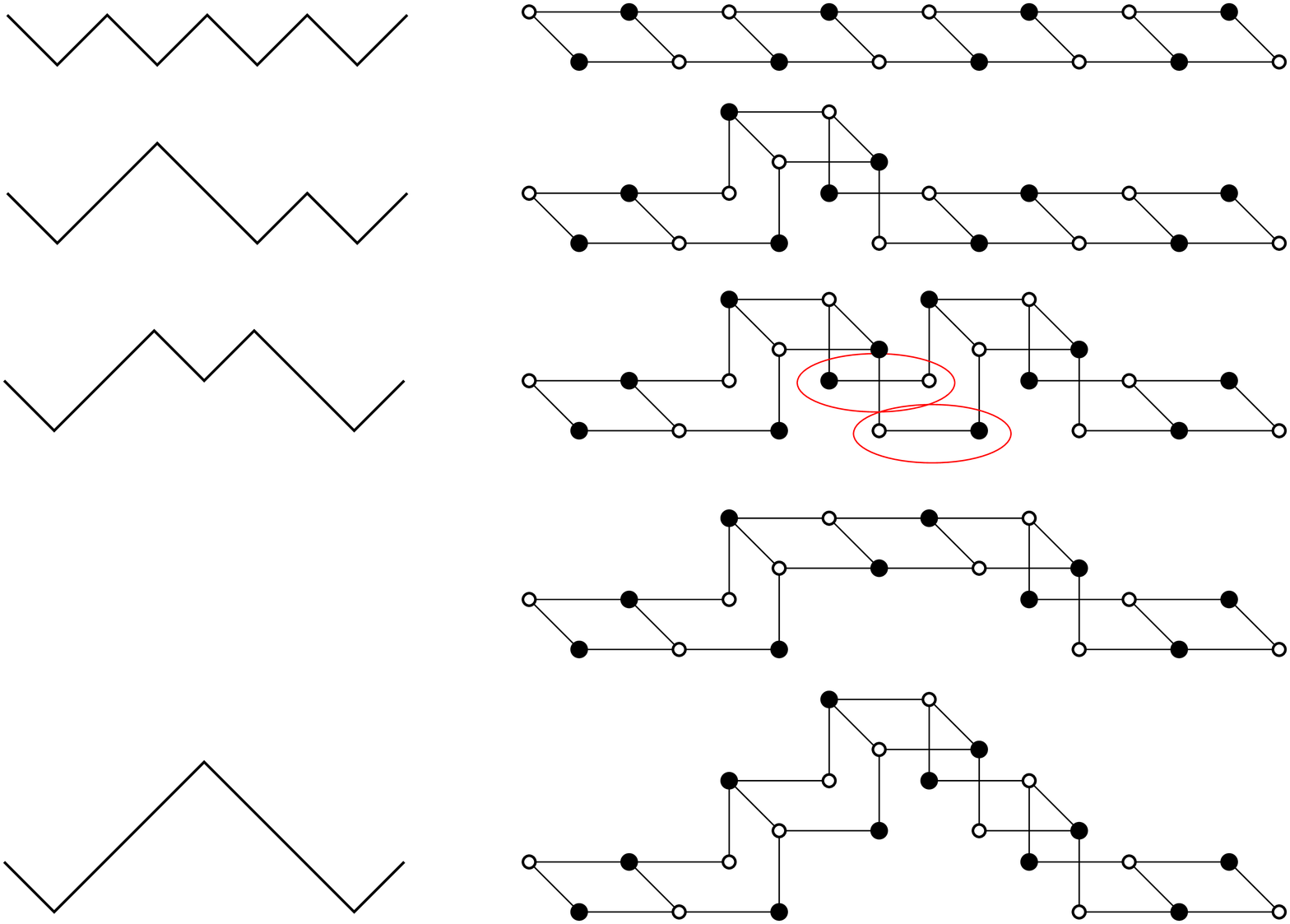}} $$
where we have represented the deformed ladder in 3 dimensions in order to make the connection more transparent.
Note that intermediate chains of black/white/black/white... vertices may be suppressed, as their dimer
covers are completely determined by the configuration of their ends. Note also that the effect of the mutation 
on the ladder graph amounts to 
the so-called ``urban renewal" move \cite{SPY}.

In general, pairs of consecutive up or down
steps in $\bk$ give rise to hexagonal faces in the transformed ladder graph, while pairs  up/down and down/up
give rise to squares. This associates bijectively a transformed ladder graph to any finite portion of path.
Note that the initial data assigned values $t_j$ are face labels of the transformed ladder graph.

We now define the dimer model on the transformed ladder graph by considering dimer coverings of the
graph, and by attaching a weight $a^{1-D}$ per square face with label $a$ whose edges are occupied by 
$D$ dimers, and a weight $a^{2-D}$ per hexagonal face with label $a$ whose edges are occupied by $D$ dimers.
The two external faces (at the ends of the transformed ladder) receive the same weights as before.

With this definition, we have the following:
\begin{thm}
The solution $T_{j,k}$ to the $A_1$ $T$-system \eqref{aonetsys} subject to the initial condition $I(\bk,\bt)$
\eqref{aoneinit} is the partition function of the dimer model on the transformed ladder graph 
corresponding to the portion of path $\bk$ between projections $j_0$ and $j_1$, and with face labels 
$(t_j)_{j\in [j_0+1,j_1-1]}$ and left/right boundary labels $t_{j_0}$ and $t_{j_1}$ respectively.
\end{thm}
\begin{proof}
We simply have to check the weight per hexagon via the network correspondence. We have two types of hexagons
corresponding to the following networks (here we adopt a 2-dimensional representation of the hexagons
rather than 3-dimensional):
\begin{eqnarray*}
\raisebox{-.5cm}{\hbox{\epsfxsize=2.9cm \epsfbox{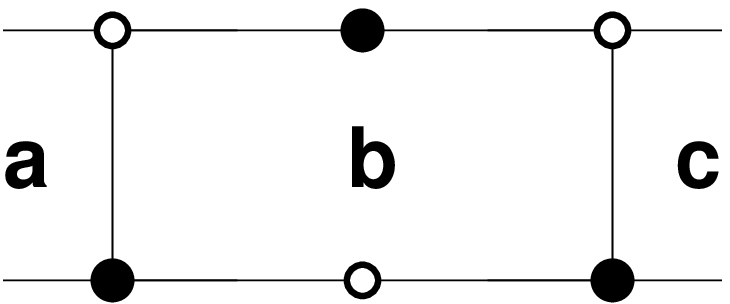}}} \quad &\to&\quad  V(a,b)V(b,c) \quad \to\quad
\raisebox{-.6cm}{\hbox{\epsfxsize=2.4cm \epsfbox{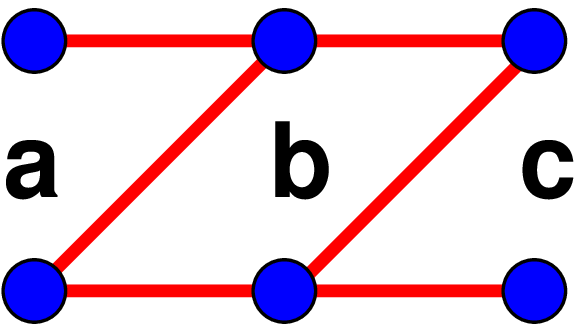}}} \\
\raisebox{-.5cm}{\hbox{\epsfxsize=2.9cm \epsfbox{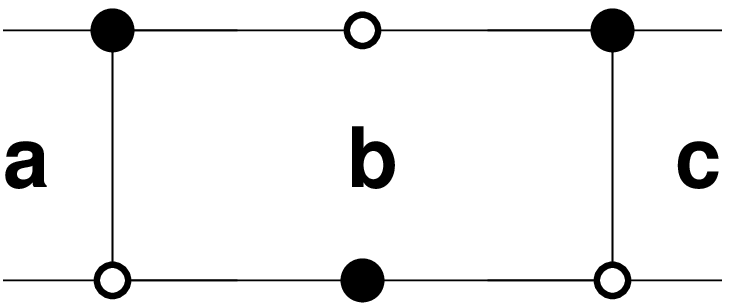}}} \quad &\to&\quad  U(a,b)U(b,c) \quad \to\quad
\raisebox{-.6cm}{\hbox{\epsfxsize=2.4cm \epsfbox{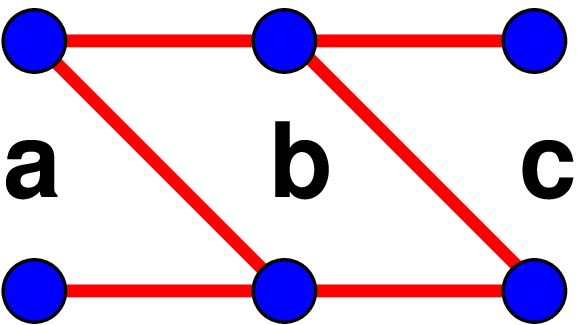}}}
\end{eqnarray*}
Each of those two networks has 4 path configurations listed below for the $VV$ case together with the 
corresponding dimer configurations, and the dependence on the label $b$ of the resulting weight:
$$  \hbox{\epsfxsize=16.cm \epsfbox{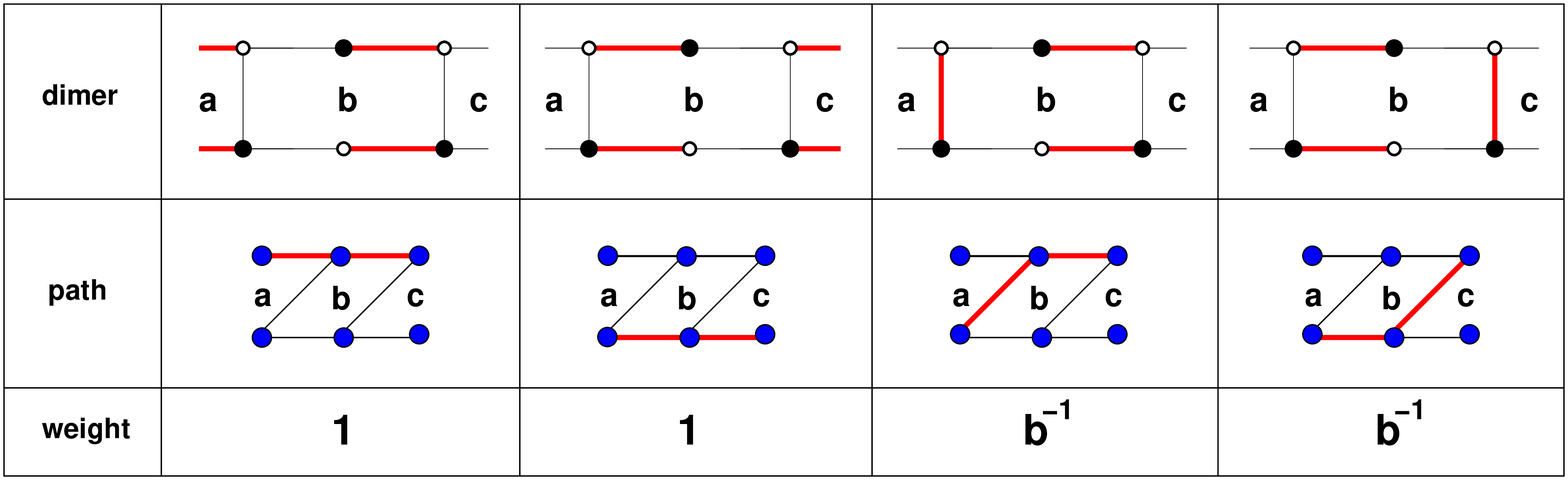}} $$
and similarly for the $UU$ case.
We find a weight $b^{2-D}$ where $D=2,3$ is the number of dimers around the hexagons, in agreement with the
above definition. The other weights have been derived earlier, and the theorem follows.
\end{proof}

\begin{example}
The network of Example \ref{examptwo}, of the form $VVUVU$ corresponds to the following transformed ladder graph:
$$  \hbox{\epsfxsize=7.cm \epsfbox{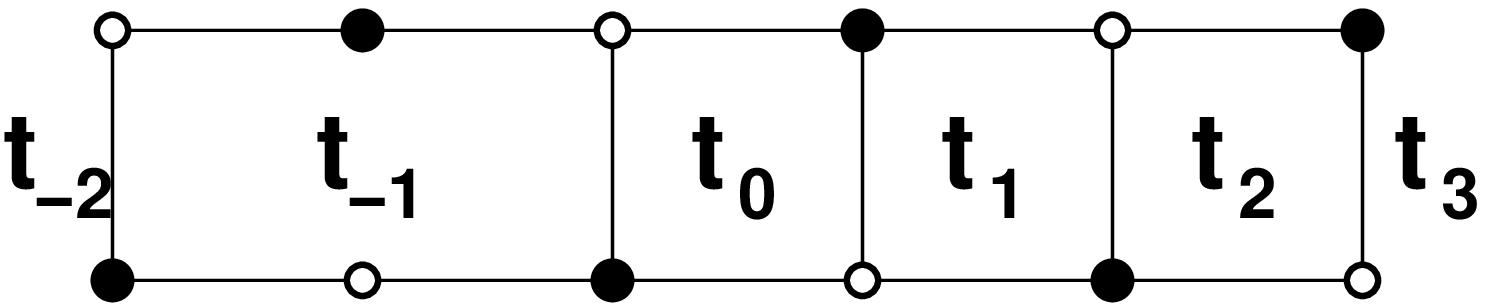}}  $$
The path displayed in Example \ref{examptwo} corresponds to the following dimer cover:
$$  \hbox{\epsfxsize=7.cm \epsfbox{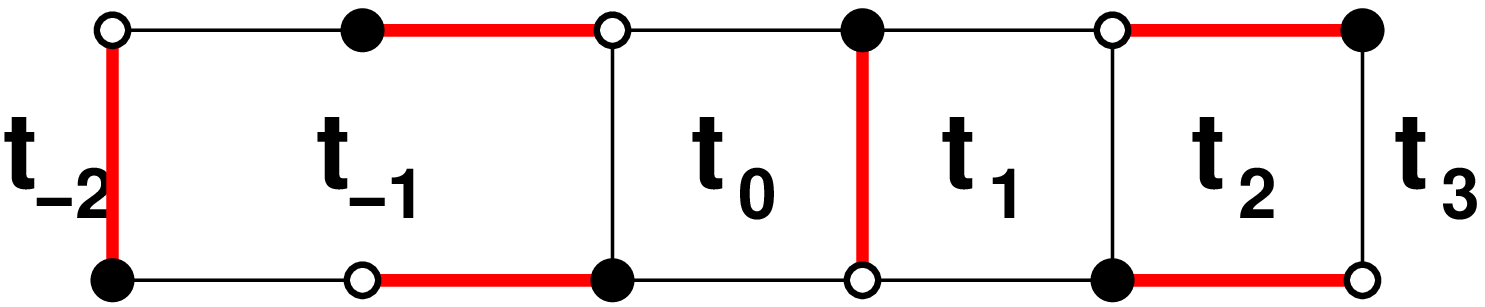}}  $$
with total weight $(t_{-2})^0 (t_{-1})^{-1} (t_0)^0(t_1)^0 (t_2)^{-1} (t_3)^1=t_3/(t_{-1}t_2)$ as expected.
\end{example}

\section{$T$-system and Dimers: the general case}

\subsection{Definitions, initial data, and cluster algebra connection}

The unrestricted $A_\infty$ $T$-system, also called octahedron recurrence, 
is the following system for formal variables $T_{i,j,k}$, $i,j,k\in \Z$:
\begin{equation}\label{tsys}
T_{i,j,k+1}T_{i,j,k-1}=T_{i,j+1,k}T_{i,j-1,k}+T_{i+1,j,k}T_{i-1,j,k} \qquad (i,j,k\in \Z)
\end{equation}
The system splits into two independent systems 
corresponding to a fixed parity of $i+j+k$. From now on we assume $i+j+k=1$ mod 2. The corresponding points
are the vertices of the Centered Cubic lattice (CC).

The system \eqref{tsys} can be considered as a three-term recursion relation in $k$.
As such it has the following sets of admissible initial data. For any ``stepped surface"
$\bk=(i,j,k_{i,j})_{i,j\in \Z}$ such that $k_{i,j}\in \Z$, $i+j+k_{i,j}=1$ mod 2, and $|k_{i+1,j}-k_{i,j}|=|k_{i,j+1}-k_{i,j}|=1$
for all $i,j\in\Z$, and any set of parameters $\bt=(t_{i,j})_{i,j\in \Z}$,
we associate the initial data $I(\bk,\bt)$:
\begin{equation}\label{infinitdata} I(\bk,\bt):\quad  T_{i,j,k_{i,j}}=t_{i,j}\qquad (i,j\in \Z)\end{equation}
namely we specify the values of $T_{i,j,k}$ at the vertices of the stepped surface $\bk$.

We still define a ``flat" stepped 
surface $\bk_0$ with $k_{i,j}^{(0)}=i+j+1$ mod 2. 
The system \eqref{tsys} is a particular mutation in an infinite rank cluster algebra of geometric type \cite{DFK13}, in which we consider 
the subset of clusters made of  the initial data assignments $\bt$. The mutation $\mu_{i,j}$ is defined as follows.
If $k_{i\pm 1,j}=k_{i,j\pm 1}=k_{i,j}+\epsilon$, for $\epsilon\in \{-1,1\}$, we have the mutated stepped surface
$\bk'=\mu_{i,j}(\bk)$ with $k_{\ell,m}'=k_{\ell,m}+2\epsilon \delta_{i,\ell}\delta_{j,m}$, and mutated assignments
$\bt'=\mu_{i,j}(\bt)$ with 
$t_{\ell,m}'=(1-\delta_{\ell,i}\delta_{m,j})t_{\ell,m}+\delta_{\ell,i}\delta_{m,j}t_{i,j}^{-1}(t_{i+1,j}t_{i-1,j}+t_{i,j+1}t_{i,j-1})$.
We see that the mutation $\mu_{i,j}$ has the effect of completing the figure with vertices $(i,j,k_{i,j}),(i\pm 1,j,k_{i,j}+\epsilon),
(i,j\pm 1,k_{i,j}+\epsilon)$ into an octahedron with the new vertex $(i,j,k_{i,j}+2\epsilon)$ (hence the name octahedron equation
often used for the $T$-system), and dropping the old vertex $(i,j,k_{i,j})$,
while keeping the rest of the stepped surface invariant. Roughly speaking, the mutation creates a local bump 
at $(i,j)$ on the stepped surface, while updating the local initial data assignment according to the $T$-system relation. 
Iterating such transformations allows to browse through all initial data $I(\bk,\bt)$.

For completeness let us describe the quiver corresponding to the flat surface $\bk_0$. The vertices are indexed by $(i,j)\in \Z^2$
and carry the labels $t_{i,j}$ of the initial data assignment \eqref{infinitdata} for $\bk_0$. The quiver is the following orientation of the
edges of the square lattice:
$$  \hbox{\epsfxsize=4.cm \epsfbox{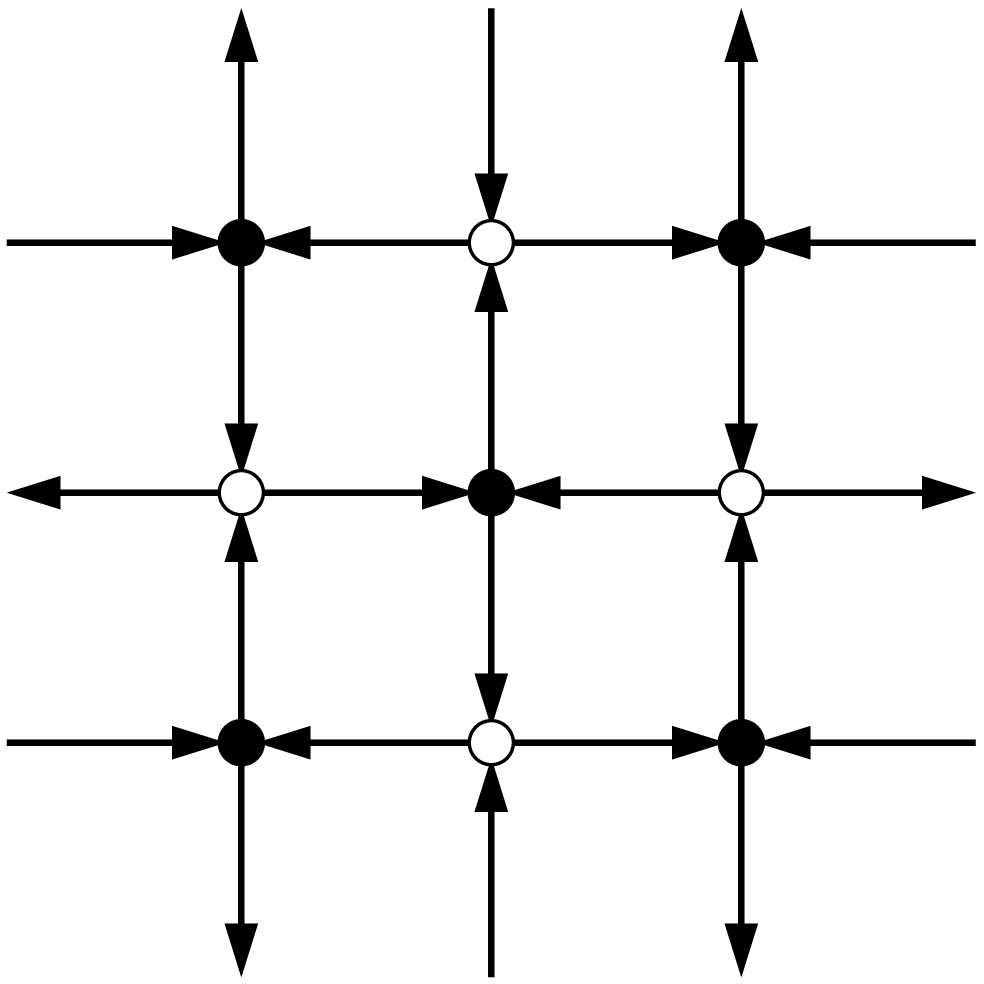}}   $$
where we have represented as filled (resp. empty) circles the vertices $(i,j)$ with $i+j=0$ (resp. $1$) mod 2.

\subsection{Matrix solution}

We define the following $2\times 2$ matrices generalizing \eqref{twobytwo}:
\begin{equation}
U(a,b,c)=\begin{pmatrix} 1 & 0 \\ \frac{c}{b} & \frac{a}{b} \end{pmatrix}\qquad 
V(a,b,c)=\begin{pmatrix} \frac{b}{c} & \frac{a}{c}  \\0 & 1 \end{pmatrix}
\end{equation}

These form a $GL_2({\mathcal A})$ connection on solutions of the $T$-system \eqref{tsys} in the following sense:

\begin{lemma}\label{octamove}
For any elements $a,b,c,u,v\in {\mathcal A}$, $b,b',c$ invertible, we have:
\begin{equation}\label{connectgen} 
V(u,a,b)\, U(b,c,v)=U(a,b',v)\, V(u,b',c) \qquad {\rm iff} \qquad b b'=ac+uv \end{equation}
\end{lemma}

In the following, we need to consider such matrices embedded into $GL_N$ for some large enough $N$. 
The embedding is as follows. 

\begin{defn}\label{uvi}
For $M=D,U$, we define $M_i\in GL_N({\mathcal A})$ as the matrix equal to the identity
except for the $2\times 2$ block at rows and columns labeled $i,i+1$, which is replaced by the $2\times 2$ matrix $M$.
\end{defn}

We wish to picture the relation of Lemma \ref{octamove}
as attached to the octahedron move described above. To this end, we consider the stepped
surface as a triangulation, with only elementary equilateral triangles of edge length $\sqrt{2}$,
having two vertices at the same time coordinate $k$ and one at time $k+\epsilon$, $\epsilon\in \{-1,1\}$. For a given stepped 
surface $\bk$, such a triangulation is not unique. Indeed, there are two distinct ways of connecting the four vertices
of an elementary regular tetrahedron with two elementary equilateral triangles. Such a tetrahedron must have two vertices 
at time coordinate $k$ and two vertices at time $k+\epsilon$. We obtain two different triangulations by connecting either pair of
equal time vertices. Note that the three edges of each triangle belong to planes parallel to each of the three coordinate planes,
and that moreover there are only 8 distinct possible configurations of triangles up to translation, corresponding to the 8 faces of an 
elementary octahedron.

To fix the abovementioned tetrahedron ambiguity, we may connect the two vertices at larger time coordinates. 
For instance, with this rule $\bk_0$
is triangulated as follows: we connect all points of the $k=1$ plane via the nearest neighbor edges of the square lattice
they form, and the four vertices adjacent to a face in this plane, say $(i\pm 1,j,1),(i,j\pm 1,1)$ are all connected to 
the center vertex $(i,j,0)$ of the $k=0$ plane, thus giving rise to four triangles. The triangulation for $\bk_0$ is
similar to an infinite ``eggbox" in which eggs can sit in half-octahedral cradles, arranged into a square lattice.
A local mutation pushes the bottom vertex of such a cradle from time $0$ to time $2$. However, we need to be able
to switch from one triangulation to another. For instance, if we want to mutate towards negative times, we need to take
the opposite convention, namely connect the two vertices at smaller time coordinates in tetrahedra. With this rule,
the same stepped surface $\bk_0$ is now triangulated in the opposite manner, with a square lattice at time $k=0$
whose 4 vertices around each square are connected to one vertex at time $k=1$, so that the eggbox now looks upside-down.
A local mutation then pushes the top of an upside-down cradle from time $k=1$ to time $k=-1$.

For any choice of triangulation of a given stepped surface $\bk$, let us further color triangles as follows. First, we bi-color
say in white and gray the 8 faces of the elementary octahedron with vertices 
$(0,0,0),(\pm 1,0,1),(0,\pm 1,1),(0,0,2)$, so that the face $(0,0,0)-(1,0,1)-(0,1,1)$ is gray. We color accordingly the faces
of the triangulation with the same color as their translate on the octahedron. The colored triangulations thus obtained
have some simple properties. Each edge parallel to the $(j,k)$ plane (which we decide to be horizontal) 
belongs to exactly one white and one gray triangle, one of which points up (towards positive $i$) and the other down 
(towards negative $i$). We may therefore decompose the triangulation into ``lozenges" made of these pairs of up/down
pointing triangles with horizontal common edge. 

Fix an initial data stepped surface and assignments $I(\bk,\bt)$, and a colored triangulation of $\bk$.
We associate
a matrix $V(u,a,b)$ to any lozenge with the down-pointing gray triangle 
whose vertices have the assigned values $(u,a,b)$ ($u$ on bottom)
and a matrix $U(b,c,v)$ to any lozenge with the up-pointing gray triangle 
whose vertices have the assigned values $(b,c,v)$ ($v$ on top). For instance,
a cradle of the eggbox for $I(\bk_0,\bt)$ is made of two lozenges, corresponding to:
\begin{equation}
\label{halfocta}
\raise-1.5cm\hbox{ $\epsfxsize=2.5cm \epsfbox{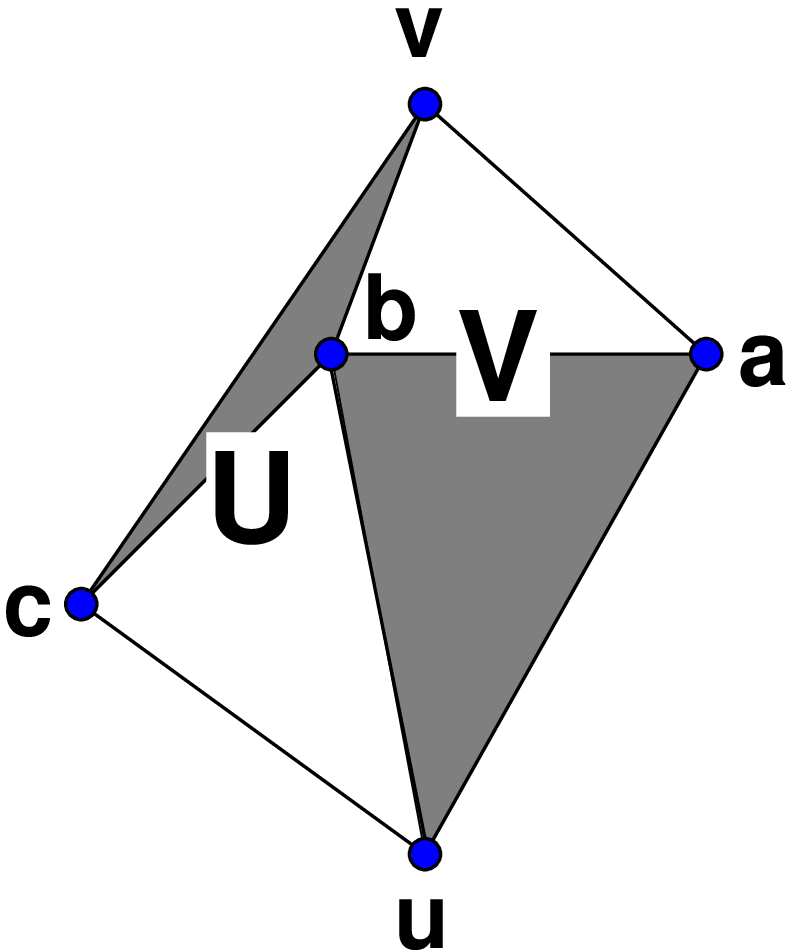}  $}  \qquad \qquad V=V_i(u,a,b)
\qquad U=U_i(b,c,v)
\end{equation}
with $b=t_{i,j}$, $a=t_{i,j-1}$, $c=t_{i,j+1}$ $u=t_{i-1,j}$ and $v=t_{i+1,j}$. (Note that here and in the following, 
the pictures and the arguments are always read from behind, from left to right and bottom to top.).
and with $U_i,V_i$ as in  Def.\ref{uvi}.

The main result recalled in this section is a formula for the solution $T_{i,j,k}$ of the $T$-system with prescribed initial data
in terms of the {\it product} of $U,V$ matrices corresponding to a domain of the initial data stepped surface.
The order in which the matrices will appear in the product is dictated by the triangulation: a matrix will be to the left of another
iff the corresponding lozenge is to the left of the other (in the direction of the $j$ axis).

The exchange relation of Lemma \ref{octamove} reads pictorially:
$$\raise-1.5cm\hbox{\epsfxsize=9.cm \epsfbox{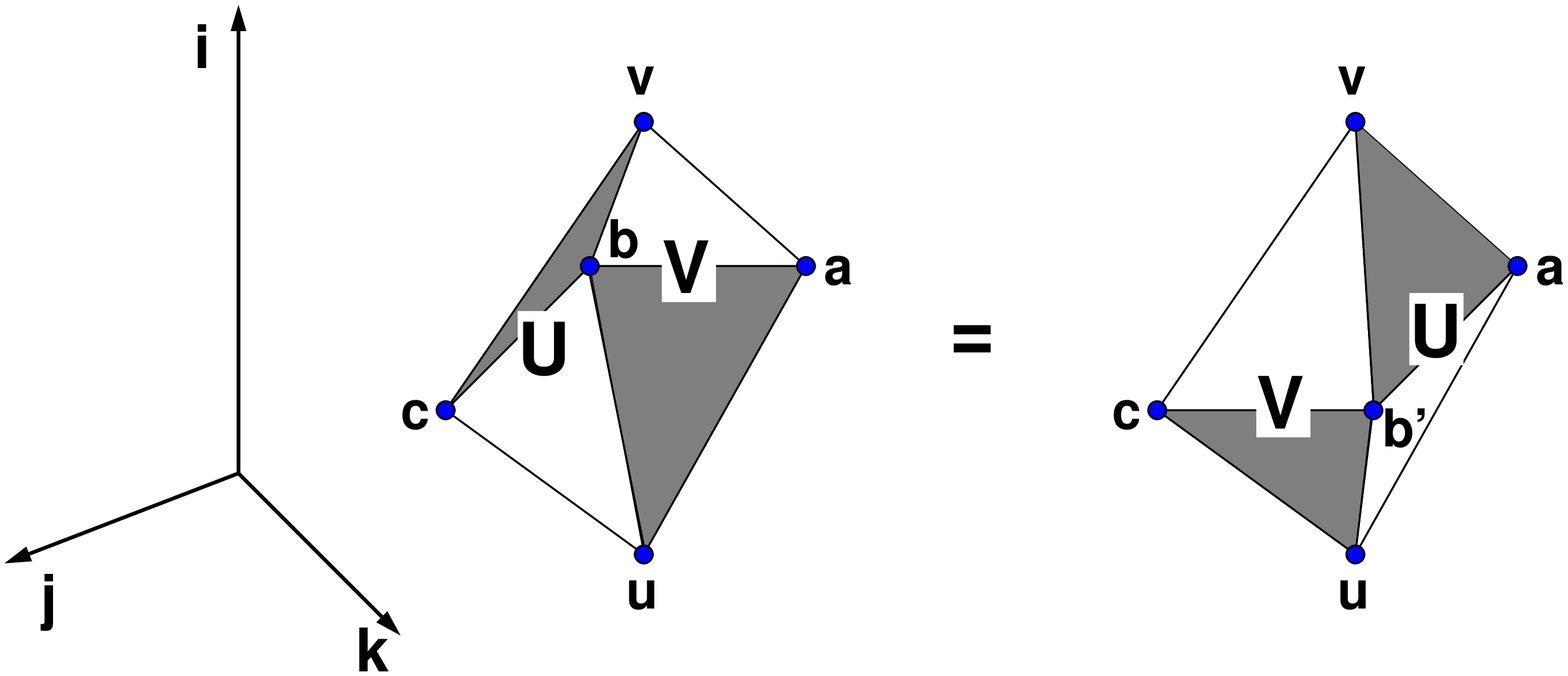} }$$

Recall that a given stepped surface $\bk$ may have many triangulations, due to the tetrahedron ambiguity.
However, the $V,U$ matrix representation is independent of the choice of triangulation in the sense of
the following:

\begin{lemma}\label{tetradec}
The two triangle decompositions of an elementary tetrahedron yield the same matrix product, namely
(triangles are viewed from behind and slightly deformed):
\begin{eqnarray*}
\raisebox{-1.cm}{\hbox{\epsfxsize=3.cm \epsfbox{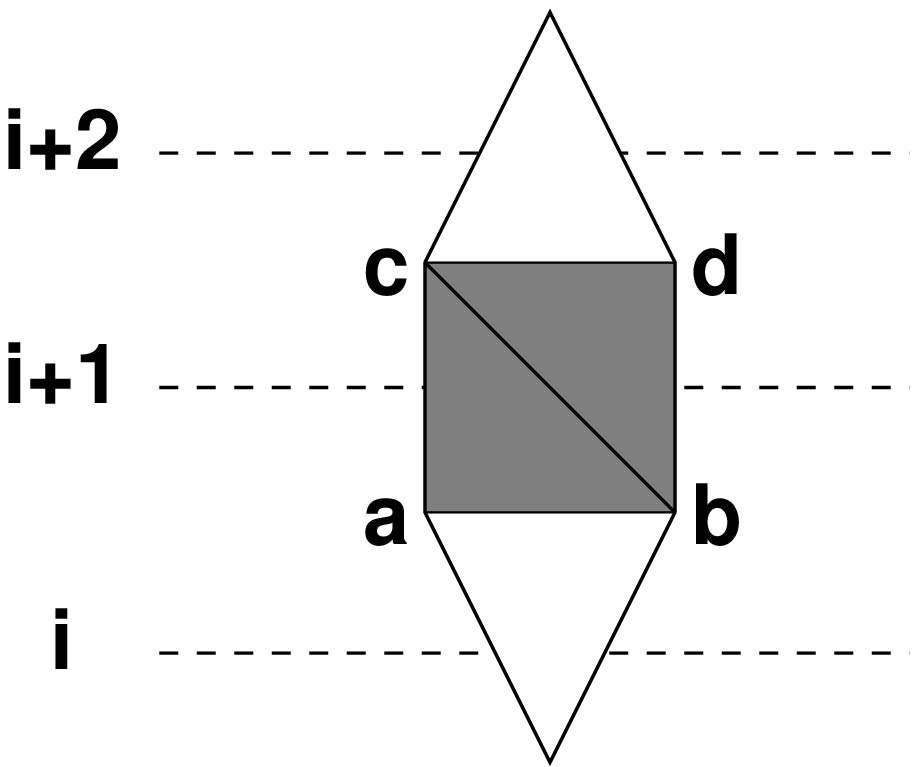}}}\quad 
&=&U_i(a,b,c)V_{i+1}(b,c,d)=V_{i+1}(a,c,d)U_i(a,b,d)  =\quad 
\raisebox{-1.cm}{\hbox{\epsfxsize=3.cm \epsfbox{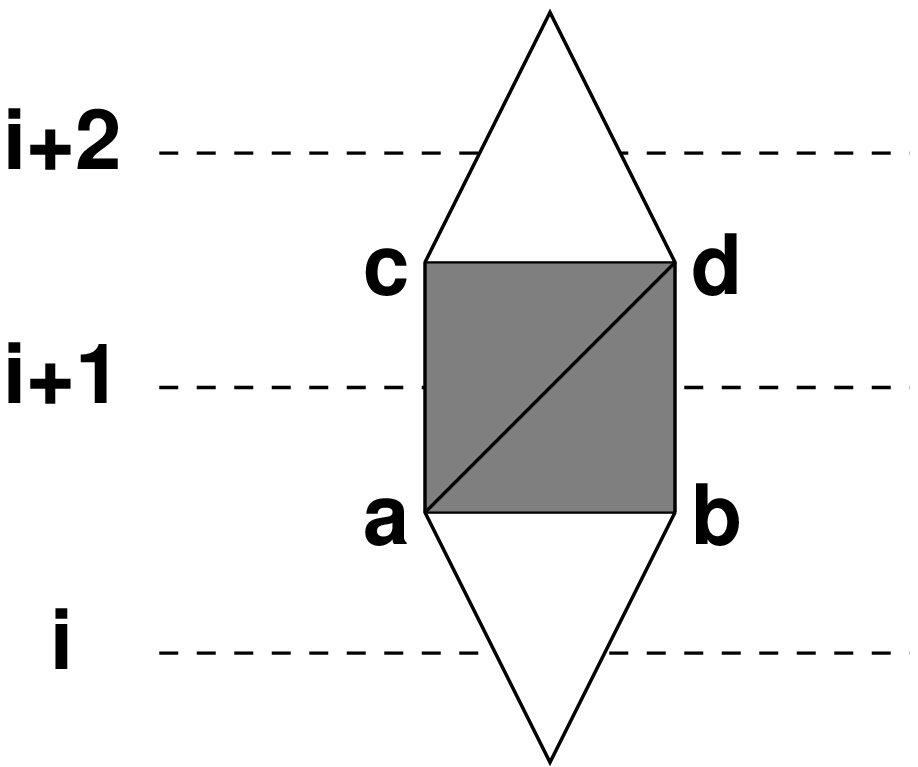}}}\\
\raisebox{-1.cm}{\hbox{\epsfxsize=3.cm \epsfbox{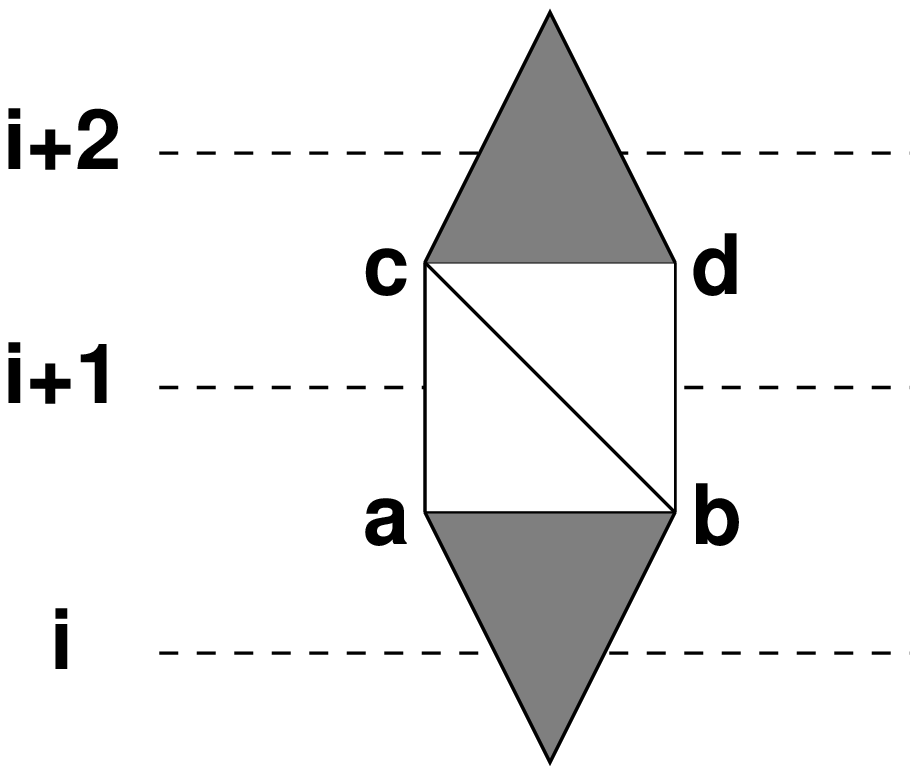}}}\quad 
&=& V_i(u,a,b)U_{i+1}(c,d,v)=U_{i+1}(c,d,v)V_i(u,a,b) =\quad 
\raisebox{-1.cm}{\hbox{\epsfxsize=3.cm \epsfbox{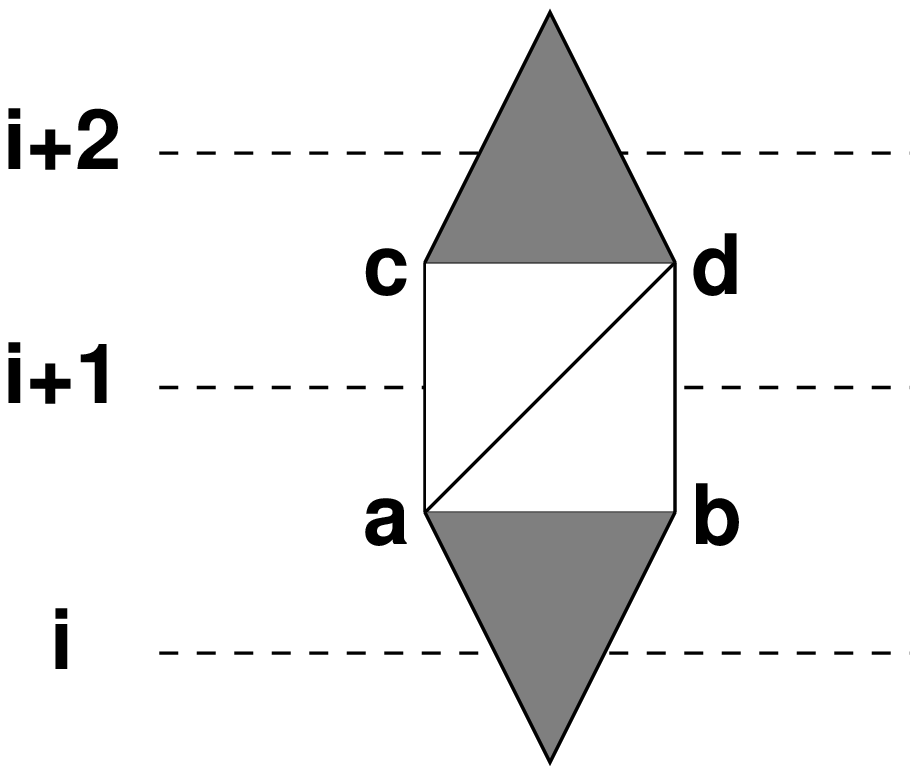}}}
\end{eqnarray*}
\end{lemma}

So, without loss of information, we may remove the ``diagonal" in each tetrahedron of a given color. 
We may represent the projection of this simplified triangulated stepped surface
onto the $(i,j)$ plane. The latter is the square lattice $\Z^2$, tessellated by gray and 
white elementary triangles and squares, and the tessellation is bi-colored. 
Alternatively, the surface $\bk$ determines uniquely such 
a tessellation, by the following local rules, where we indicate the value of $k_{i,j}$ at each vertex 
of a given square of the above projection:
\begin{equation*}
\raisebox{-1.1cm}{\hbox{\epsfxsize=14.cm \epsfbox{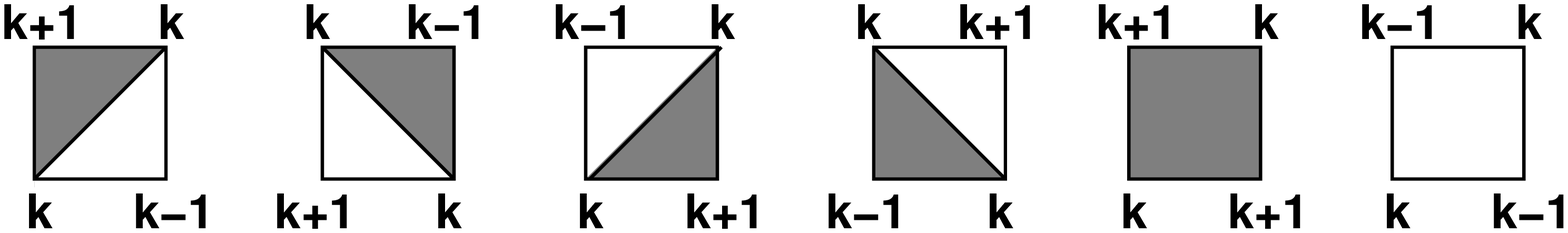}}} 
\end{equation*}
These may be summarized as follows: the diagonal of a square is an edge iff the adjacent vertices 
have same value of $k$, and the other two vertices of the square have distinct values of $k$.
Moreover the face above (resp. below) a horizontal edge joining vertices at increasing  (resp. decreasing)
values of $k$ from left to right is gray.

\begin{remark}
Note that the quiver for the data $I(\bk,\bt)$ in the cluster algebra has vertices indexed by $\Z^2$, and that the arrows
are just obtained by orienting clockwise the gray faces of the above square/triangle decomposition of the stepped surface
(in projection in the $(i,j)$ plane). 
\end{remark}

To each tessellation with squares and triangles of a finite domain of stepped surface $\mathcal D$, 
decomposed into bi-color lozenges sharing a horizontal edge (for arbitrary choices of a diagonal in each 
uni-color square), we may associate the matrix $M_{\mathcal D}$
equal to the product of $V,U$ matrices that correspond to these lozenges via \eqref{halfocta}.

\begin{figure}
\centering
\includegraphics[width=12.cm]{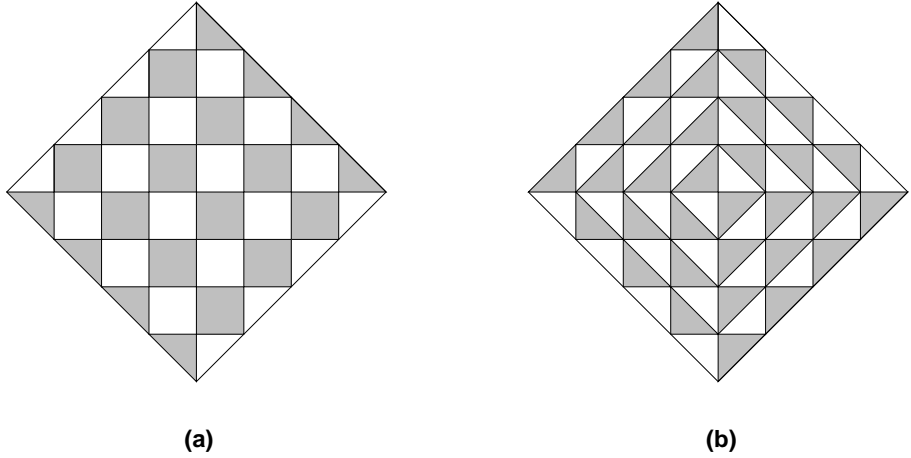}
\caption{The tessellated domains $\theta_{\rm min}(5)$ (a) and $\theta_{\rm max}(5)$ (b). 
The centers of the domains correspond respectively to
the vertices $(0,0,1)$ and $(0,0,5)$, while all the boundary vertices have $k=1$. }\label{fig:thetamin}
\end{figure}

In this language, the flat stepped surface $\bk_0$ is simply a checkerboard tessellation of $\Z^2$ with 
gray and white squares only (no triangles). For instance, the matrix $M_{\theta_{\rm min}(k)}$ corresponding 
to the ``square" domain $(x,y,z)\in \bk_0$ such that $|x-i|+|y-j|\leq |z-k|$ (see Fig.\ref{fig:thetamin} (a) for an example) reads:
\begin{eqnarray*} M_{\theta_{\rm min}(k)}&=&V_i (V_{i-1}U_iV_{i+1})(V_{i-2}U_{i-1}V_iU_{i+1}V_{i+2})\cdots 
(V_{i-k+2}U_{i-k+3}...V_{i+k-2})\\
&&\times (U_{i-k+2}V_{i-k+3}...U_{i+k-2})\cdots (U_{i-1}V_iU_{i+1})U_i
\end{eqnarray*}
if $k$ is odd, and with the substitution $U\leftrightarrow V$ when $k$ is even.
In this expression, the lozenges are obtained by picking systematically the first diagonal in each uni-color square, and we have
omitted the arguments of the matrices for simplicity.
For a given point $(i,j,k)$ at a time $k\geq 1$, we have the following simple expression for the
solution $T_{i,j,k}$ of the $T$-system \eqref{tsys} in terms of the initial data $I(\bk_0,\bt)$. For any matrix $M$,
let $|M|_{i_1,i_2,...,i_k}^{j_1,j_2,...,j_k}$ be the $k\times k$ minor of $M$ obtained by keeping only entries in rows $i_1,...,i_k$
and in  columns $j_1,...,j_k$.

\begin{thm}\cite{DFK12}\label{solflat}
Let $M_{\bk_0,\bt}=M_{\theta_{\rm min}(k)}$
be the product of $V,U$ matrices corresponding to the tessellation $\theta_{\rm min}(k)$. Then we have:
\begin{equation}
\label{tsolflat}
T_{i,j,k}=\left( \prod_{m=1}^{k-1}t_{m+i-k,j+1-m}^{-1}\,
\prod_{p=1}^k t_{p+i-k,j-1+p}\right) \, 
\left\vert M_{\bk_0,\bt}\right\vert_{i-k+2,i-k+3...,i-1,i}^{i-k+2,i-k+3...,i-1,i}
\end{equation}
\end{thm}
\begin{proof}
The proof is identical to that of Theorem \ref{aonesol}. We start with a ``maximal" tessellated initial data surface 
$\theta_{\rm max}(k)$ containing the vertex $(i,j,k)$, and with vertices $(x,y,z)$ such that $|x-i|+|y-j|< |z-k|$, $z\geq 0$,
together with the boundary $|x-i|+|y-j|=k-1$.
Let us show that the formula \eqref{tsolflat} holds for this domain, with new initial data assignments $\bu$.
The corresponding matrix reads:
\begin{eqnarray*} M_{\theta_{\rm max}(k)}&=&U_i (U_{i-1}U_iU_{i+1})(U_{i-2}U_{i-1}U_iU_{i+1}U_{i+2})\cdots 
(U_{i-k+2}U_{i-k+3}...U_{i+k-2})\\
&&\times (V_{i-k+2}V_{i-k+3}...V_{i+k-2})\cdots (V_{i-1}V_iV_{i+1})V_i
\end{eqnarray*}
Noting that $U$ is lower triangular and $V$ upper triangular, this expresses $M_{\theta_{\rm max}(k)}$ 
as a product ${\mathcal U}\times
{\mathcal V}$ of a lower by an upper triangular matrix. The relevant minor of $M_{\theta_{\rm max}(k)}$ is 
its principal $k\times k$ minor, which involves only the lower half of the domain $\theta_{\rm max}(k)$.
It is readily calculated as the product of all the $U,V$ diagonal matrix elements in the lower half of $\theta_{\rm max}(k)$.
The various products collapse row by row and cancel the prefactors,
leaving us with only the center assigned value 
\begin{equation}\label{maxres}
u_{i,j}=T_{i,j,k}=\left( \prod_{a=1}^{k-1}u_{a+i-k,j+1-a}^{-1}\,
\prod_{b=1}^k u_{b+i-k,j-1+b}\right) 
\, \left\vert M_{\theta_{\rm max}(k)}\right\vert_{i-k+2,i-k+3...,i-1,i}^{i-k+2,i-k+3...,i-1,i}
\end{equation}
We may then attain the domain $\theta_{\rm min}(k)$ by iterated mutations (backward in time) from Lemma \ref{octamove}, 
which implies that both minors for $\theta_{\rm min}(k)$ and $\theta_{\rm max}(k)$ are identical.
Noting finally that the boundary asssigned values have all $k=1$ in both domains, they actually coincide,
and so do the prefactors in \eqref{tsolflat} and \eqref{maxres}, and the theorem follows.
\end{proof}

\begin{figure}
\centering
\includegraphics[width=16.cm]{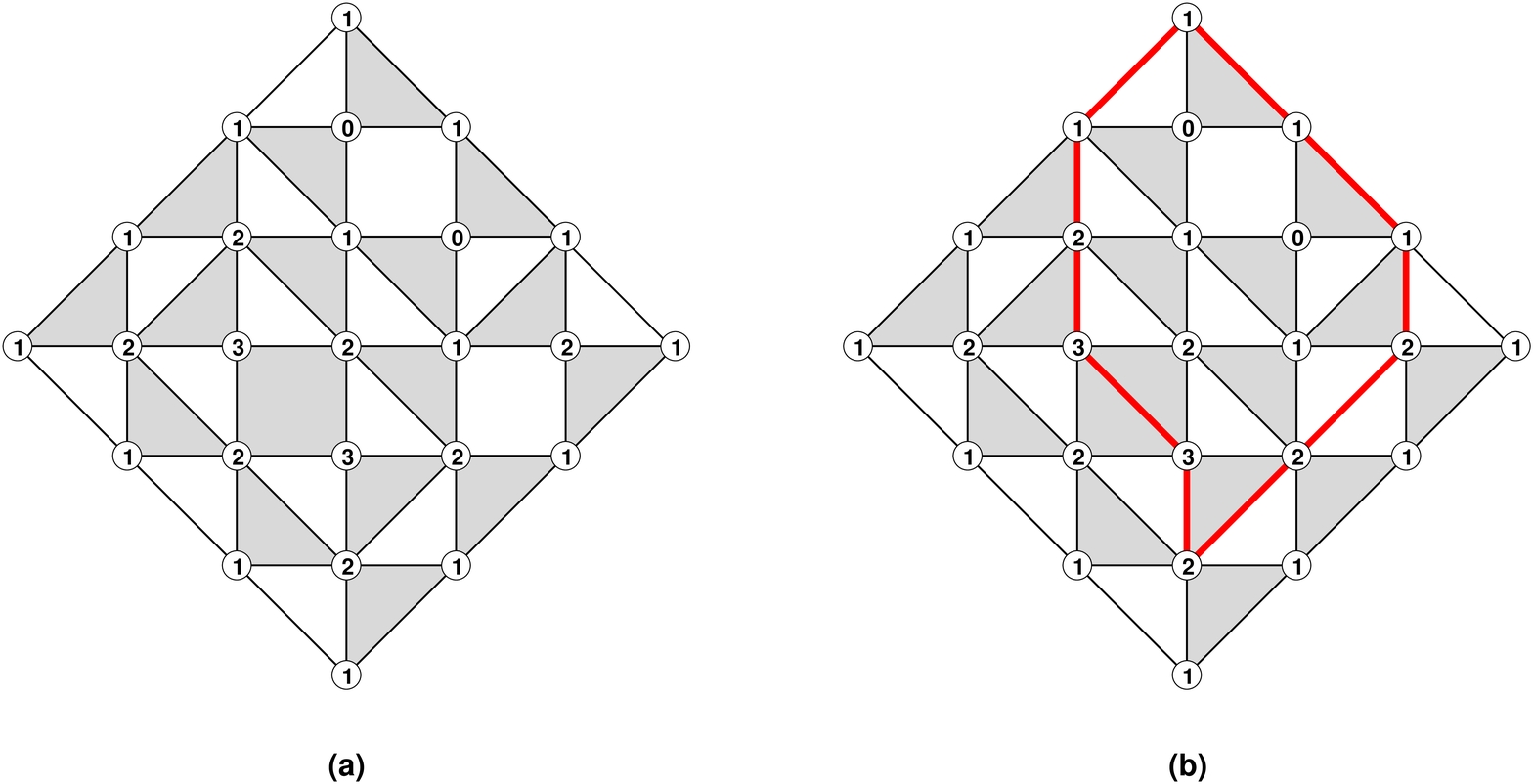}
\caption{The tessellated shadow of the point $(i,j,4)$ onto a particular initial data stepped surface. 
We have indicated for each vertex $(x,y)\in \Z^2$ of the $(i,j)$ plane projection the corresponding value 
of $k=k_{x,y}$ in a circle. The center is at coordinates $(i,j,2)$.}\label{fig:shadow}
\end{figure}

This result is easily extended to arbitrary initial data surfaces $I(\bk,\bt)$ \eqref{infinitdata} 
by the same method of proof by induction
under mutation. We simply have to define the relevant lozenge covering of the domain of initial data that must be
fed into the formula \eqref{tsolflat}. To this effect, we define the {\it shadow} of the point $(i,j,k)$
onto the initial data stepped surface $I(\bk,\bt)$ to be the {\it interior} ${\mathcal D}^{\circ}$ of the domain $\mathcal D$
of points $(x,y,z)\in \bk$  such that $|x-i|+|y-j|\leq |z-k|$, obtained by eliminating iteratively the vertex with lowest value of $k$ 
from any elementary triangle included in $\mathcal D$. The corresponding projected domain in the $(i,j)$
plane is clearly convex. Finally, if the boundary of the resulting domain contains any horizontal edge, 
we complete it by a triangle so as to form a bi-colored lozenge. The domain ${\mathcal D}^\circ$ is therefore naturally
decomposable into bi-colored lozenges, by arbitrarily choosing diagonals in its uni-colored squares.
We have represented an example of such a shadow in Fig.\ref{fig:shadow}. The corresponding matrix
(for the choice of the first diagonal in the white square) is:
$$M_{{\mathcal D}^{\circ}}= V_iV_{i+1}V_{i+2} V_{i-1}V_i V_{i+1}U_{i+2}U_{i+1}U_i $$
where again we omitted the arguments of $U,V$ for simplicity. Let us denote by $(i_a,j_a)_{a=1}^{\ell-1}$ the 
sequence of vertices $(x,y)$
on the boundary of the shadow ${\mathcal D}^{\circ}$ such that $x\leq i-1$ and $y\leq j$ (South-West corner) and 
$(i_b',j_b')_{b=1}^\ell$ the 
sequence of vertices $(x,y)$
on the boundary of the shadow ${\mathcal D}^{\circ}$ such that $x\leq i$ and $y\geq j$ (South-East corner), say from bottom to top.

We have the following:

\begin{thm}\label{solgen}
The solution $T_{i,j,k}$ of the $T$-system \eqref{tsys} with initial conditions $I(\bk,\bt)$ \eqref{infinitdata}
reads:
\begin{equation}\label{tsolgen}
T_{i,j,k}=\left( \prod_{a=1}^{\ell-1}t_{i_a,j_a}^{-1} \, \prod_{b=1}^\ell t_{i_b',j_b'} \right) \, 
\left\vert M_{{\mathcal D}^\circ}\right\vert_{i-\ell+1,i-\ell+2...,i-1,i}^{i-\ell+1,i-\ell+2,...,i-1,i}
\end{equation}
\end{thm}

\subsection{Network interpretation}

We may now extend the network interpretation to  $U_i,V_i$ matrices as follows. We use the $GL_N$ embedding
of Def.\ref{uvi} to interpret $U_i,V_i$ as network chips that connect entry points labeled $i,i+1$ to exit points labelled
$i,i+1$. When considering a combination of several such chips, we simply concatenate them by identifying matching label
exit and entry connectors of successive chips. A product of $U,V$ matrices corresponds to a network with face labels
according to the rule:
\begin{equation}\label{rep2UVi}
U_i(a,b,u)=
\raisebox{-.4cm}{\hbox{\epsfxsize=2.2cm \epsfbox{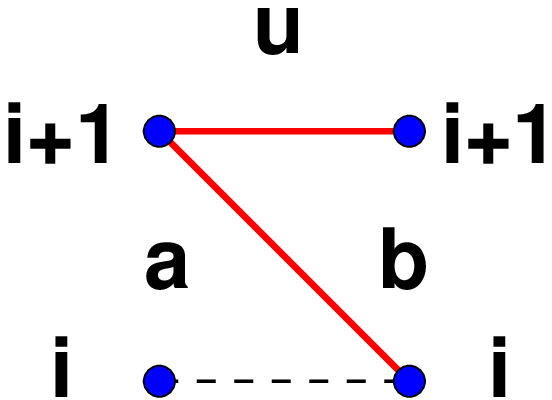}}}\qquad V_i(v,a,b)=
\raisebox{-.9cm}{\hbox{\epsfxsize=2.2cm \epsfbox{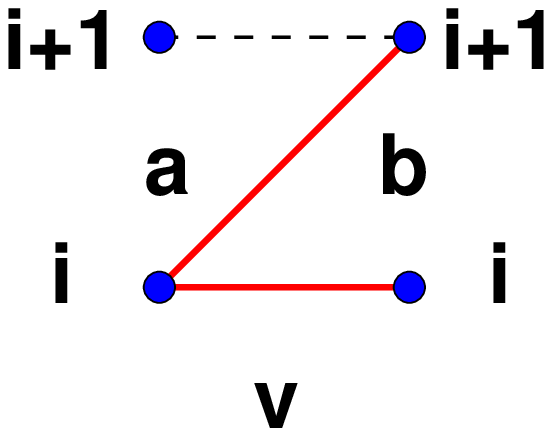}}}
\end{equation}

\begin{example}
The network for the flat initial data of the tessellation $\theta_{\rm min}(5)$ of Fig.\ref{fig:thetamin} (a) reads:
$$\raise-1.cm\hbox{\epsfxsize=11.cm \epsfbox{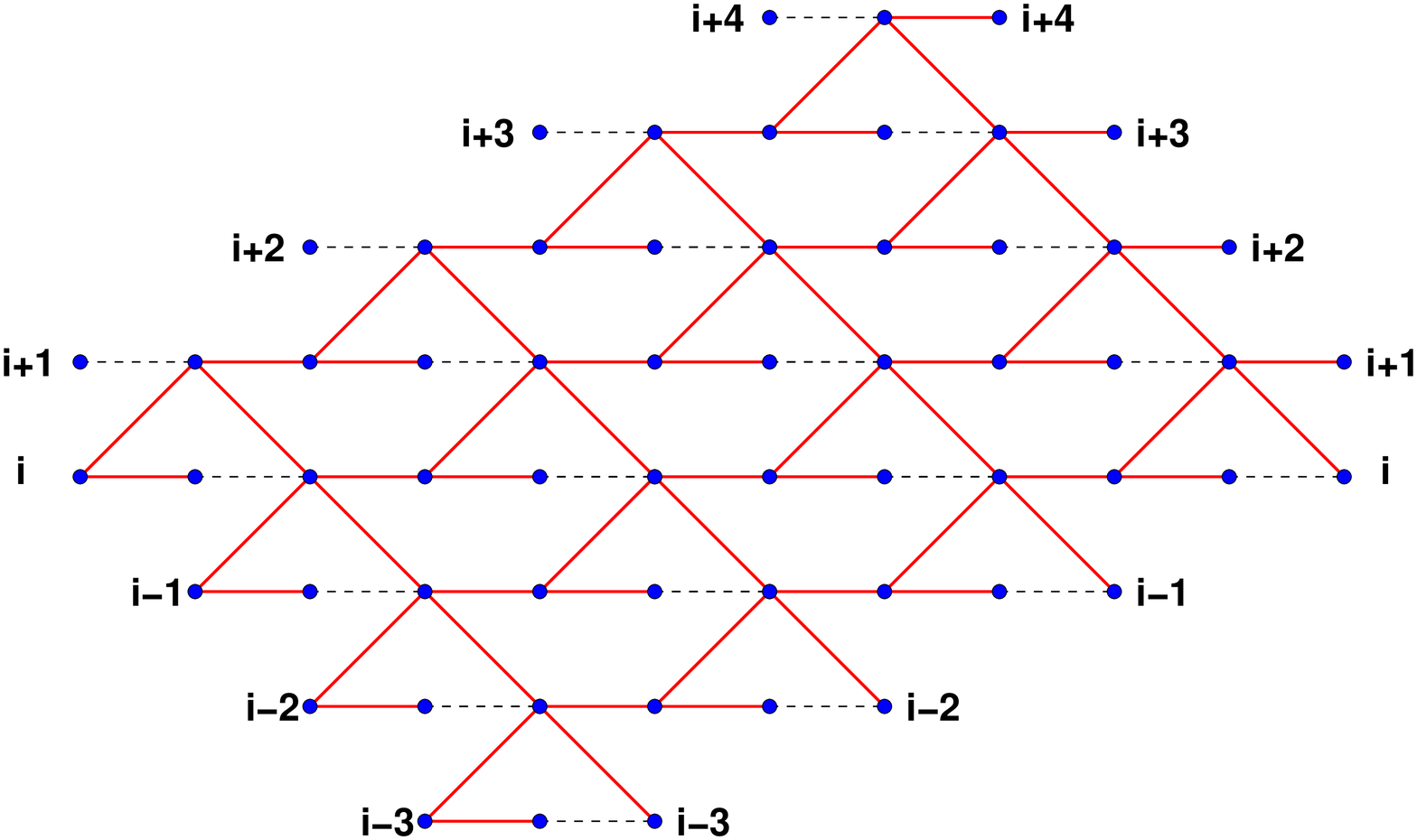} } $$
\end{example}

\begin{example}
The network for the initial data of Fig.\ref{fig:shadow} (b) reads:
$$\raise-1.cm\hbox{\epsfxsize=8.cm \epsfbox{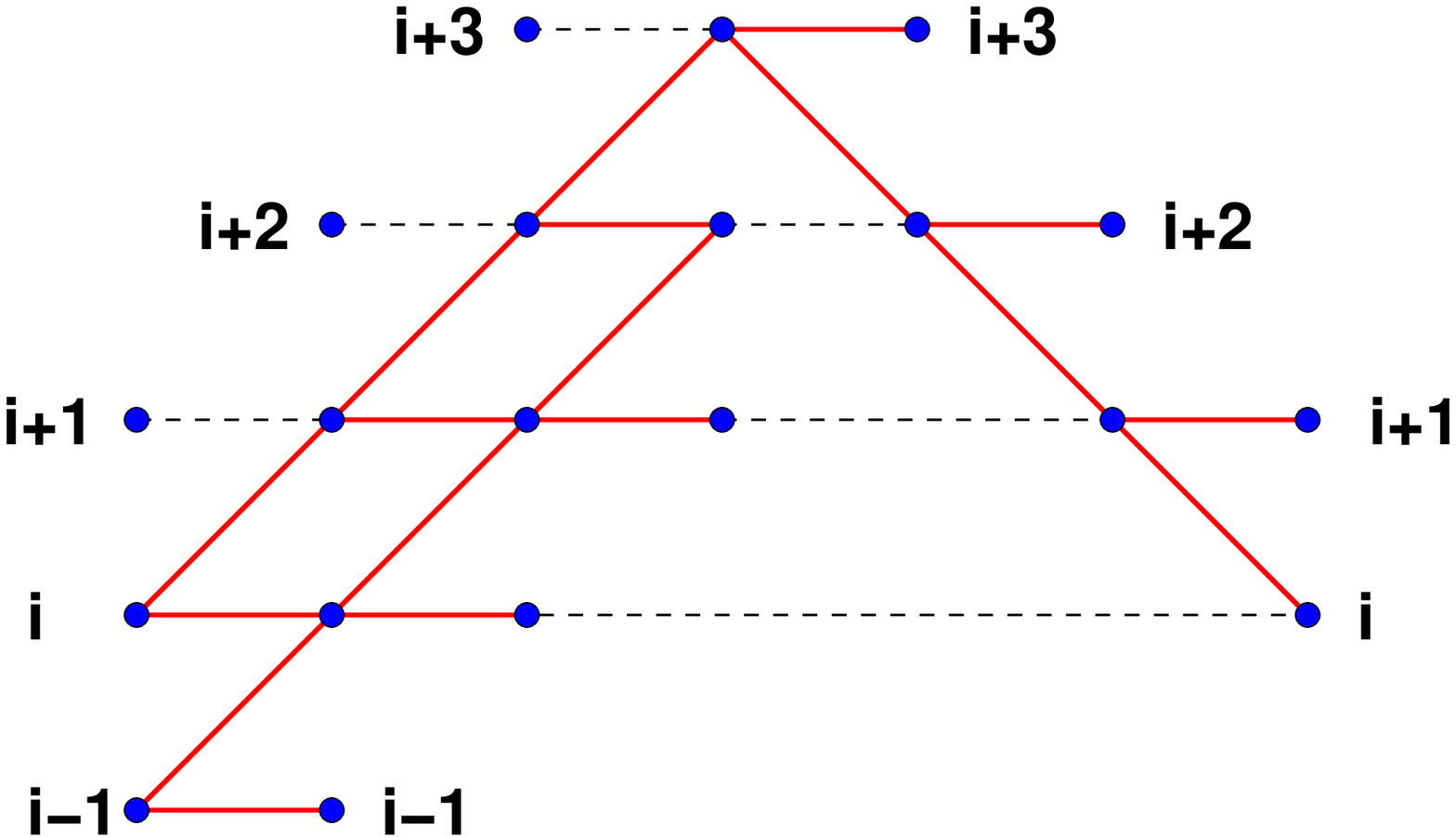} } $$
\end{example}

\subsection{Flat initial data and domino tilings of the Aztec diamond}

\begin{figure}
\centering
\includegraphics[width=16.cm]{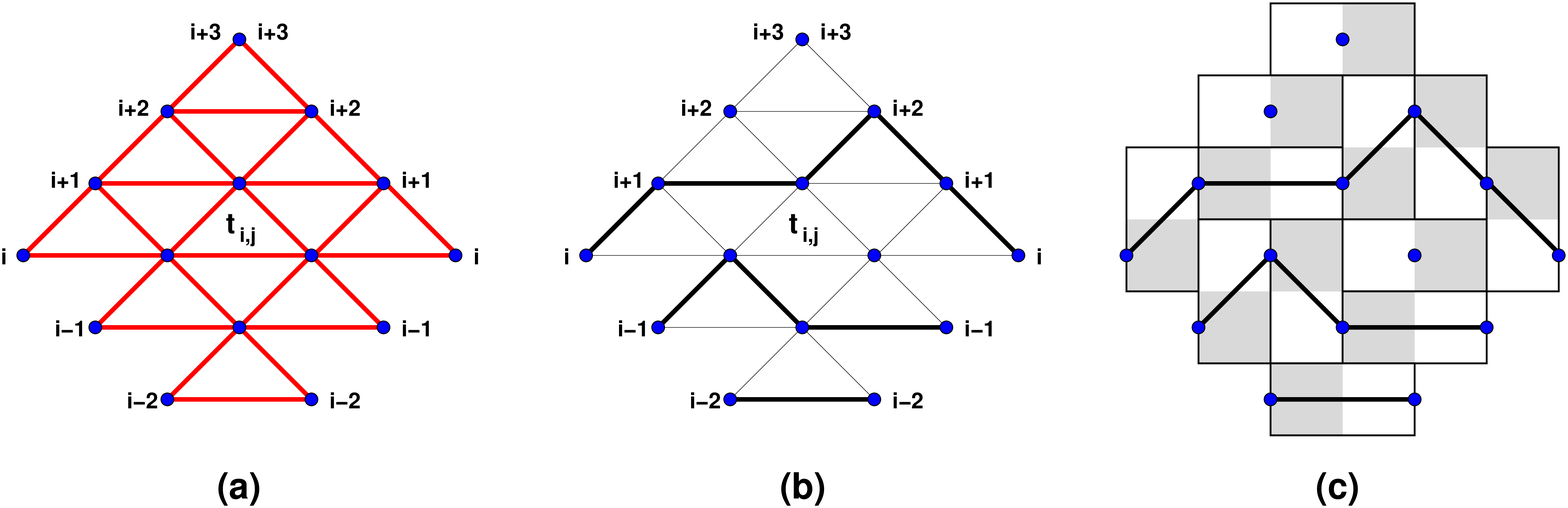}
\caption{The simplified network for the expression of $T_{i,j,k}$ as a function of $I(\bk_0,\bt)$ for $k=4$ (a), and
a sample configuration of non-intersecting paths contributing to the minor $|M_{\bk_0,\bt}|_{i-2,i-1,i}^{i-2,i-1,i}$ (b).
The path configuration is associated bijectively with a domino tiling of the Aztec diamond of size $k=4$ (c).
}\label{fig:triang}
\end{figure}

The network for the flat initial data domains $\theta_{\rm min}(k)$ may be simplified
by forming pairs of matrices $U_jV_{j+1}$ and replacing their product by a new elementary piece of
network:
\begin{equation}\label{rultri} 
U_j(a,b,c)V_{j+1}(b,c,d)\, =\, { \raise-.9cm\hbox{\epsfxsize=3.cm \epsfbox{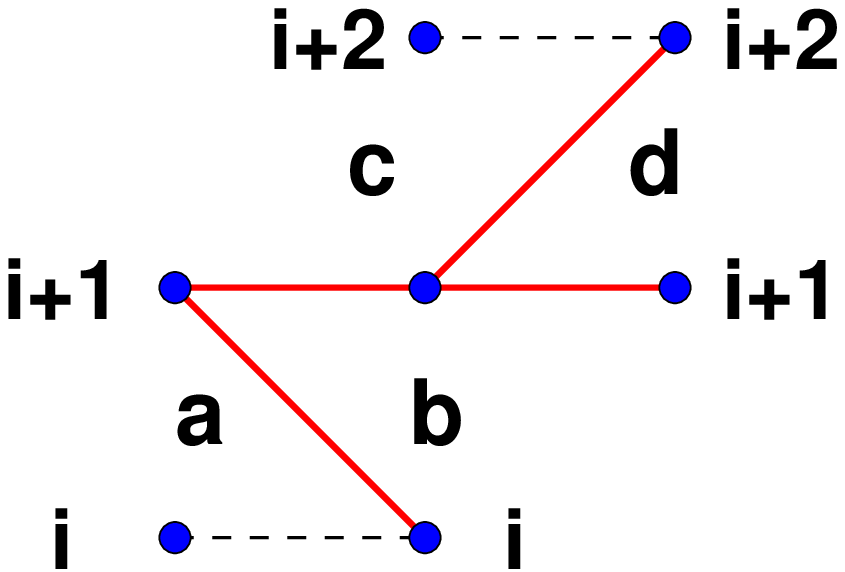} }}\, 
= \,\, { \raise-1.3cm\hbox{\epsfxsize=3.cm \epsfbox{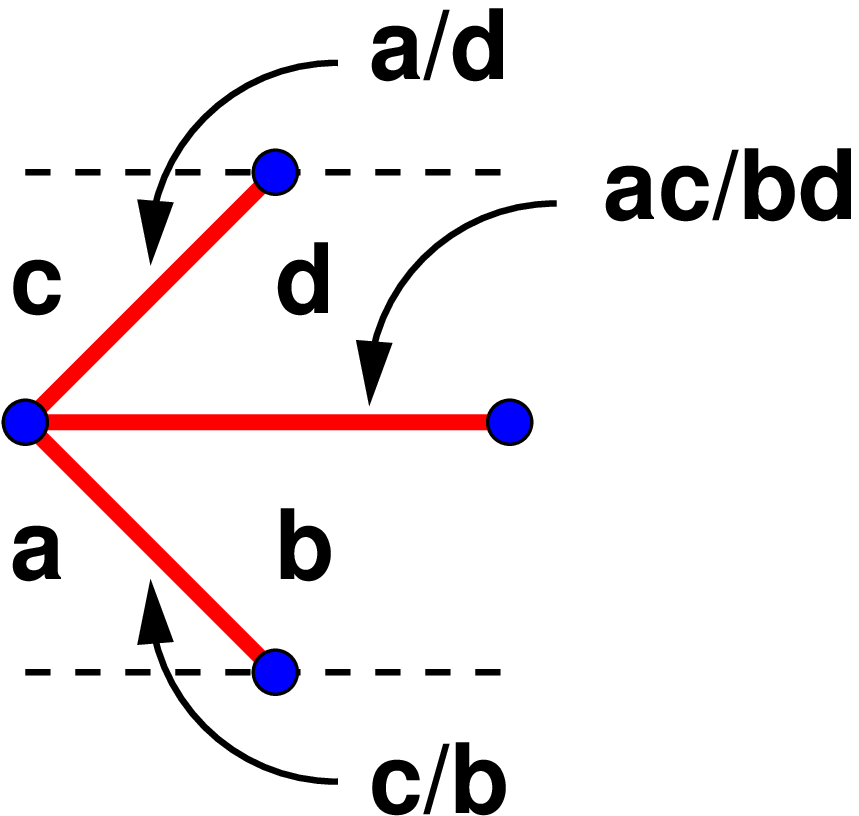} }}  
\end{equation}
where we have indicated the new edge weights (the dashed lines all come with a trivial weight $1$).
This allows to simplify the network for the domain $\theta_{\rm min}(k)$, by embedding it into a
triangular lattice with edge lengths $2$ and $\sqrt{2}$ as indicated in Fig. \ref{fig:triang} (a).
The expression for $T_{i,j,k}$ of Theorem \ref{solflat} is up to the usual prefactor the partition
function of $k-1$ non-intersecting paths from entry to exit points $i,i-1,...,i-k+2$ on the simplified network above.
We have represented a configuration of such paths in Fig.\ref{fig:triang} (b).

\begin{figure}
\centering
\includegraphics[width=15.cm]{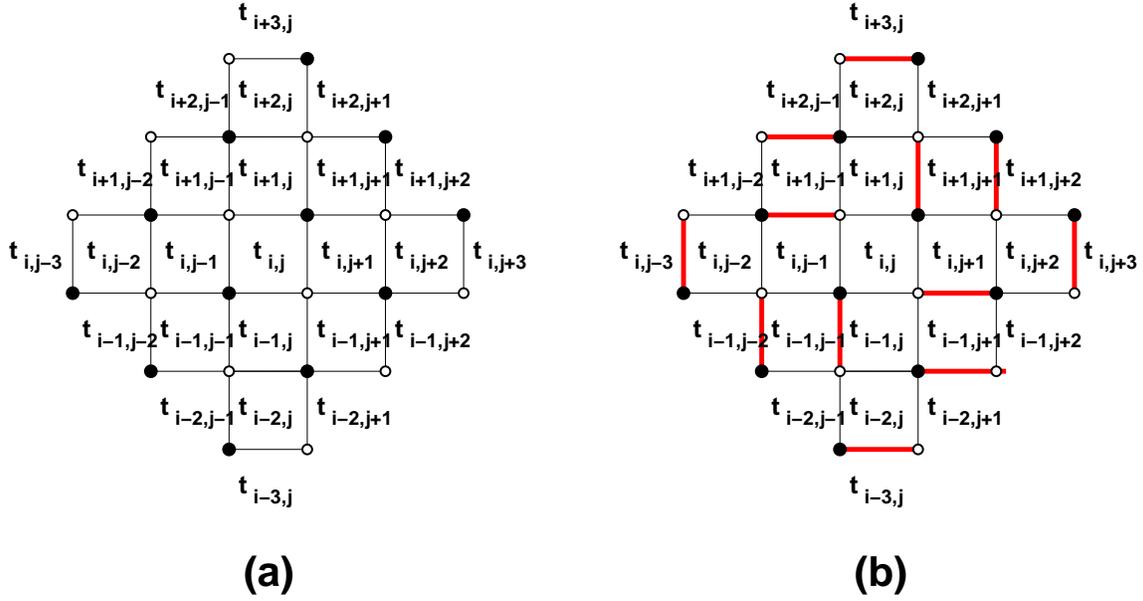}
\caption{The dual $\theta_{\rm min}(4)^*$ of $\theta_{\rm min}(4)$ (a) and its dimer configuration (b) corresponding to the domino 
tiling of Fig.\ref{fig:triang} (c). We have indicated the face labels corresponding to $I(\bk_0,\bt)$.
The dimer configuration has the total weight: 
$t_{i-2,j-1}t_{i-1,j-1}^{-1}t_{i-1,j+1}^{-1}t_{i-1,j+2}t_{i,j}t_{i+1,j-2}t_{i+1,j-1}^{-1}t_{i+1,j+1}^{-1}t_{i+2,j+1} $.}\label{fig:diamondim}
\end{figure}

We may finally interpret the path configurations as domino tiling configurations of the Aztec diamond, by use
of the bijection \eqref{futiles}, as illustrated in Fig.\ref{fig:triang} (c). We note that the underlying chessboard
bicolored square lattice in Fig.\ref{fig:triang} (c) is nothing but the square-triangle tessellation $\theta_{\rm min}(k)$
of the flat surface $\bk_0$ as depicted in Fig.\ref{fig:thetamin} (a), with the missing halves of the boundary squares. 
The domino tiling is equivalent
to a dimer covering of the (vertex-bicolored) dual $\theta_{\rm min}(k)^*$ of this tessellation.  
The latter is made of squares with bicolored
vertices, whose face labels are the assigned values $t_{i,j}$ of $I(\bk_0,\bt)$ (see Fig.\ref{fig:diamondim} (a-b) for
an example, the dimer configuration being equivalent to the domino tiling of Fig.\ref{fig:triang} (c)).

As before, we define weights of the dimer configurations as follows:
(i) a weight $b^{1-D}$ per face of the graph $\theta_{\rm min}(k)^*$ with label $b$
and whose adjacent edges are occupied by $D$ dimers (ii) a weight $a^{1-\delta}$ for each external boundary label
whose adjacent edges (2 for a corner, 1 for a vertical or horizontal single edge) are occupied by $\delta$ dimers.
The partition function for dimers on $\theta_{\rm min}(k)^*$ is the sum over all dimer configurations of the product 
of these weights.
The following theorem was first obtained by Speyer \cite{SPY}. We give an alternative proof based on the previous
constructions, which serves as a warmup for later sections.

\begin{thm}
The solution $T_{i,j,k}$ of the $T$-system is the partition function of dimers on the graph $\theta_{\rm min}(k)^*$.
\end{thm}
\begin{proof}
We simply have to check that the weights of the path model produce the correct dimer weights. Let us concentrate on
the dependence of the path model weights on a given inner label $b$. By inspection,
due to the rules \eqref{rultri}, we see that $b$
may appear with powers $0,\pm 1$ according to the following situations:
$$ \hbox{\epsfxsize=15.cm \epsfbox{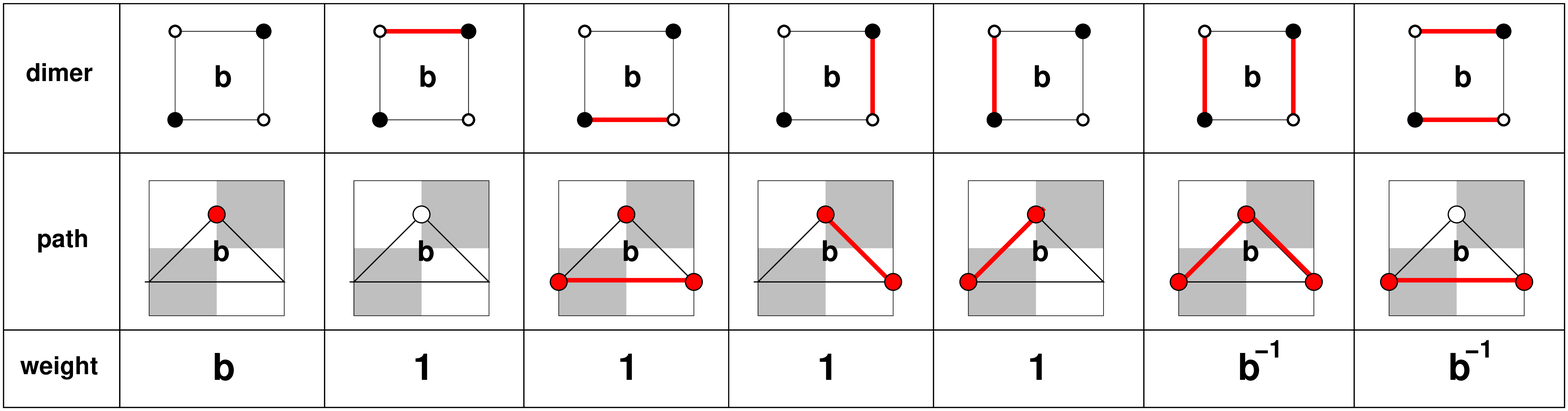}}$$
where we have represented by filled red circles the vertices of the triangle visited by a path
and by empty circles those {\it not} visited.
A similar table holds for down-pointing triangles. This produces the weight $b^{1-D}$ per
square face of $\theta_{\rm min}^*(k)$. For boundary labels, we must distinguish the four corners and the four single-edge
boundaries, denoted respectively by $NW,NE,SW,SE$ and $N,S,W,E$ with the obvious meaning. For the single-edge
boundaries, we have the following weights from the path configuration:
$$ \hbox{\epsfxsize=16.cm \epsfbox{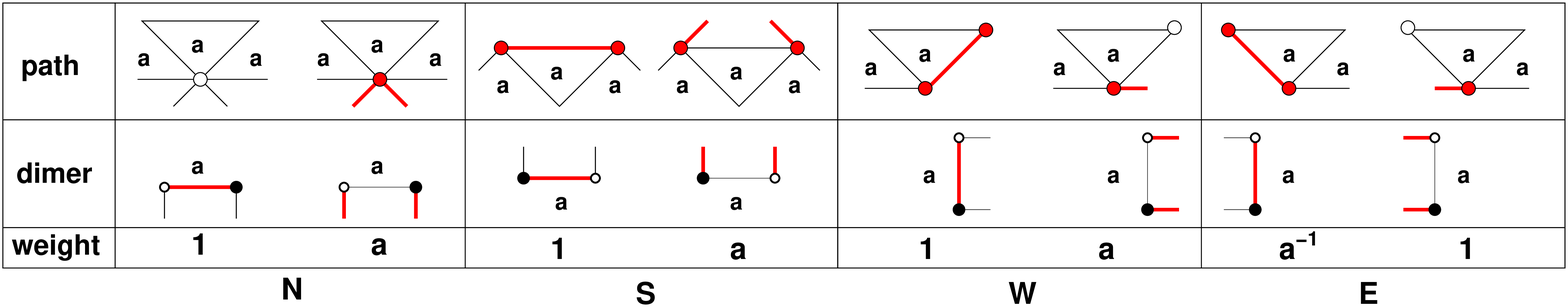}}$$
where we have also represented the corresponding single-edge dimer configuration, and indicated the contribution 
of the path weight. We note that the weights are compatible with the formula $a^{1-\delta}$ where $\delta$ is the number
of dimers on the single edge, for all cases but the $E$ one. However, in that case, we may take the prefactor $a=t_{i,j+k-1}$
in \eqref{tsolflat} and absorb it into a redefinition of the weight, which fixes the case $E$. Finally, the weights for the four
types of corners read:
$$  \hbox{\epsfxsize=13.cm \epsfbox{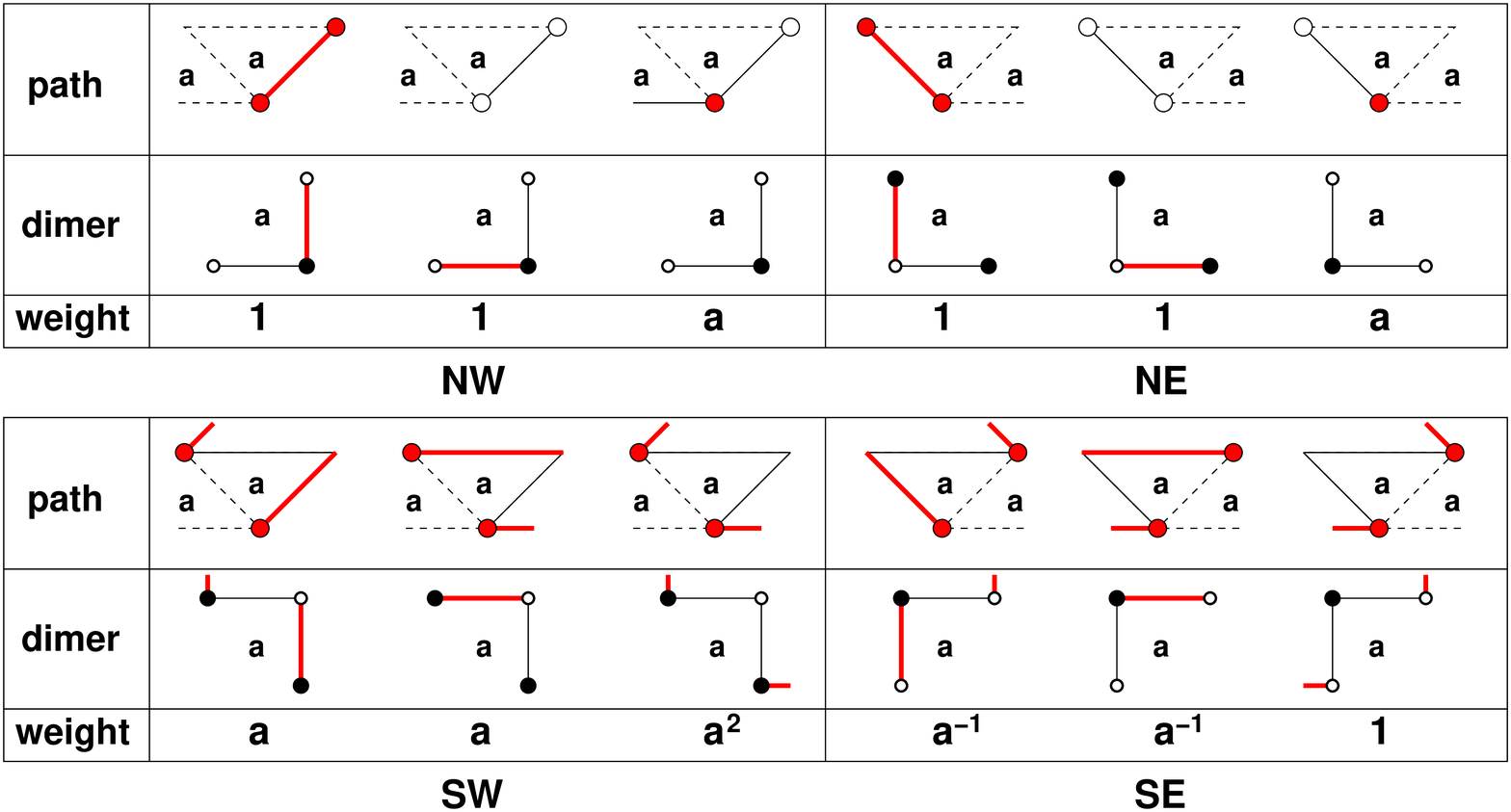}}$$
Again, the formula $a^{1-\delta}$ reproduces correctly the weight only for the $NW$ and $NE$ cases. However,
we may absorb into the weights the factors $a^{-1}=t_{m+i-k,j+1-m}^{-1}$ (resp. $a=t_{p+i-k,j-1+p}$) 
coming from the prefactor in \eqref{tsolflat} to fix the discrepancy
in the $SW$ (resp. $SE$) cases. The theorem follows.
\end{proof}
\subsection{General initial data and dimer models on 4-6-8 graphs}
This section is a generalization of the previous one to the case of the solution of the $T$-system
with an arbitrary initial data assignment $I(\bk,\bt)$. Our starting point is Theorem \ref{solgen}.

We first associate to the shadow ${\mathcal D}^\circ$ it dual, vertex-bicolored graph ${\mathcal D}^*$.
The vertices are colored white or black according to the color of their dual face (white or gray). 
For clarity, we choose to represent on the line $x=i+\frac{1}{2}$
all vertices corresponding to triangles or squares that belong to the same horizontal strip $x\in[i,i+1]$.
With this choice, we may only have inner faces of ${\mathcal D}^*$ that are squares, hexagons or octagons, with vertices
on two consecutive horizontal lines of the form $x=i\pm \frac{1}{2}$ (and with at least two vertices on each line),
and all the edges joining the two lines are represented vertical.
A given inner vertex of ${\mathcal D}^\circ$ may indeed be shared by: (i) 8 triangles (ii) 2 squares and 4 triangles, or
(iii) 4 squares, with alternating colors around the vertex. Finally we label each face of ${\mathcal D}^*$ with the initial data
assignment of the dual vertex. This includes external labels, which label external regions (external ``faces")
separated by horizontal dashed lines. 
Moreover, we erase any vertex of ${\mathcal D}^*$ that is unique on its horizontal line, and replace the vertical edge
connecting it to a neighboring line by a dashed line, to indicate that it separates two regions of distinct labels.
We call such graphs 4-6-8 graphs. As an illustration, the 4-6-8
graph dual to the shadow depicted in Fig.\ref{fig:shadow} (b) is:
$$  \hbox{\epsfxsize=13.cm \epsfbox{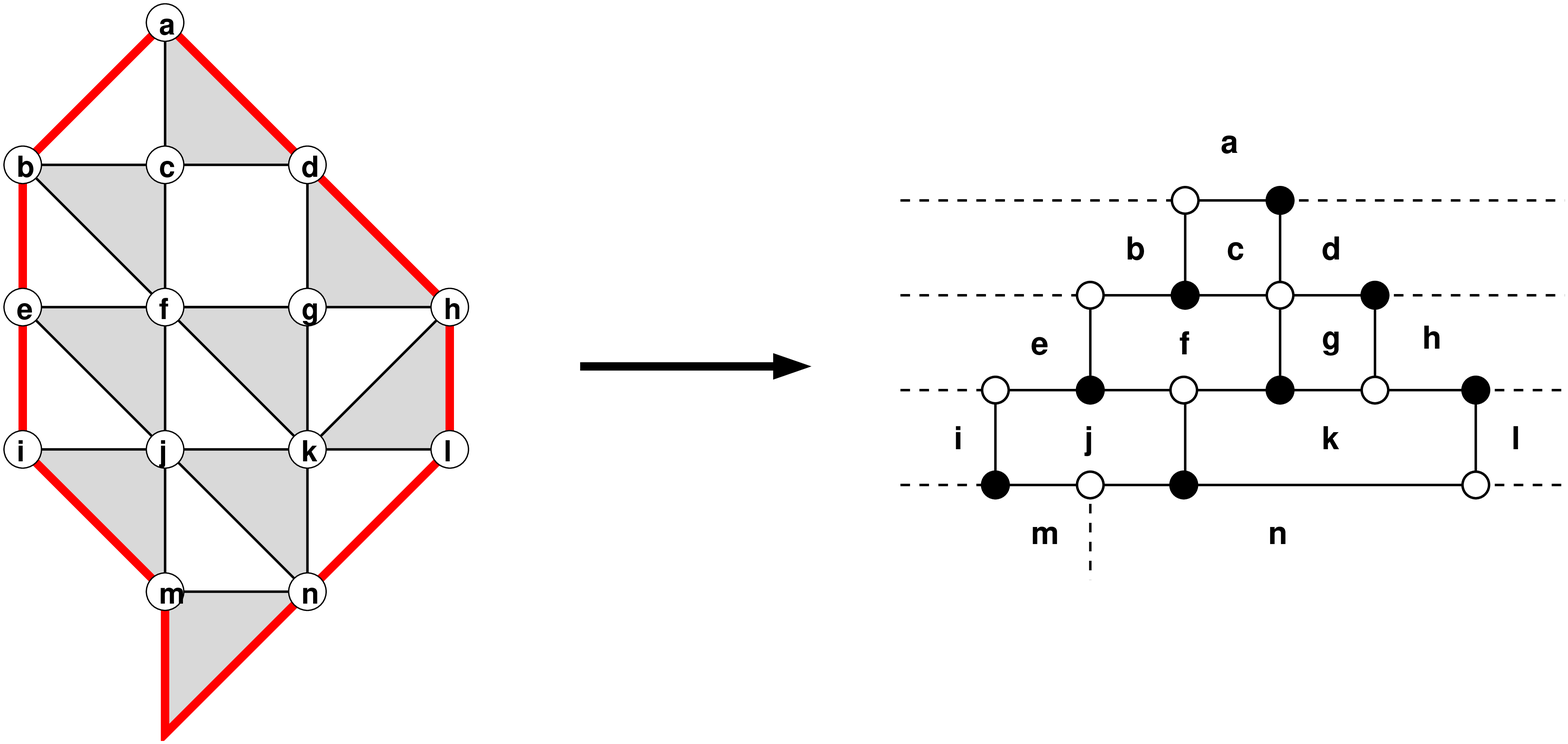}}$$
(We have indicated the vertex labels on the original tessellated initial data surface ${\mathcal D}^\circ$
and the corresponding face labels on the dual ${\mathcal D}^*$.).

We may consider the partition function for dimers on the 4-6-8 graph ${\mathcal D}^*$, defined as follows.
For any face $F$ of ${\mathcal D}^*$, we define the valency of $F$, denoted
by $v(F)$, to be its degree ($v(F)\in \{4,6,8\}$ if it is an inner face of ${\mathcal D}^*$, and the number of its non-dashed 
adjacent edges if it is external ($v(F)\in \{1,2\}$).
Let  $F$ be a face with label $a$, with exactly $D$ adjacent edges occupied by dimers.
We define the face weight $w_F(a)$ to be:
\begin{equation}\label{weightdimer}
w_F(a)=\left\{ \begin{matrix} a^{\frac{v(F)}{2}-1-D} & {\rm if} \, F\, {\rm is}\, {\rm inner} \\
a^{1-D} & {\rm otherwise} \end{matrix}\right. 
\end{equation}
The partition function for dimers on ${\mathcal D}^*$ is defined as usual as the sum over all dimer configurations
on ${\mathcal D}^*$ of the product of all face weights.
We may now state our main result.

\begin{thm}\label{gendimth}
The solution $T_{i,j,k}$ of the $T$-system \eqref{tsys} with initial conditions $I(\bk,\bt)$ \eqref{infinitdata}
is the partition function of dimers on the dual ${\mathcal D}^*$ of the shadow ${\mathcal D}^\circ$ of
the point $(i,j,k)$ onto the initial data stepped surface.
\end{thm}
\begin{proof}
We start from the network interpretation of the formula of Theorem \ref{solgen}. The 
minor $\left\vert M_{{\mathcal D}^\circ}\right\vert_{i-\ell+1,i-\ell+2...,i-1,i}^{i-\ell+1,i-\ell+2,...,i-1,i}$ in 
Theorem \ref{solgen} is the partition function for configurations of $\ell$ non-intersecting paths joining
the $\ell$ bottom-most left entry points to the $\ell$ bottom-most right exit points of the network corresponding
to $M_{{\mathcal D}^\circ}$. Recall that this matrix is a product of $U,V$ matrices according to the lozenge
decomposition of ${{\mathcal D}^\circ}$. We may directly connect the product of $U,V$
matrices to the dual graph ${{\mathcal D}^*}$, by first bijectively associating respectively the matrices $U_i,V_i$ to 
single vertical edges connecting lines $x=i-\frac{1}{2}$ and $x=i+\frac{1}{2}$ with respectively a black vertex on top, bottom
as follows:
\begin{equation}\label{ddnet}
U_i(a,b,u)\quad \to \quad 
\raisebox{-.6cm}{\hbox{\epsfxsize=1.2cm \epsfbox{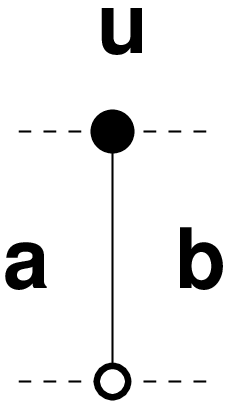}}}
\qquad V_i(v,a,b)\quad \to \quad
\raisebox{-.9cm}{\hbox{\epsfxsize=1.2cm \epsfbox{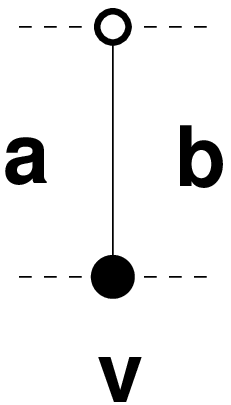}}}
\end{equation}
and any two consecutive vertices of the same color along a horizontal line are identified, so that for instance we get
\begin{equation}\label{otnet}
U_i(a,b,u)V_{i+1}(u,v,c) \to \raisebox{-1.4cm}{\hbox{\epsfxsize=1.2cm \epsfbox{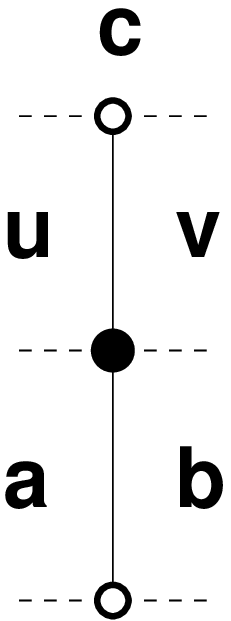}}}
\end{equation}
In turn, each matrix element of $U,V$ is the weight for a path step on the corresponding network. More precisely,
the local configuration of path on $V,U$ chips of network determines uniquely the dimer configuration of the edges
adjacent to the black vertex, as follows:
$$  \hbox{\epsfxsize=15.cm \epsfbox{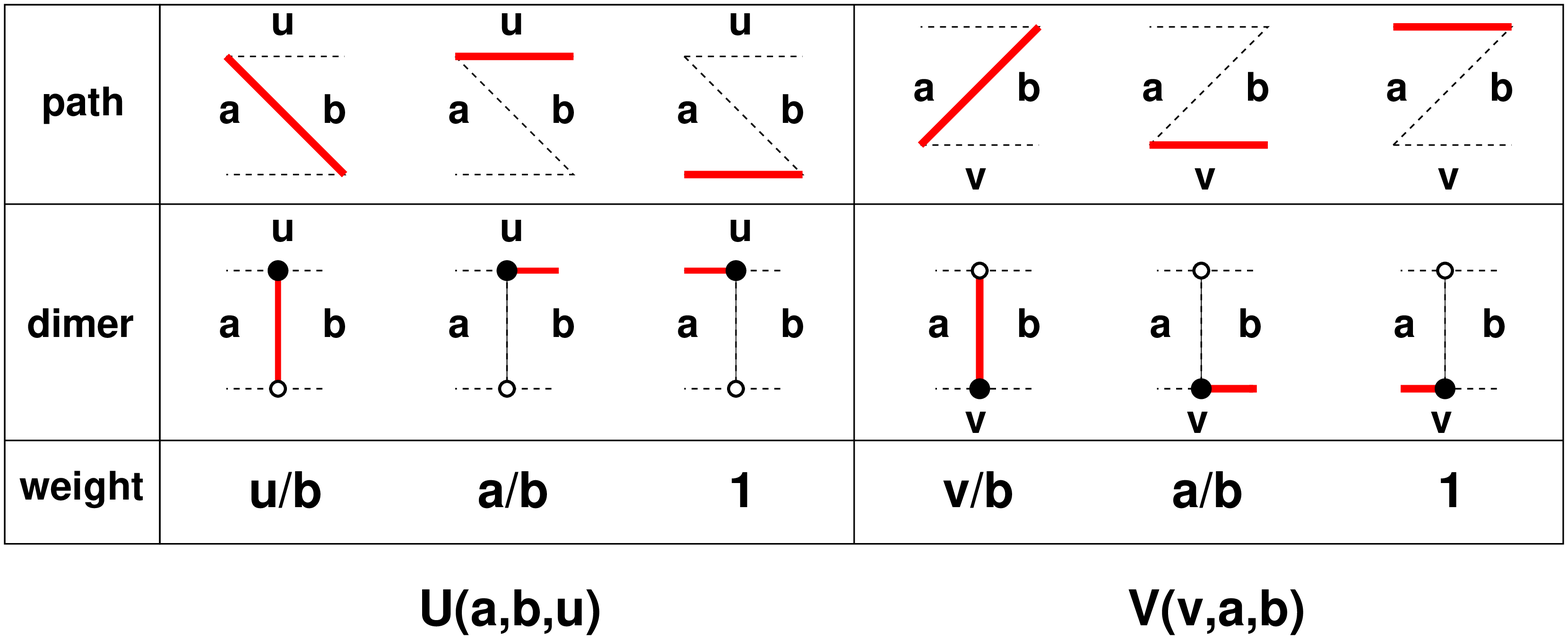}}$$
This  allows us to redistribute the step weights of the paths on the dimer configurations in the following way.
We attach a weight to each pair formed by a black vertex and an adjacent face $F$ with label $a$, depending
on whether the vertex is (i) at a  corner of the face ($c_{NW},c_{SW},c_{NE},c_{SE}$), (ii) at
the junction $h_T,h_B$ between two horizontal edges, along the top 
or bottom border of the face, (iii) on the single left or right vertical edge of an external face ($e_L,e_R$),
or simply (iv) on a top or bottom ($e_T,e_B$) edge of an external face.
Inspecting the above table, and denoting by $D$ the number of dimers ($\in \{0,1\}$)
that are adjacent to the black vertex, we get weights:
\begin{eqnarray*}w(c_{NW})&=&w(c_{SW})=a^{-D},\quad w(c_{NE})=w(c_{SE})=a^{1-D},\quad 
w(h_T)=w(h_B)=a^{1-D},\\
w(e_L)&=&a^{-D}, \quad w(e_R)=a^{1-D}, \quad w(e_T)=w(e_B)=a^{1-D}
\end{eqnarray*}
If $F$ is an inner face, the product of the weights from the various black vertices of $F$ is $a^{N_\bullet-1-D}$,
where $N_\bullet=v(F)/2$ is the total number of black vertices adjacent to $F$,
in agreement with \eqref{weightdimer}. Indeed, exactly one of the two corners $c_{NW}$ and $c_{SW}$ is a black vertex,
and similarly for $c_{NE}$ and $c_{SE}$.
If $F$ is an external face, we must distinguish if it has a single vertical adjacent edge, in which case $w(e_R)$
reproduces \eqref{weightdimer}, but $w(e_L)=a^{-D}$ has a factor of $a$ missing. However, as in the flat case, the weight can
be corrected by borrowing the weight $a$ out of the prefactor of \eqref{tsolgen}. If $F$ is an external face with only top
or bottom edges, then $w(e_T),w(e_B)$ reproduce the weight \eqref{weightdimer}. Finally, if $F$ is an external face
with a corner-like set of edges, we find by inspection the following weights:
\begin{eqnarray*} w\left( \raisebox{-.5cm}{\hbox{\epsfxsize=1.5cm \epsfbox{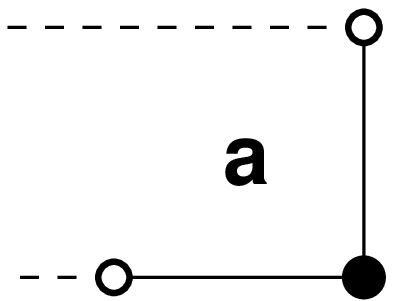}}}\right) &=& a^{1-D},\quad w\left( 
\raisebox{-.5cm}{\hbox{\epsfxsize=1.5cm \epsfbox{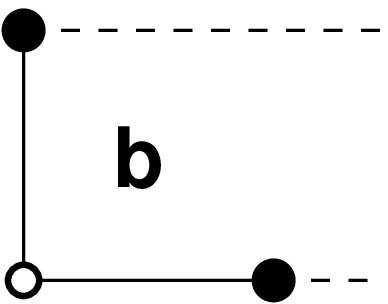}}}\right) = b^{1-D} \\
w\left(\raisebox{-.5cm}{\hbox{\epsfxsize=1.5cm \epsfbox{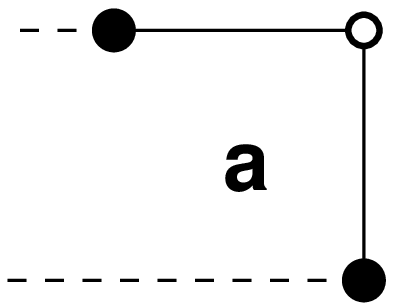}}} \right) &=& a^{2-D},\quad w\left( 
\raisebox{-.5cm}{\hbox{\epsfxsize=1.5cm \epsfbox{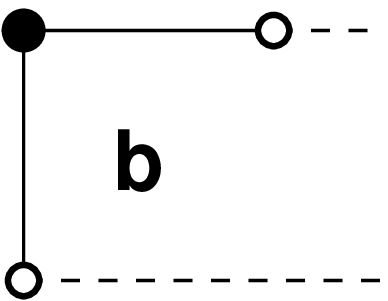}}}\right) = b^{-D} \end{eqnarray*}
hence as in the flat case, we must borrow the factors $a^{-1}$ and $b$ from the prefactor in \eqref{tsolgen}
to fix the last two corner weights. The theorem follows.
\end{proof}

\section{Crystal melting and Yang-Baxter equation}

\subsection{V,U matrices and the Yang-Baxter equation}

In the previous sections, we have expressed the $T$-system octahedron relation 
as a sort of flatness condition on some $GL_2$ connection defined on the stepped
surface that supports the initial data of the system, leading to an explicit formula
for the solution in terms of general initial conditions.

The same equation may be obtained as a braiding condition that generalizes the 
Yang-Baxter equation in the context of networks. We have the following lemma, 
easily proved by direct calculation:

\begin{lemma}\label{ybe}
The following ``braiding" relation in $GL_3(\mathcal A)$:
$$ V_{1}(u,a,b)V_{2}(b,c,d)V_{1}(u,b,e)=V_{2}(a,c,b')V_1(u,a,e)V_2(e,b',d)  $$
holds if and only if the octahedron condition
$$ b b'= ec+ad $$
is satisfied.
\end{lemma}

\begin{remark}
In the lozenge picture of the previous sections, this expresses two different ways of
decomposing an hexagon into V-type lozenges (with the gray triangle on the bottom):
\begin{equation}\label{hexadec}
\hbox{\epsfxsize=8.cm \epsfbox{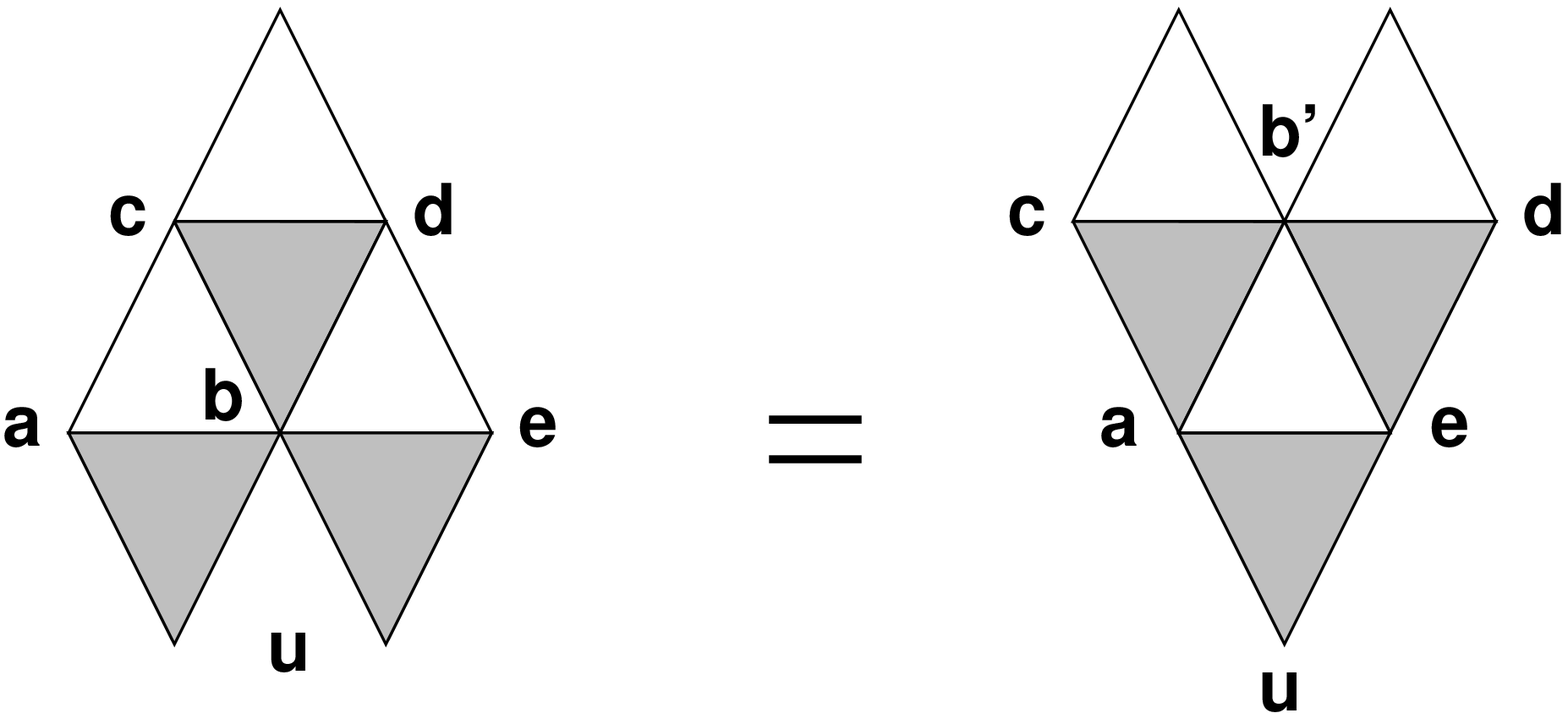}}
\end{equation}
\end{remark}

\begin{remark}
A similar relation holds for $U$ matrices, in which case all lozenges of the previous remark
must have white and gray triangles switched.
\end{remark}

\begin{remark}
The network formulation of the above relation reads:
$$ \hbox{\epsfxsize=8.cm \epsfbox{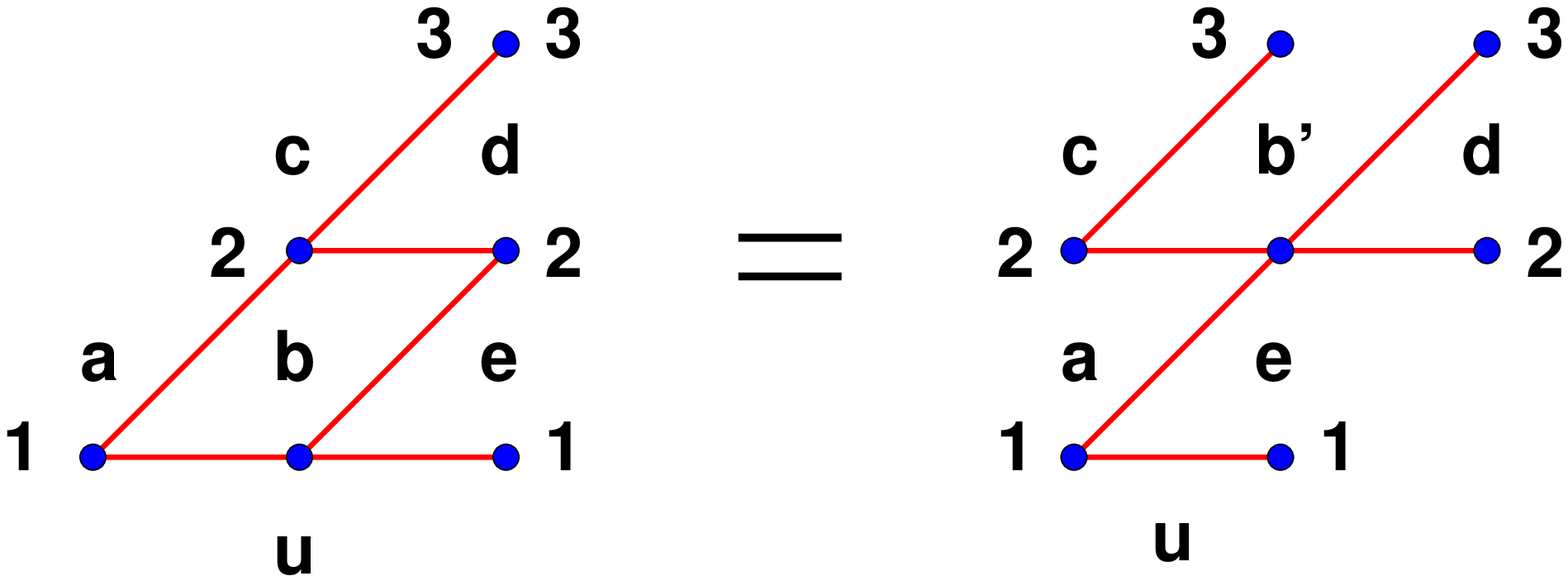}}$$
\end{remark}

It is interesting to try to relate Lemmas \ref{octamove} and \ref{ybe}. We may compute the product in
Lemma \ref{ybe} as follows:
$$ \hbox{\epsfxsize=15.cm \epsfbox{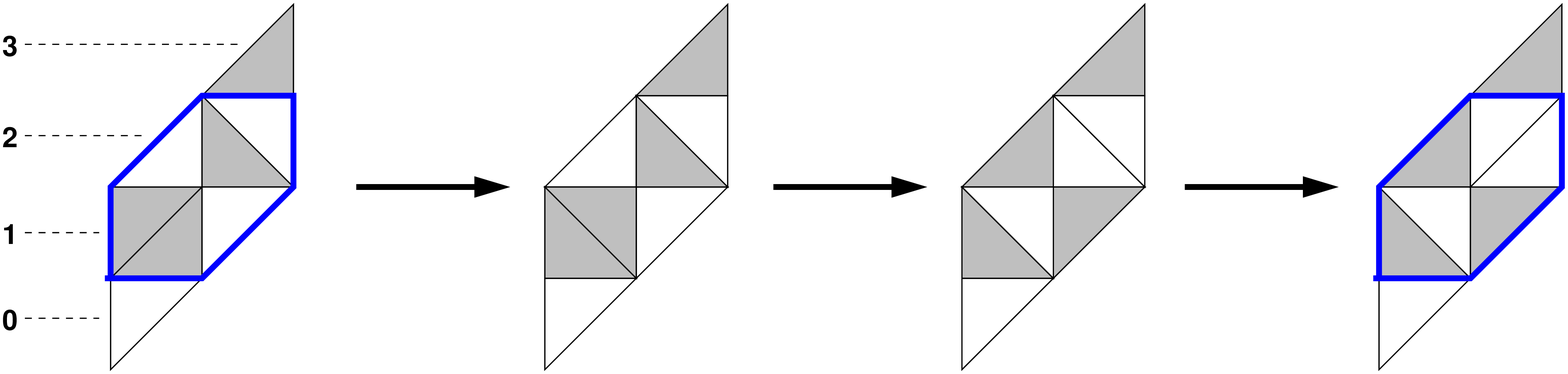}} $$
where in the first and last step we have flipped the diagonal in a unicolor square
(using Lemma \ref{tetradec}), and where the central identity is a direct application of Lemma \ref{octamove},
upon substituting the central vertex label $b\to b'$. 
In matrix terms, using the obvious embedding
into $GL_4$ (with matrix indices $0,1,2,3$ as indicated on the figure), this reads:
$$V_1(u,a,b)U_0(u,e,b)U_1(b,d,c)U_2(c,f,g)=U_0(u,e,a)U_1(a,b',c)U_2(c,f,g)V_1(e,b',d) $$
If we concentrate on the inner hexagon transformation,
we recognize the relation \eqref{hexadec}, up to a global rotation by $\pi/4$ and the straightening of some edges.

This expresses that the lozenge configurations can be read in different ways/directions. In other words, different matrix
products can be attached to given lozenge configurations, corresponding to different directions in which the matrix indices
are chosen. This is even more transparent in the network language.

Let us consider the $GL_3$ embedding of the matrix $U_1(a,b,c)U_2(c,c,c)$, with indices $1,2,3$:
$$U_1(a,b,c)U_2(c,c,c)=\begin{pmatrix} 1 & 0 & 0 \\ \frac{c}{b} & \frac{a}{b} & 0\\
0 & 1 & 1 \end{pmatrix}$$
We may ``read" the matrix differently by only focussing on the submatrix with row indices $2,3$ and column
indices $1,2$. In network terms, we consider the subgraph of the initial network for the $GL_3$ embedding
with entry points $2,3$ and exit points $1,2$:
$$ \hbox{\epsfxsize=8.cm \epsfbox{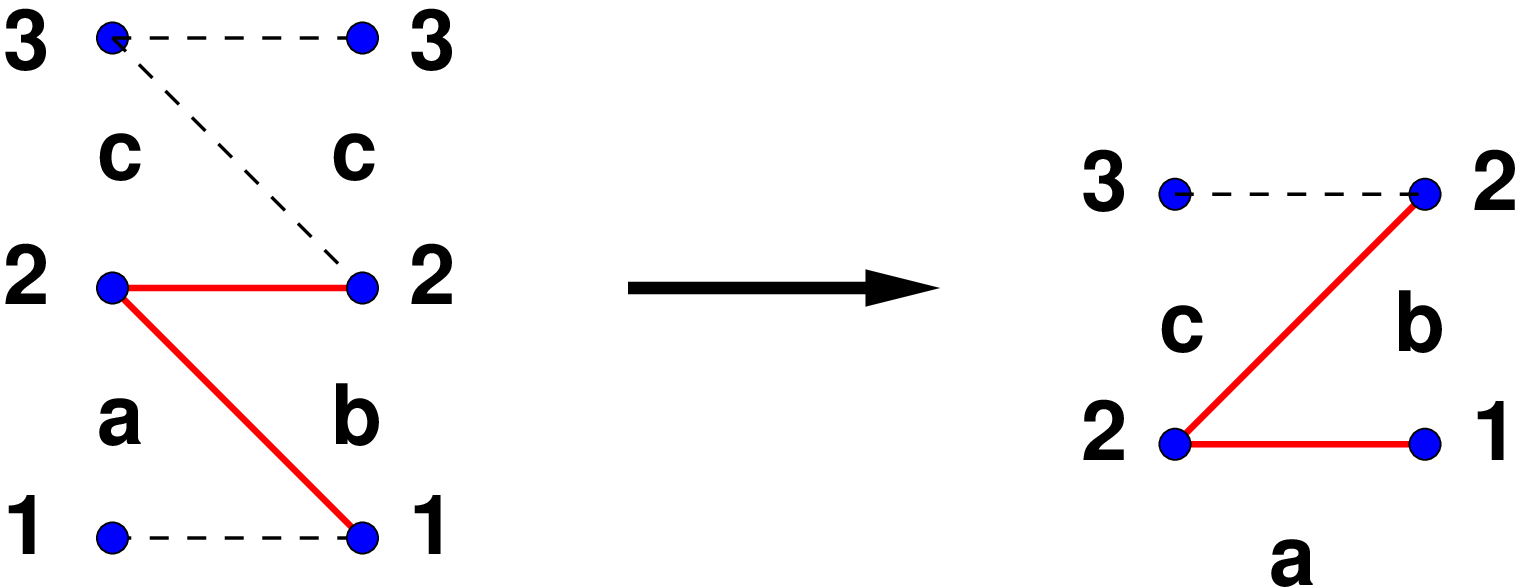}} $$
This amounts to reading the network in a different direction $2\to 1$, $3\to 2$. Remarkably, the second reading 
simply corresponds to the $2\times 2$ matrix $V(a,c,b)$.  So the network definitions allow to read the same
graph in various directions, which may imply that some chips formerly read as $U,V$'s may be read as $V,U$'s.
The choice of direction amounts to particular choices of entry/exit points on the network.

%This become even more obvious in the dimer language. 

\subsection{The octahedron equation on a cube corner}

\begin{figure}
\centering
\includegraphics[width=8.cm]{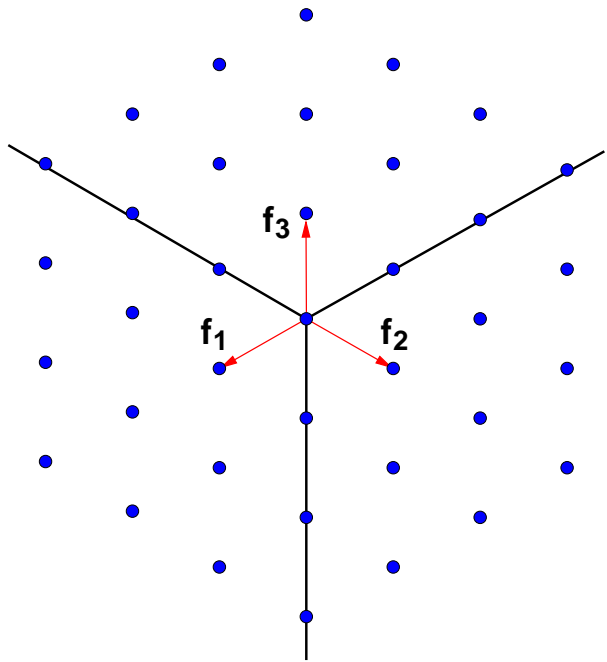}
\caption{The cube corner with apex $(a,b,c)$, represented in the projection $\pi$.}\label{fig:corcub}
\end{figure}

Instead of considering the ordinary $T$-system as a discrete $2+1$ dimensional evolution
equation with initial data $I(\bk,\bt)$, we may consider the same equation as the
evolution of the infinite surface of the corner of an infinite cube in $\Z^3\langle \vec{e}_1,\vec{e}_2,\vec{e}_3\rangle$ 
pointing in the
$(1,1,1)$ direction, say with apex at $(a,b,c)$ for some
$a,b,c\in \Z$ (see Fig.\ref{fig:corcub} for an illustration). 
The surface of the cube corner is made of three infinite quarter planes of square lattice
sharing the vertex $(a,b,c)$,
with respective equations $(x=a,y\leq b,z\leq c)$ and its two circular permutations.
The orthogonal projection $\pi$ of the cube surface onto the plane $(1,1,1)^\perp$
is a regular triangular lattice say with basis vectors $\vec{f_1}=\pi(\vec{e_1})=\frac{2\vec{e_1}-\vec{e_2}-\vec{e_3}}{3}$
and 
$\vec{f}_2=\pi(\vec{e}_2)=\frac{2\vec{e}_2-\vec{e}_1-\vec{e}_3}{3})$. We also define 
$\vec{f}_3=\pi(\vec{e}_3)=-\vec{f}_1-\vec{f}_2$, and fix the projection of the apex to be the origin of the projection plane,
represented in Fig.\ref{fig:corcub}.

The evolution of the cube surface is by evaporation (melting) of unit cubes
off the initial cube, in such a way that the resulting surface remains stepped. More precisely, for any unit cube corner
of the form $(a+1,b,c),(a+1,b+1,c),(a,b+1,c),(a,b+1,c+1),(a,b,c+1)(a+1,b,c+1)$ with apex $(a+1,b+1,c+1)$, the melting
replaces the apex with the bottom vertex $(a,b,c)$ that completes the cube. To each such evaporation, we attach
an evolution equation
\begin{equation}\label{cuboc} 
\theta_{a,b,c}\theta_{a+1,b+1,c+1} =\theta_{a+1,b,c}\theta_{a,b+1,c+1}+\theta_{a,b+1,c}\theta_{a+1,b,c+1} 
\end{equation}
for some variable $\theta_{a,b,c}$ defined at the vertices $(a,b,c)$ of the stepped surface. 
We can think of time as the component along
$(1,1,1)$ (equal to $a+b+c$ in our conventions) so that the evaporation process goes back in time by 3 units.
Note that Eqn.\eqref{cuboc} is identical to the so-called cube equation \cite{CS}, 
without the term $\theta_{a,b,c+1}\theta_{a+1,b+1,c}$.

Writing this equation in the projection $\pi$ leads to the following change of variables:
$T_{i,j,k}=\theta_{a,b,c}$ where $i=a-c$, $j=b-c$ and say $k=a+b-1$ so that $i+j+k=1$ mod 2. 
This clearly reduces \eqref{cuboc} to the $T$-system \eqref{tsys}. 
Moreover, the process of cube evaporation is exactly described by the identity of Lemma \ref{ybe}, or pictorially
via \eqref{hexadec}. More precisely, the latter expresses the passage:
\begin{equation}\label{evapo}   \hbox{\epsfxsize=8.cm \epsfbox{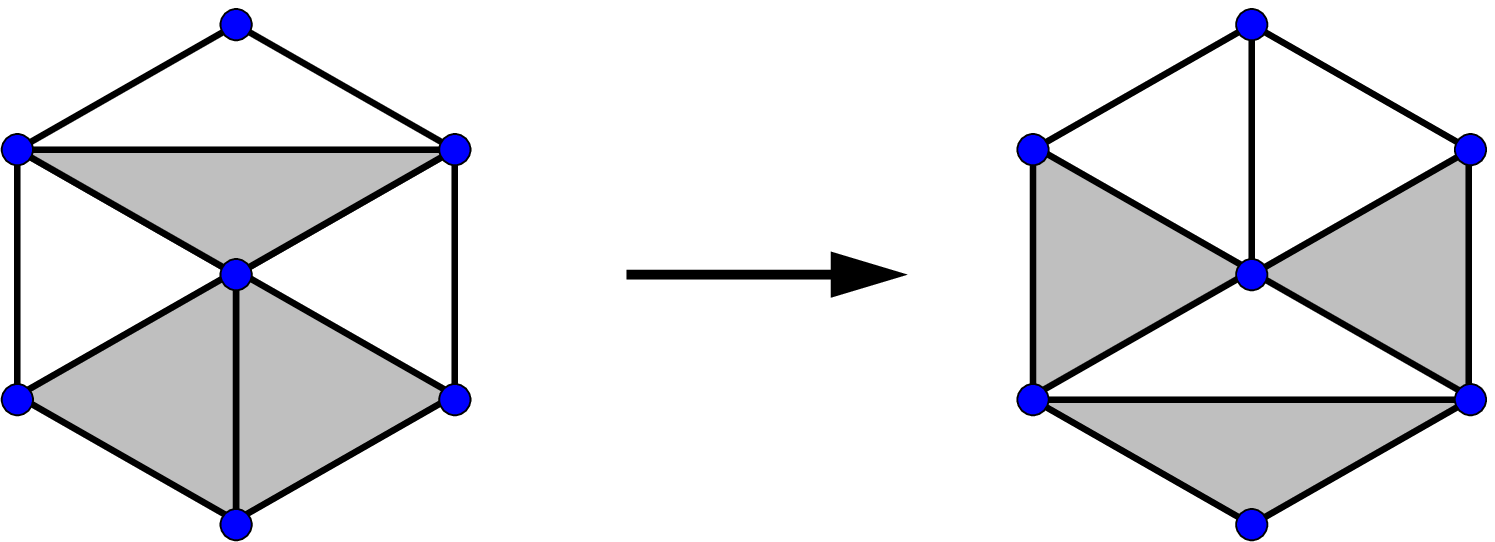}}  \end{equation}
for $bb'=ac+de$, which indeed corresponds to the evaporation of a cube:
$$  \hbox{\epsfxsize=8.cm \epsfbox{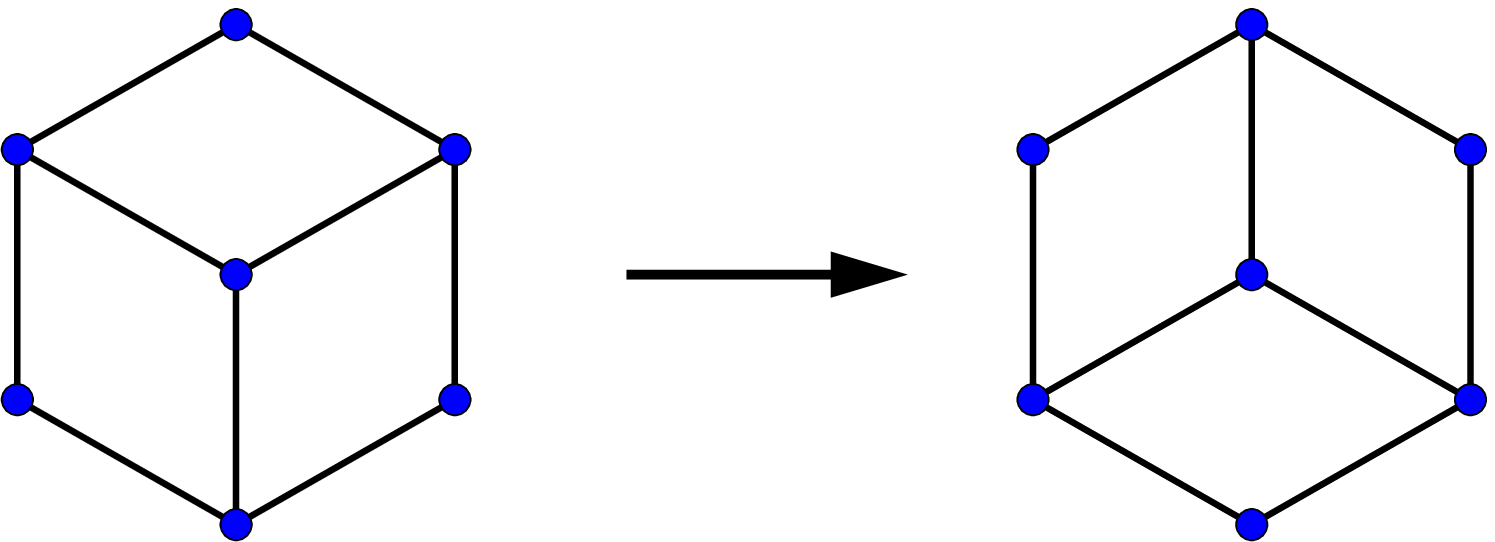}}   $$
Note that we may view the cube surface as a stepped surface above the projection plane:
$\vec{u}=x\vec{f}_1+y\vec{f}_2\, \mapsto \, z_{\vec{u}}$, where the ``time" $z_{\vec{u}}$ satisfies the step conditions:
$|z_{\vec{u}\pm \vec{f_i}}-z_{\vec{u}}|=1$ for all $i=1,2,3$, and $z_{\vec{0}}$ is the time distance from the apex to 
the projection plane (fixed arbitrarily). 
The three faces of the infinite cube corner project respectively onto three positive cones of the form 
$\lambda (\vec{f}_i+\vec{f}_j)+\mu (\vec{f}_i+\vec{f}_k)$ for $\lambda,\mu\in \Z_+$ and $(i,j,k)\in \{(1,2,3),(2,1,3),(3,1,2)\}$.

\subsection{Solution via $V$ matrices and networks}

We may reverse the above by starting with some arbitrary stepped surface (by a slight abuse of language 
we still call stepped surface here an arbitrary evaporated configuration 
of the cube corner), and let it evolve by the ``unit cube addition"
process, inverse of the evaporation, which then replaces the bottom vertex of the unit cube with the top one. 

Starting from the ``flat" initial data\footnote{The reader should not confuse this term with its meaning in the 
previous section: the notion of flatness is relative to the geometry of the underlying cubic lattice here, and the ``flat" 
surface is perpendicular to the direction $(1,1,1)$. The term ``flat" is used with this meaning throughout this section.
Whenever ambiguous, we will refer to the flat initial data of the cubic lattice, as opposed to the flat initial data of the CC lattice.} 
that corresponds to intersecting the lattice $\Z^3$ with three consecutive planes 
perpendicular to $(1,1,1)$, let us consider the evolution of the corresponding stepped surface by ``unit cube addition".
More precisely, we have to solve for the quantity $\theta_{a,b,c}$ say for $a+b+c\geq 0$, obeying \eqref{cuboc},
in terms of initial data of the form:
\begin{equation}\label{indatet} \theta_{x,y,z}= \tau_{x,y,z} \qquad 
{\rm} \quad x,y,z,\in \Z^3\quad {\rm and}\quad x+y+z\in \{0,1,2\} 
\end{equation}
%This is equivalent to solving the $T$-system \eqref{tsys} for the quantity $T_{i,j,k}$ with initial data
%$t_{u,v}= \tau_{x,y,z}$
%along the surface
%$\bk$ defined by 
%$$ k_{u,v}= \frac{u+v-\epsilon_{u,v}}{3}-1 \qquad {\rm where} \quad \epsilon_{u,v}=\left\{ \begin{matrix}
%u+v \, {\rm mod}\, 3 & {\rm if}\,  u+v=0,2\, {\rm mod}\,  3 \\
%4 & {\rm if}\,  u+v=1\, {\rm mod}\,  3 
%\end{matrix}\right. $$
%and with $x=(k_{u,v}+1+u-v)/2$, $y=(k_{u,v}+1-u+v)/2$, and
%$z=(k_{u,v}+1-u-v)/2$.

The cube corner whose apex is the point  $(a,b,c)$ in $\Z^3$ intersects the flat initial data surface along 
a ``triangle", which by analogy with the previous sections could be called the ``shadow" of the point $(a,b,c)$
onto the initial data stepped surface. However, for technical reasons, we will need a larger domain
in order to produce a compact formula for $\theta_{a,b,c}$, although the elements added to the shadow
are purely spectator, and in particular never undergo cube additions. The inside of the cube corner surface
with equation $x\leq a,y\leq b,z\leq c$ intersects the flat stepped surface
along three triangles at times 0,1,2 respectively. 

Let $N=a+b+c-2$.
Via the projection $\pi$ we may decompose this 
intersection into ${N+1\choose 2}$ hexagons, which are the projection of bottom half cubes, 
ready for undergoing the cube addition process.
For $a+b+c=5$, hence $N=3$, this gives the following decomposition into ${4\choose 2}=6$ hexagons:
\begin{equation}\label{sha}  \hbox{\epsfxsize=6.cm \epsfbox{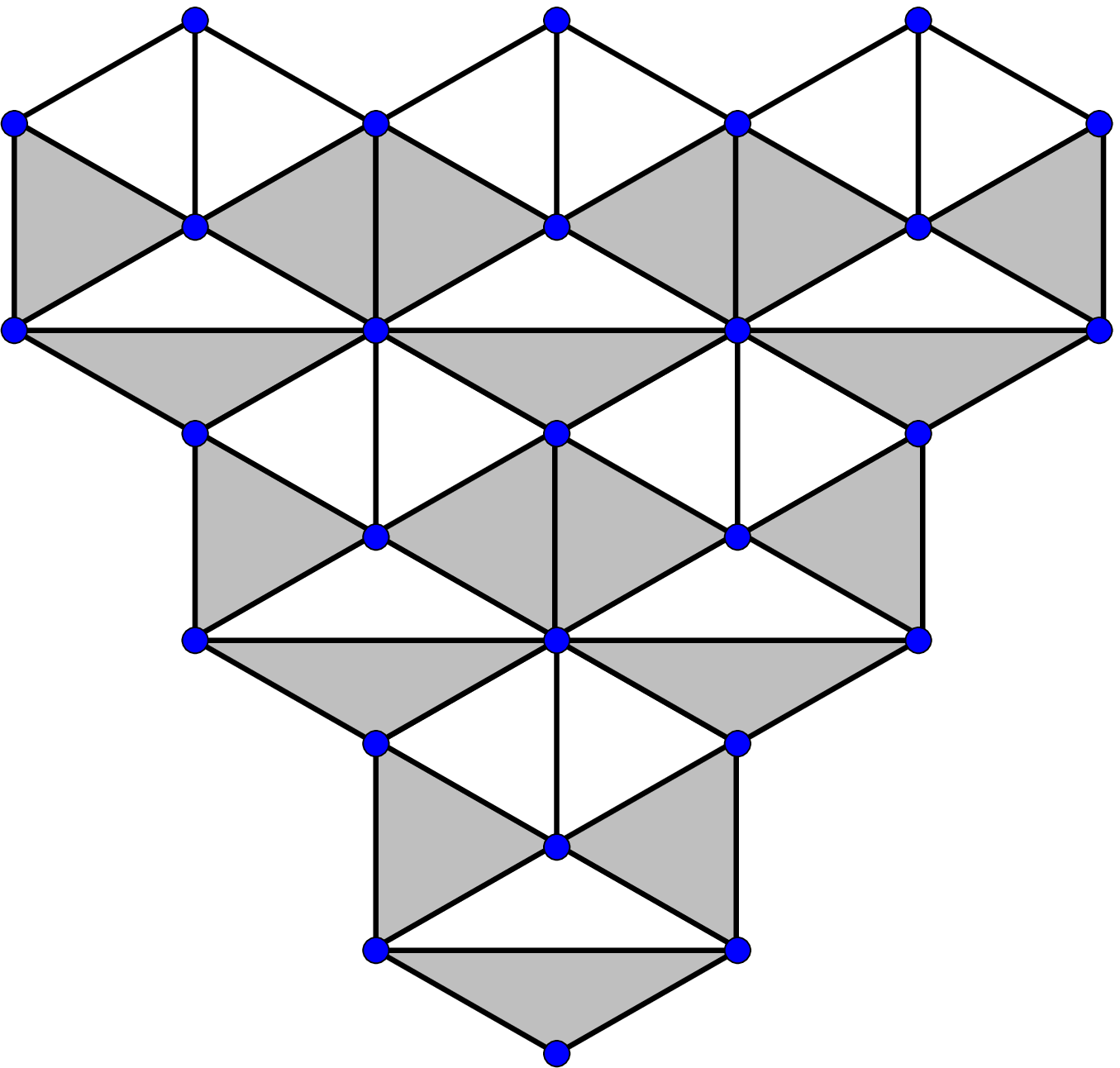}}  \end{equation}
where each hexagon is further decomposed into three lozenges of type $V$ as in the (evaporated) r.h.s. of 
\eqref{evapo}. The vertex at the center of the shadow is at time $\epsilon=2$ here, while the boundary vertices alternate
between times  $1$ and $2$ (independently of $\epsilon$). We further complete the top of the shadow by an
additional  ``triangle" made
of ${N\choose 2}$ $V$ type lozenges as follows:
\begin{equation}\label{exsha} \hbox{\epsfxsize=6.cm \epsfbox{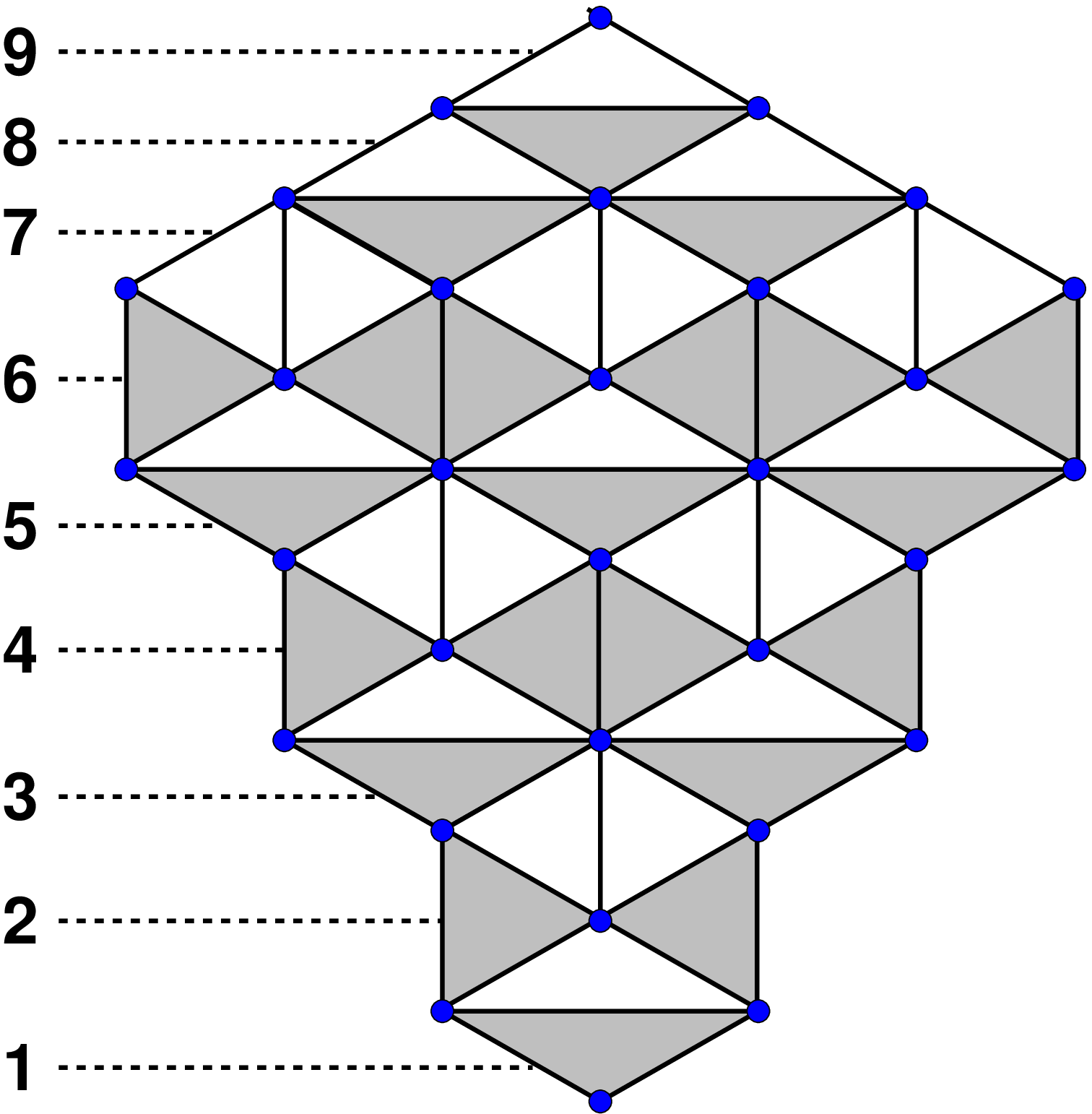}}  \end{equation}
The new domain is called the {\it augmented shadow} ${\mathcal S}_{a,b,c}$ of $(a,b,c)$. 
To this domain, naturally decomposed into ($V$-type) lozenges with the gray triangle on the bottom, we may associate
as before a product $M_{{\mathcal S}_{a,b,c}}$ of matrices $V_i(\alpha,\beta,\gamma)$ of $GL_{3N}$, 
with arguments $\alpha,\beta,\gamma$ equal to the prescribed initial data
$\tau_{x,y,z}$ around the gray triangle. For instance, each hexagon in the decomposition corresponds to a product
of the form $H_i=V_{i+1}V_iV_{i+1}$. The product corresponding to the augmented shadow of \eqref{exsha}
is (we drop arguments for simplicity):
$$M_{{\mathcal S}_{2,2,1}}= H_5 H_3 H_1V_7 H_5 H_3 V_8V_7 H_5  $$
corresponding to the natural labeling in the $GL_9$ embedding.

There is a natural network formulation of the matrix $M_{{\mathcal S}_{a,b,c}}$, obtained by concatenating
the network chips for the $V_i$ matrices. For instance, in the case $N=3$, we get the following network:
\begin{equation}\label{netflat} \raisebox{-5.cm}{\hbox{\epsfxsize=16.cm \epsfbox{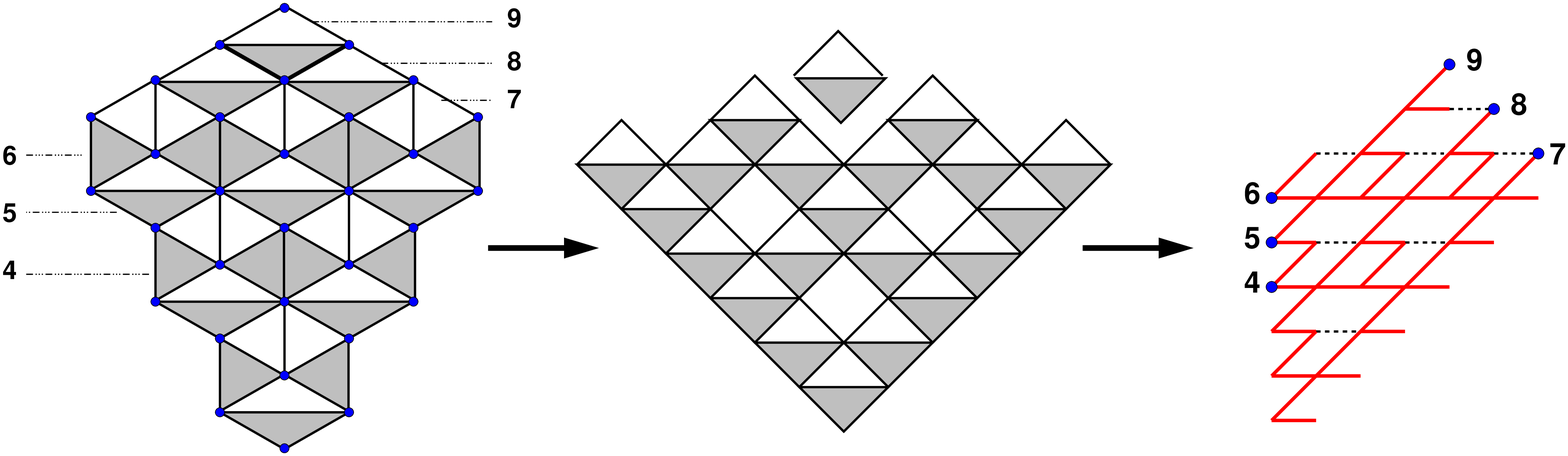}} } \end{equation}
where we have first straightened up the lozenge configuration, and then replaced each ($V$-type)
lozenge by its network chip \eqref{rep2UVi} (dashed lines correspond to weight 1 edges, 
whereas other edges receive the usual weights, determined by the surrounding face labels, which we 
have omitted here for simplicity).

In the augmented shadow ${\mathcal S}_{a,b,c}$, let us denote by $\sigma_1,\sigma_2,...,\sigma_{3N},\sigma_{3N+1}$
(resp. $\tau_1=\sigma_1,\tau_2,...,\tau_{3N},\tau_{3N+1}=\sigma_{3N+1}$) the assigned $\tau(x,y,z)$ initial
data at vertices along the West (resp. East) border of ${\mathcal S}_{a,b,c}$, read from bottom to top.
We are now ready for the main theorem of this section.

\begin{thm}\label{thetaexact}
The solution $\theta_{a,b,c}$ of the system \eqref{cuboc} is expressed in terms of its initial data \eqref{indatet}
as:
\begin{equation}\label{thetaformu}
\theta_{a,b,c} = \left( 
\prod_{i=N+2}^{2N}\sigma_i^{-1} \, \prod_{j=2N+1}^{3N}\tau_j \right) \, 
\left\vert M_{{\mathcal S}_{a,b,c}} \right\vert_{N+1,N+2,...,2N}^{2N+1,2N+2,...,3N}
\end{equation}
\end{thm}
\begin{proof}
The proof imitates that of Theorem \ref{solflat}, and is based on the identity between the matrices
$M_{{\mathcal S}_{a,b,c}}$ and $M_{{\mathcal C}_{a,b,c}}$ corresponding respectively to the augmented
shadow ${\mathcal S}_{a,b,c}$ and the surface of its maximal filling with cubes ${\mathcal C}_{a,b,c}$
with updated assigned vertex values. 
Indeed, both matrices are related via a finite number of applications of Lemma\ref{ybe}, and are therefore equal.
The surface ${\mathcal C}_{a,b,c}$
corresponds to the following lozenge decomposition (shown here for $N=3$):
\begin{equation}\label{maxmat} \hbox{\epsfxsize=6.cm \epsfbox{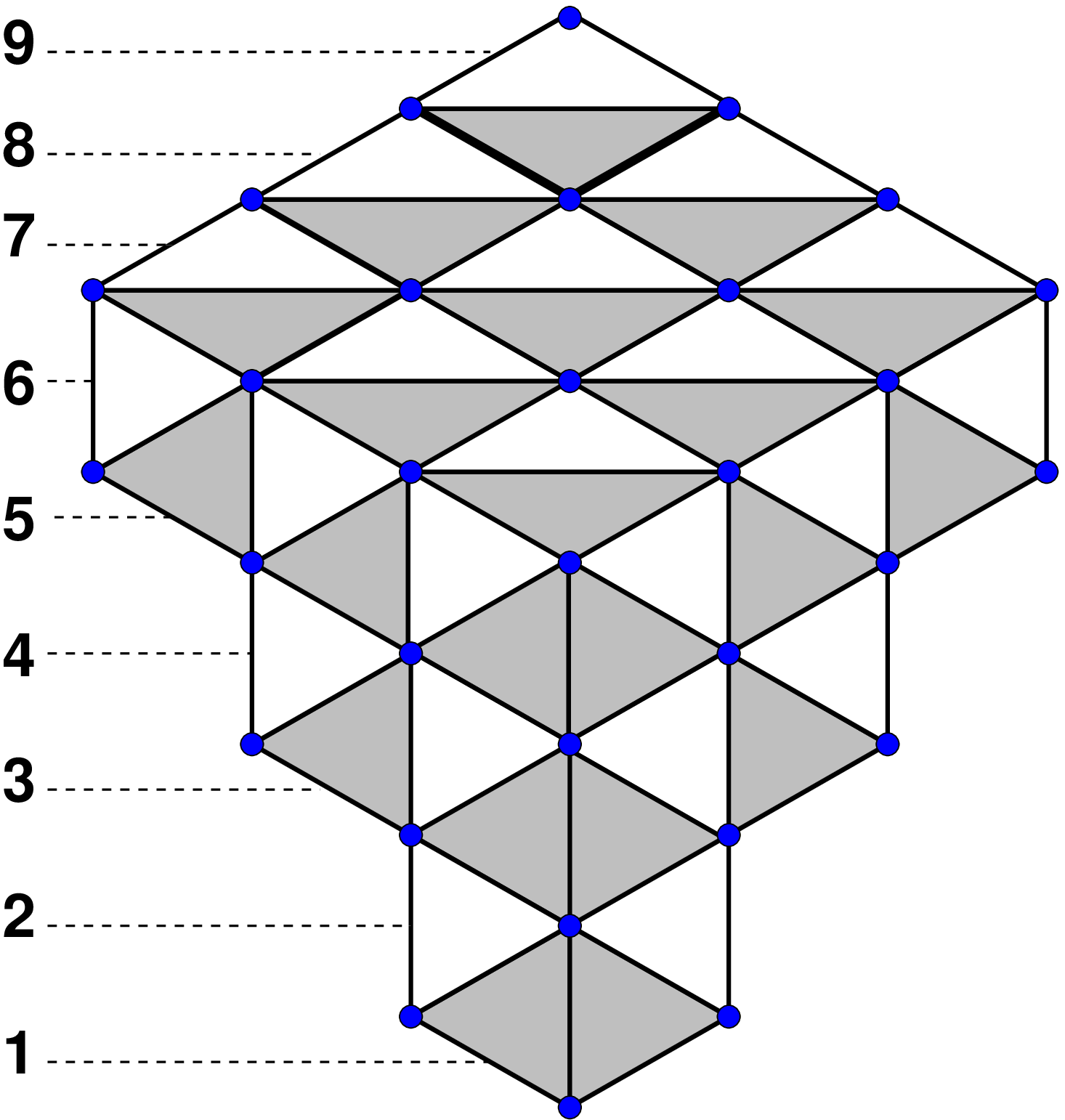}}  \end{equation}
in which all the new vertex data $\tau'(x,y,z)$ are the result of an iterative update at each cube addition. In particular, 
the central value is $\tau'(a,b,c)=\theta_{a,b,c}$ as the vertex $(a,b,c)$ is the apex of the cube corner in
${\mathcal C}_{a,b,c}$. Moreover, the boundary values $\sigma_i,\tau_j$ remain unchanged as they do not undergo 
any update. We are left with the computation of \eqref{thetaformu} with ${\mathcal S}_{a,b,c}$ replaced by 
${\mathcal C}_{a,b,c}$. To do this, we note that there is also a natural network formulation for the 
matrix $M_{{\mathcal C}_{a,b,c}}$, obtained as usual by replacing each $V_i$ matrix by a $V$-type chip.
For $N=3$, the correspondence reads:
$$ \raisebox{-5.cm}{\hbox{\epsfxsize=16.cm \epsfbox{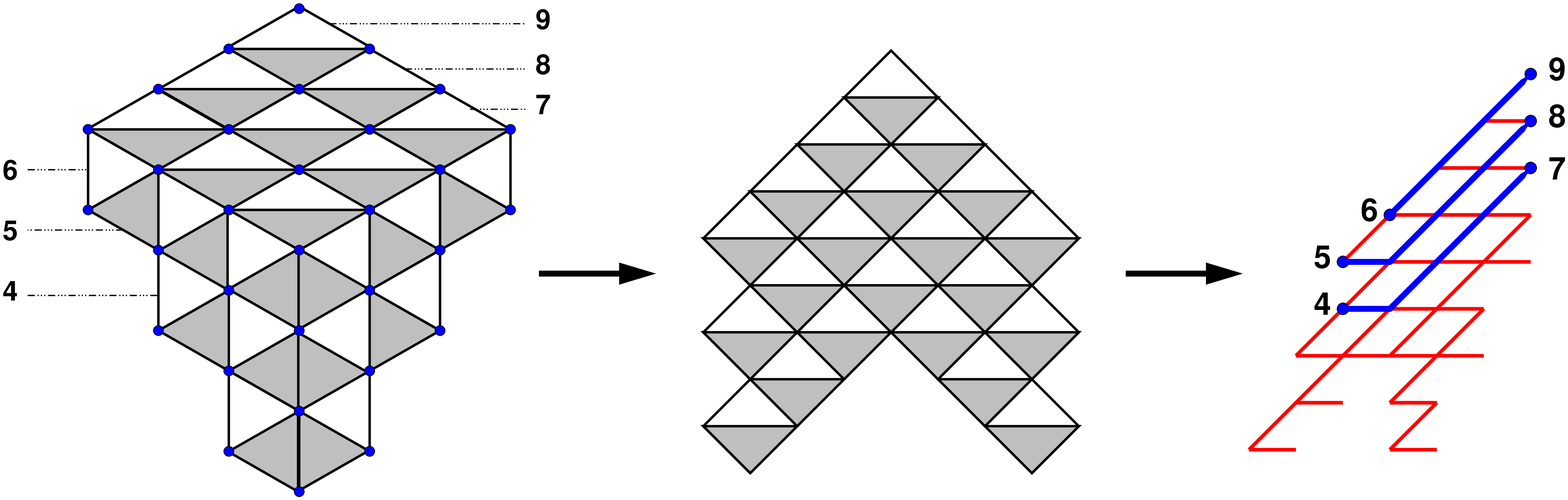}}} $$
The minor $\left\vert M_{{\mathcal C}_{a,b,c}} \right\vert_{N+1,N+2,...,2N}^{2N+1,2N+2,...,3N}$
is the partition function for non-intersecting paths starting at entry vertices $N+1,N+2,...,2N$ and ending at exit
vertices $2N+1,2N+2,...,3N$. There is only one such configuration (represented as thick lines above). Moreover,
the edge weights have the form of ratios, whose product is telescopic, leaving us with only boundary contributions,
which are cancelled out by the prefactor of the $\sigma$'s and $\tau$'s. However, one term remains due to the
imbalance between the two prefactors, and it is precisely the central value $\tau'(a,b,c)=\theta_{a,b,c}$. This gives:
$$\theta_{a,b,c}=\left( 
\prod_{i=N+2}^{2N}\sigma_i^{-1} \, \prod_{j=2N+1}^{3N}\tau_j \right) \, 
\left\vert M_{{\mathcal C}_{a,b,c}} \right\vert_{N+1,N+2,...,2N}^{2N+1,2N+2,...,3N}$$
and the theorem follows.
\end{proof}

\begin{cor}\label{arbitcub}
The formula \eqref{thetaformu} holds for ${\mathcal S}_{a,b,c}$ replaced by any intermediate surface
obtained by arbitrary cube additions, so it gives access to arbitrary initial data in this setting as well.
\end{cor}

\subsection{Dimer formulation}

We are now ready to give the dimer formulation of Theorem \ref{thetaexact}. As usual, we start from the dual
of the lozenge decomposition of the flat initial data surface, which looks as follows for $N=7$:
$$ \raisebox{-5.cm}{\hbox{\epsfxsize=14.cm \epsfbox{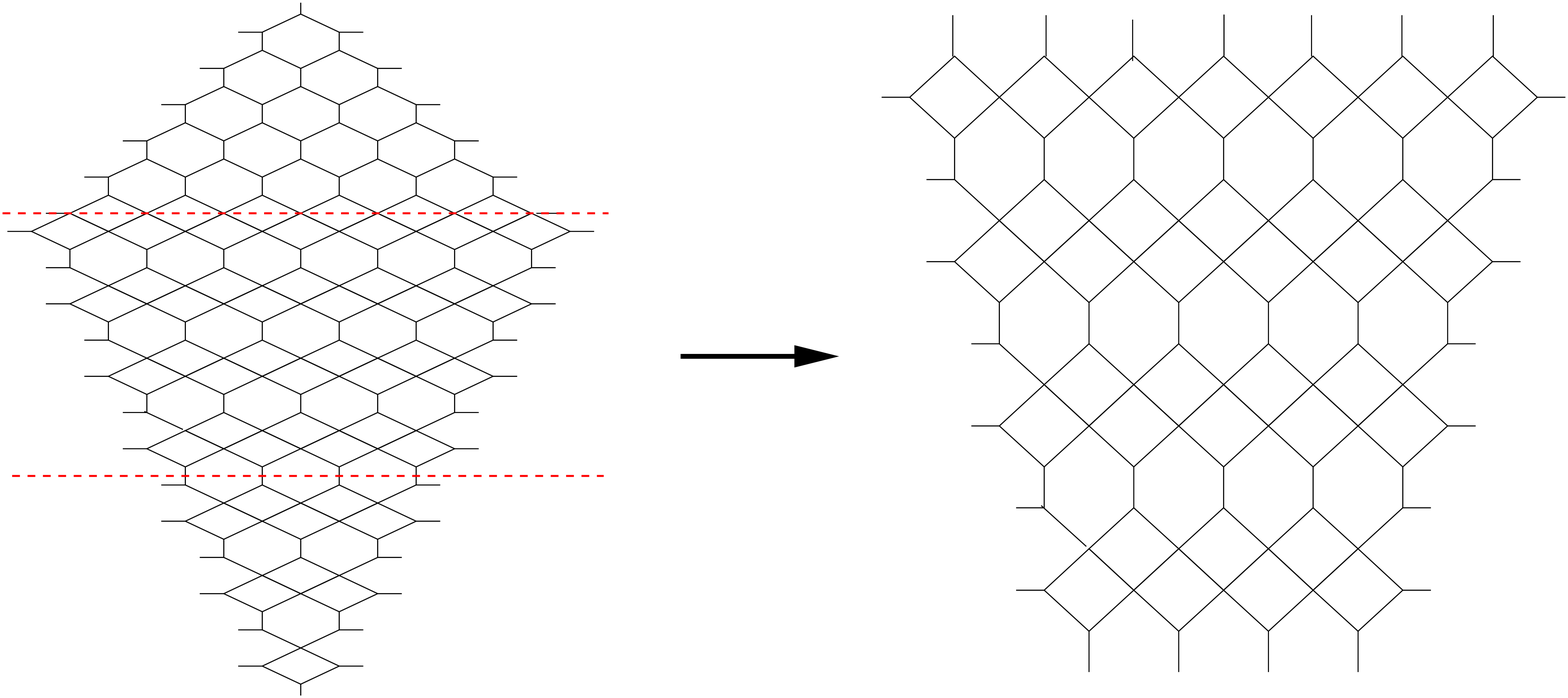}}} $$
We have moreover indicated by dashed lines the only portion of this graph
relevant to the expression of Theorem \ref{thetaformu}. We call $\Gamma_N$ this cut graph 
($\Gamma_7$ is represented on the right).
It is the only part contributing 
to the network paths, up to some overall scaling factor, as is apparent from the general structure below:
$$ \raisebox{-5.cm}{\hbox{\epsfxsize=6.cm \epsfbox{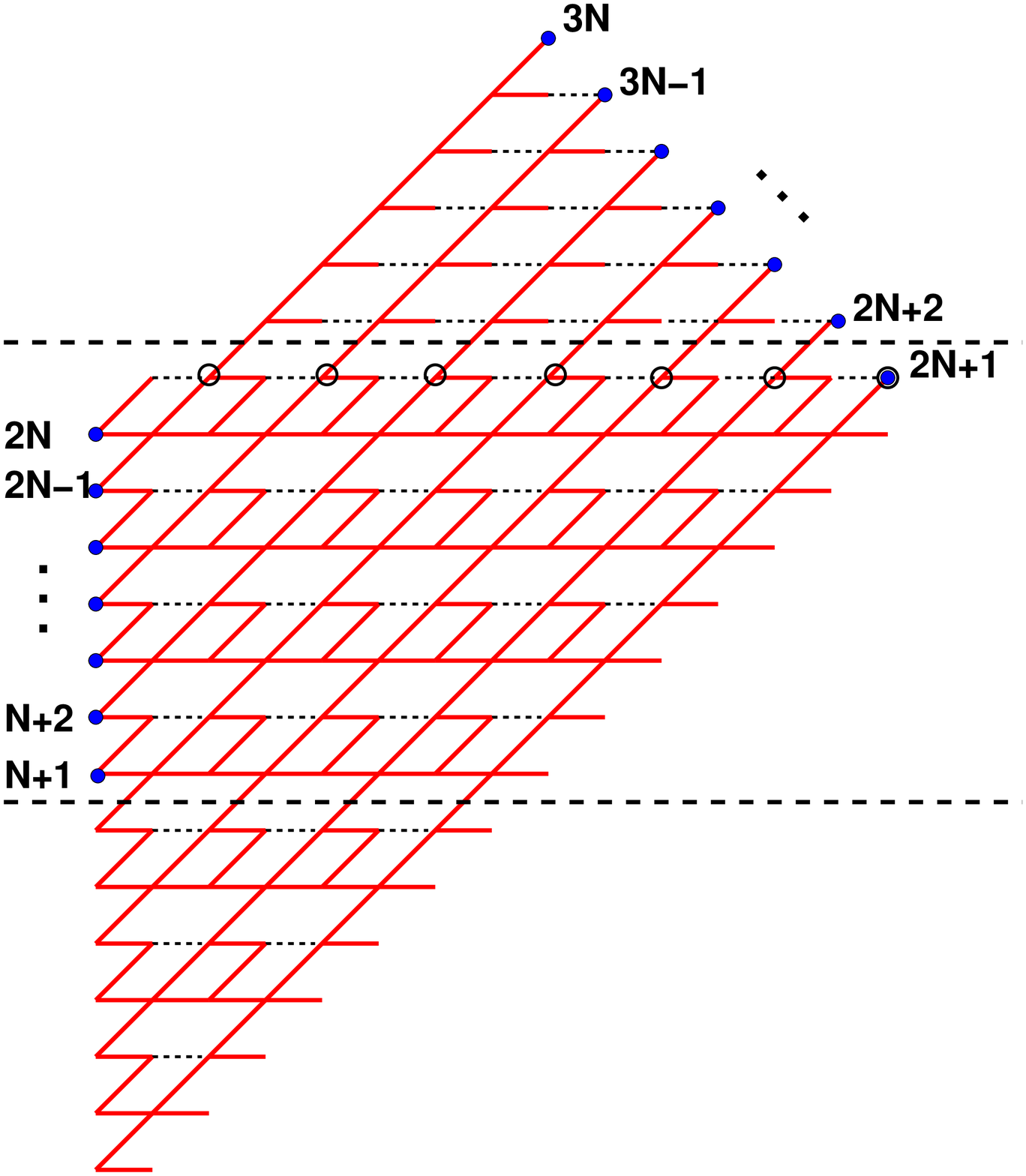}}}$$
where we see that the last portions of the $N$ non-intersecting paths on the lattice are straight lines joining
the empty circles to the endpoints $2N+1$, $2N+2$, ...,$3N$, and these contribute a rescaling of the overall factor,
as the weights are cancelled in telescopic products along the diagonal edges, and
moreover as paths only go horizontally or diagonally up, we can eventually restrict the network
to the portion in-between the two dashed lines.

Note that $\Gamma_N$ has $N$ vertical external edges on the top, $N$ horizontal 
ones on the West and East borders and that the top structure is the
same for all graphs, whereas the bottom structure depends on the parity of $N$ (for even $N$ the graph
has $N/2$ hexagons on the bottom, for odd $N$ it has $(N+1)/2$ squares, as shown here for $N=7$). 
Also, the faces of the graph inherit the vertex labels of the dual lozenge decomposition, which are nothing 
but initial data assignments $\tau(x,y)$.

\begin{figure}
\centering
\includegraphics[width=14.cm]{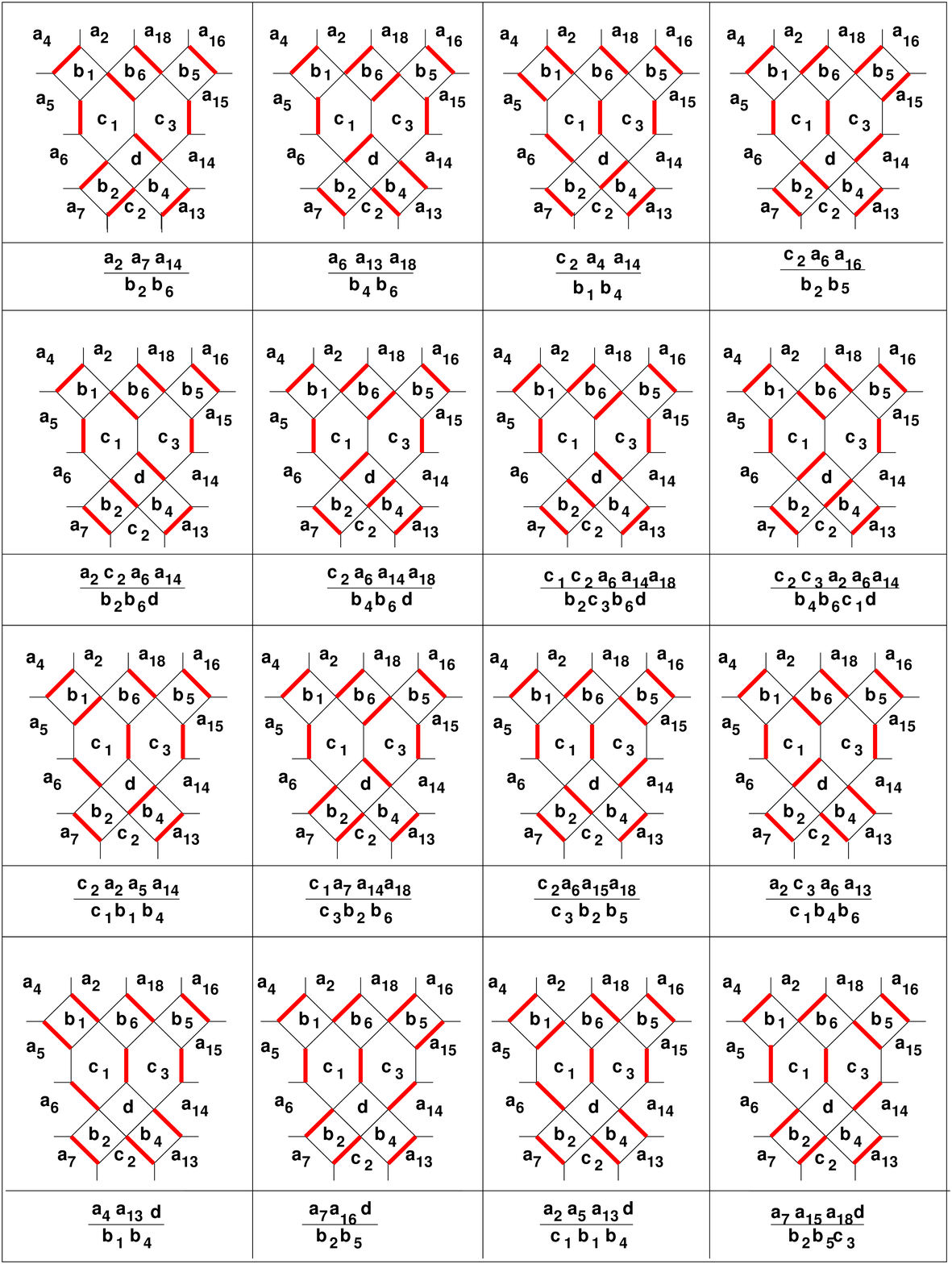}
\caption{The 16 dimer configurations for $N=3$ together with their respective weights.}\label{fig:sixteen}
\end{figure}

As before, we consider the dimer model on the graph $\Gamma_N$ with both internal and external face labels.
Each external edge (i.e. with a terminal vertex) remains empty. 
As before, an external face with label $a$ receives a weight $a^{1-\delta}$ if the 
adjacent edges are occupied by $\delta$ dimers, while an internal face with label $b$ receives a weight $b^{1-D}$
(square) or $b^{2-D}$ (hexagons), where $D$ is the total number of dimers adjacent to the face. The partition function
is the sum over dimer configurations of the product of internal and external face weights.
We are now ready for the final theorem of this section.

\begin{thm}\label{dimcub}
The solution $\theta_{a,b,c}$ of the system \eqref{cuboc} is expressed in terms of its initial data \eqref{indatet}
as the partition function for dimers on the graph $\Gamma_N$.
\end{thm}
\begin{proof}
The proof is exactly along the same lines as that of Theorem \ref{gendimth}, only simpler as we only have
to deal with $V$-type network path configurations and their associated dimer configurations. The prefactors are
again nicely cancelled out by the boundary face weights of the dimer model.
\end{proof}

The Theorem is illustrated in the case $N=3$ in Fig.\ref{fig:sixteen}, where we list the 16 contributions to say $\theta_{1,2,2}$
in terms of the initial data assignments shown in the dual picture.

\subsection{Comparison with the previous solution}

The geometry of the cube corner evaporation (with coordinates $(a,b,c)$ in the 3D cubic lattice $\Z^3$)
can be embedded into the original 3D CC lattice (with coordinates $(i,j,k)$ in $\Z^3$ such that $i+j+k=1$ mod 2), 
for some specific choices of initial data stepped surfaces, and directions of mutation. The flat cubic initial data
stepped surface is made of ``hexagons" with a bottom vertex at time $k=0$, three vertices at time $k=1$
and three at time $2$. Let us embed one of these hexagons say with bottom-most vertex $(a,b,c)\mapsto (i,j,k)$
into the CC lattice as follows: 
$$  \raisebox{0.cm}{\hbox{\epsfxsize=10.cm \epsfbox{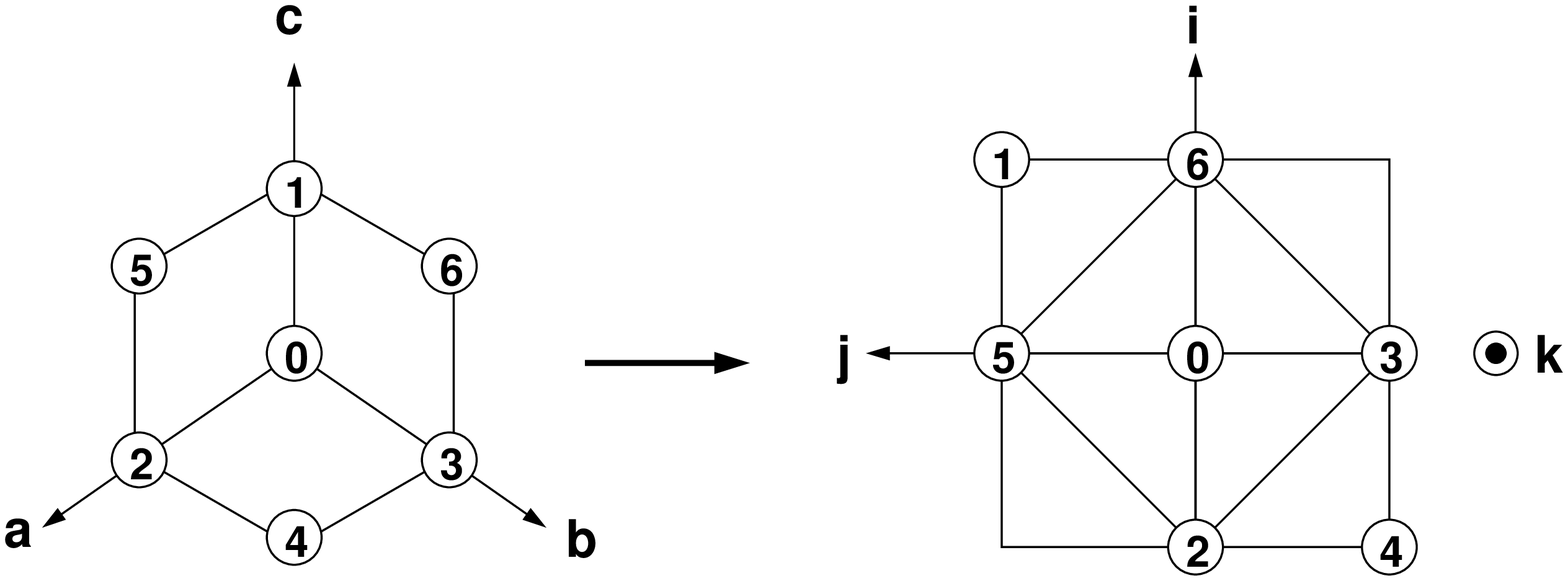}}}   $$
where we labeled with the same integer the points and their images:
$$\begin{matrix} 0 &=& (a,b,c) & \mapsto & (i,j,k) \\
1 &=& (a,b,c+1) & \mapsto & (i+1,j+1,k) \\
2 &=& (a+1,b,c) & \mapsto & (i-1,j,k+1) \\
3 &=& (a,b+1,c) & \mapsto & (i,j-1,k+1) \end{matrix}\quad \begin{matrix}
4 &=& (a+1,b+1,c) & \mapsto & (i-1,j-1,k+2) \\
5 &=& (a+1,b,c+1) & \mapsto & (i,j+1,k+1) \\
6 &=& (a,b+1,c+1) & \mapsto & (i+1,j,k+1) \\
\end{matrix}$$
Using this, we may represent in the CC lattice the cube corner's augmented shadow ${\mathcal S}_{a,b,c}$ of $(a,b,c)$
onto the flat initial data surface (represented here for $N=3$)
as well as any of its mutations via unit cube deposition as follows:
$$  \raisebox{0.cm}{\hbox{\epsfxsize=16.cm \epsfbox{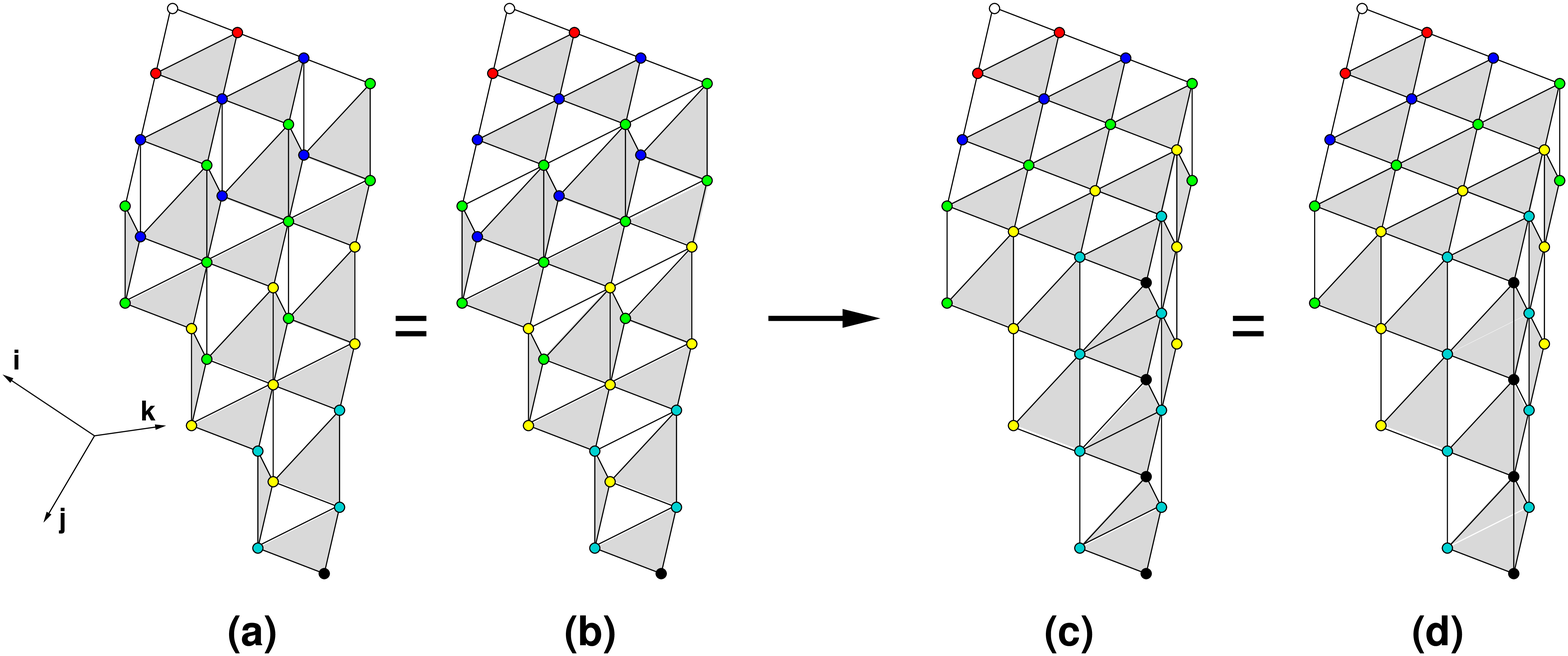}}}  $$
where we show the CC lattice embedding of 
(a) the embedded augmented shadow (b) its ``mutation-ready" version obtained by flipping the top
vertical edges of white ``squares" in the hexagons (c) its fully mutated image into the cube corner (d) the embedded actual cube corner after
flipping diagonals of gray squares.

Using this map, we may therefore recover the general solution of the cube corner evaporation from that of the CC $2+1$D evolution.
This shows that we may consider particular families of initial data stepped surfaces in the CC model, and have them evolve
in particular directions, and still get some systematic reasonably simple expressions for the solution as either network path
partition functions or dimer partition functions on particular families of graphs. Related questions on Gayle-Robinson sequences,
brane tilings and the enumeration of beehives were considered in \cite{MUSIK}.

\section{Discussion and conclusion}

In this paper, we have explored solutions of the $T$-system recurrence in various geometries of
initial data. In all cases, we have been able to formulate the equation as a flatness condition for
a $GL_n$ connection expressed as a product of embedded $U,V$-type $2\times 2$ matrices,
and to use this formulation to derive a compact formula for the solution in terms of the initial data.
The $U,V$ matrices can in turn be interpreted as network chips, and their products as networks,
i.e. oriented graphs with weighted edges, supporting path configurations. Finally, we showed how
the latter configurations can be reinterpreted in terms of dimer configurations on some associated 
bipartite graph. So the building blocks $U,V$ of our solution can be interpreted as local transfer operators
for the dimer model.

We may wonder whether more dimer or dimer-related models could be described by $U,V$ matrices.

%\subsection{U,V matrices and higher rank cases}
%
%\subsection{U,V matrices and the cube recurrence}
%
%\subsection{U,V matrices and the hexahedron equation}

Recently a system of recursion relations called the hexahedron relation \cite{KEN}, generalizing 
earlier work by Kashaev \cite{KASH} appeared in relation to
the so-called $Y-\Delta$ relation of the Ising model, itself expressed in the framework of urban renewal of dimer
graphs and cluster algebra mutations. Here we show that this system may be 
obtained as a generalized Yang-Baxter relation obeyed by $U,V$ matrices. This is encapsulated in
the following lemma, easily proved by direct calculation.

\begin{lemma}\label{hexayb}
Let $R_i(x,a,b,c,d)=V_i(c,b,x)U_i(x,d,a)$ and $S_i(x,a,b,c,d)=U_i(b,x,a)V_i(c,x,d)$ in any $GL_n$ embedding. 
The following identity holds:
\begin{eqnarray*}&&R_1(b_2,a_2,a_3,a_4,c)S_2(b_1,a_1,a_2,c,a_6) R_1(b_3,a_6,c,a_4,a_5)\\
&&\qquad \qquad \qquad  =
S_2(b_3',a_1,a_2,a_3,c') R_1(b_1',c',a_3,a_4,a_5) S_2(b_2',a1,c',a_5,a_6) 
\end{eqnarray*}
if and only if the variables satisfy the following system of algebraic equations:
\begin{eqnarray*}
c b_1b_1'&=&a_2 a_4 a_6 + b_1 b_2 b_3 + c a_1 a_4 \\
c b_2b_2'&=&a_2 a_4 a_6 + b_1 b_2 b_3 + c a_3 a_6 \\
c b_3b_3'&=&a_2 a_4 a_6 + b_1 b_2 b_3 + c a_2 a_5 \\
b_1 b_2 b_3 c^2 c'&=& (a_2 a_4 a_6+b_1 b_2 b_3)^2 +  (a_2 a_4 a_6+b_1 b_2 b_3)(a_1a_4+a_3a_6+a_2a_5)c \\
&&+ (a_1 a_2 a_4 a_5+ a_1 a_3 a_4 a_6  +  a_2 a_3 a_5 a_6 )c^2 + a_1 a_3 a_5 c^3
\end{eqnarray*} 
\begin{eqnarray*}
\end{eqnarray*}
\end{lemma}

\begin{remark}
The relation of Lemma \ref{hexayb} is a generalization of the Yang-Baxter relation for so-called dimerized 
quantum spin chains or ladders \cite{DFD}\cite{MN}, in which the structure of the $R$-matrix ($R$ or $S$ here) 
is staggered, i.e. alternates between neighboring sites.
\end{remark}

\begin{remark}\label{pictorem}
The pictorial interpretation of Lemma \ref{hexayb} is the following
identification between two lozenge decompositions of a hexagon:
$$ \raisebox{0.cm}{\hbox{\epsfxsize=10.cm \epsfbox{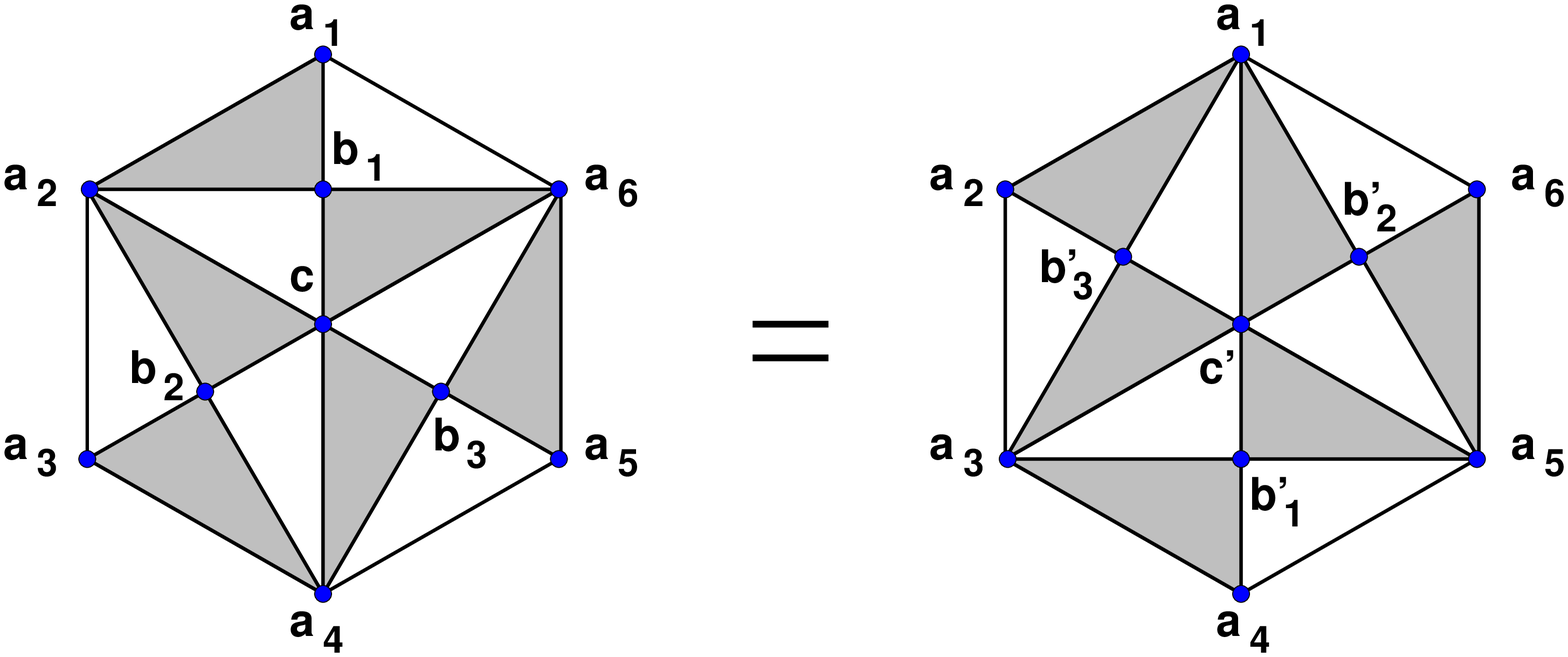}}} $$
which may be interpreted as a cube evaporation, similar to that of the previous section. This is the point of view 
adopted in \cite{KEN}, where $c$ is at the cube top vertex, and $b_1,b_2,b_3$ at the centers of the three top faces.
\end{remark}

\begin{remark}
The rational transformation $(b_1,b_2,b_3,c)\mapsto (b_1',b_2',b_3',c')$ is easily invertible, by noting that the 
up/down reflection of the pictures of Remark \ref{pictorem}, gives the same identity, upon
switching $a_1\leftrightarrow a_4, a_2\leftrightarrow a_3,a_5\leftrightarrow a_6, b_1\leftrightarrow b_1',c\leftrightarrow c',
b_2\leftrightarrow b_3',b_3\leftrightarrow b_2'$. This is a reflection of the fact that cluster algebra mutations are involutions.
\end{remark}

Lemma \ref{hexayb} allows to define as before a $GL_{3N}$ connection invariant under any cube addition/evaporation
above the flat initial data (say for a cube corner apex at position $(a,b,c)$ with $a+b+c=N+2$). Any formula involving
this connection that produces the apex value in the case of the full corner surface would in particular yield
the expression of the latter in terms of any intermediate initial data, as in the previous section.
On the other hand, the expression for the value at the apex of the cube corner in terms of flat initial data was given in \cite{KEN},
as the partition function for double-dimers with some particular boundary condition. It would be interesting to
relate this to our $U,V$ matrix connections, which we earlier interpreted as dimer model transfer matrices.

It would also be interesting to extend this type of analysis to the cube equation \cite{CS} which is known to be related to compound
mutations of a restricted cluster algebra, and whose solution has a combinatorial description in terms of groves \cite{CS}. 

Finally, a quantum version of the $GL_2$ connection used in the present paper was introduced in \cite{DFK12} to solve the
quantum $A_1$ $T$-system, a $q$-commuting version of the $A_1$ $T$-system suggested by the natural quantum deformation
of the associated cluster algebra, as defined in \cite{BZ} for finite rank. A simple adaptation of this construction should provide
us with quantum versions of the various recurrences studied above.

%Another interesting direction is the quantum deformation. 

%dual cube corner graph? 
%dual flat surface graph?
%double-dimer model vs single dimer model?
%
%
%composites: higher rank, 
%cube?

%\subsection{Limit shapes}

%ellipses

%higher genus case

%\section{Discussion and Conclusion}

%\begin{figure}
%\centering
%\includegraphics[width=12.cm]{network.eps}
%\caption{Two interpretations for the network matrix.}\label{fig:network}
%\end{figure}

%\begin{example}
%\end{example}

\end{document}